%% file: Starek.Barbee.Pavone.JGCD15.tex
\documentclass[notitlepage,nofootinbib]{AIAA}	

\input{preamble_confpaper}
\usepackage{tikz}
\usepackage{thm-restate}
\usepackage[noprefix]{nomencl}

\makenomenclature

\makeatletter
\newdimen\nomwidest
\renewcommand{\nomlabel}[1]{%
	\sbox\z@{#1\hspace{\labelsep}$=$}%
	\ifdim\nomwidest<\wd\z@\global\nomwidest\wd\z@\fi
	#1\hfil\hspace{\labelsep}$=$%
}
\renewcommand{\nompostamble}{%
	\protected@write\@auxout{}{\global\nomwidest=\the\nomwidest}%
}
\makeatother

\togglefalse{AddressCountry}		
\togglefalse{AddressBuilding}		
\togglefalse{AddressZipExtension}	
\renewcommand{\DepartmentOf}{}		

\makeatletter
\@ifclassloaded{article}{%
	\newcommand{\frontmatter}{%
		\title{\large \textbf{\uppercase{Fast, Safe, and \Fuel-Efficient \\ Spacecraft Planning under Clohessy-Wiltshire-Hill Dynamics}}}
		\author{\large\textbf{%
			Joseph A.\ Starek%
			\thanks{Graduate Student, \StanfordAeroAstro, \Stanford, \StanfordAeroAstroAddress[Rm 009], \texttt{jstarek@stanford.edu}},
			Edward Schmerling%
			\thanks{Graduate Student, \StanfordICME, \Stanford, \StanfordICMEAddress[], \texttt{schmrlng@stanford.edu}},
			Gabriel D.\ Maher%
			\thanks{Graduate Student, \StanfordICME, \Stanford, \StanfordICMEAddress[], \texttt{gdmaher@stanford.edu}},
			Brent W.\ Barbee%
			\thanks{Aerospace Engineer, Navigation and Mission Design Branch (Mail Code 595), \NASAGSFC, \NASAGSFCAddress[Bldg 11, C002H]}, \texttt{brent.w.barbee@nasa.gov},
			and Marco Pavone%
			\thanks{Assistant Professor, \StanfordAeroAstro, \Stanford, \StanfordAeroAstroAddress[Rm 261], \texttt{pavone@stanford.edu}}
			}
		}
		\date{}
	}
}{}
\@ifclassloaded{aiaa-tc}{%
	\newcommand{\frontmatter}{%
		\title{Fast, Safe, and \Fuel-Efficient Spacecraft Planning under Clohessy-Wiltshire-Hill Dynamics}
		\author{%
			Joseph A.\ Starek\thanks{Graduate Student, \StanfordAeroAstro, \StanfordAeroAstroAddress[Rm 009]}, Edward Schmerling\thanks{Graduate Student, \StanfordICME, \StanfordICMEAddress[]}, and Gabriel D.\ Maher\thanks{Graduate Student, \StanfordICME, \StanfordICMEAddress[]} \\
			\emph{\Stanford, \StanfordAddress*}
			\and
			Brent W.\ Barbee\thanks{Aerospace Engineer, Navigation and Mission Design Branch  (Mail Code 595), \NASAGSFCAddress[Bldg 11, C002H], AIAA Senior Member} \\
			\emph{\NASAGSFC, \NASAGSFCAddress*}
			\and
			Marco Pavone\thanks{Assistant Professor, \StanfordAeroAstro, \StanfordAeroAstroAddress[Rm 261], AIAA Member} \\
			\emph{\Stanford, \StanfordAddress*}
		}
	}
}{}
\@ifclassloaded{AIAA}{%
	\newcommand{\frontmatter}{%
		\title{Fast, Safe, and \Fuel-Efficient Spacecraft\\\vspace*{-0.09cm} Planning under Clohessy-Wiltshire-Hill Dynamics}
		\author{Joseph A.\ Starek,\footnote{Graduate Student, \StanfordAeroAstro, \StanfordAeroAstroAddress[Rm 009]} Edward Schmerling,\footnote{Graduate Student, \StanfordICME, \StanfordICMEAddress[]} and Gabriel D.\ Maher\footnotemark[2]}
		\affiliation{\Stanford, \StanfordAddress*}
		\author{Brent W.\ Barbee\footnote{Aerospace Engineer, Navigation and Mission Design Branch (Mail Code 595), \NASAGSFCAddress[], AIAA Senior Member}}
		\affiliation{\NASAGSFC, \NASAGSFCAddress*}
		\author{Marco Pavone\footnote{Assistant Professor, \StanfordAeroAstro, \StanfordAeroAstroAddress[Rm 261], AIAA Member}}
		\affiliation{\Stanford, \StanfordAddress*\phantom{}} 
	}
}{}
\@ifclassloaded{ieeeconf}{%
	\newcommand{\frontmatter}{%
		\title{\large \textbf{\uppercase{Fast, Safe, and \Fuel-Efficient Spacecraft \\ Planning under Clohessy-Wiltshire-Hill Dynamics}}}
		\author{%
			Joseph A.\ Starek*\thanks{*Graduate Student, \StanfordAeroAstro, \Stanford, \texttt{jstarek@stanford.edu}},
			Edward Schmerling$\ddagger$\thanks{$\ddagger$Graduate Student, \StanfordICME, \Stanford, \texttt{schmrlng@stanford.edu}},
			Gabriel D.\ Maher$\dagger$\thanks{$\dagger$Graduate Student, \StanfordICME, \Stanford, \texttt{gdmaher@stanford.edu}},
			Brent W.\ Barbee$\mathsection$\thanks{$\mathsection$Aerospace Engineer, \NASAGSFC, NASA, \texttt{brent.w.barbee@nasa.gov}},
			Marco Pavone$\mathparagraph$\thanks{$\mathparagraph$Assistant Professor, \StanfordAeroAstro, \Stanford, \texttt{pavone@stanford.edu}}
		}
	}
}{}
\makeatother

\allowdisplaybreaks     	

\providecommand{\bibfiles}{../../../bib/main}
\newcommand{\figurepath}{../../../../Figures}
\graphicspath{%
	{./figures/}%
	{../../2016/Starek.Pavone.AeroConf16/figures/}
	{../../2016/Starek.Pavone.AeroConf16/figures/plots/Exact Waypoint Convergence/}%
	{../../2016/Starek.Pavone.AeroConf16/figures/plots/Inexact Waypoint Convergence/}%
	{../Starek.Barbee.Pavone.AASGNC15/figures/}%
	{\figurepath/Control/}%
	{\figurepath/Dynamics}%
	{\figurepath/Planning/Concept/}%
	{\figurepath/Planning/Collision Detection/}%
	{\figurepath/Planning/Safety/}%
	{\figurepath/Simulations/2014/CWH FMTstar/}%
	{\figurepath/Spacecraft/}%
}

\newcommand{\Fuel}{Propellant\xspace} 
\newcommand{\fuel}{propellant\xspace} 
\newcommand{\fuelCost}{propellant-cost\xspace}

\newcommand{\posCrossTrack}{\delta x}
\newcommand{\posInTrack}{\delta y}
\newcommand{\posOutofPlane}{\delta z}
\newcommand{\velCrossTrack}{\delta \dot{x}}
\newcommand{\velInTrack}{\delta \dot{y}}
\newcommand{\velOutofPlane}{\delta \dot{z}}
\newcommand{\VecDeltaVx}{\VecDeltaVMag_{x,}{}}
\newcommand{\VecDeltaVy}{\VecDeltaVMag_{y,}{}}
\newcommand{\VecDeltaVz}{\VecDeltaVMag_{z,}{}}

\newcommand{\ThrustAvailable}{\eta_i{}}
\newcommand{\BodyFixedFrame}{\mathrm{bf}}
\newcommand{\AllocatedThrust}{\DeltaV_i{}}
\newcommand{\VecAllocatedThrust}{\vecformat{\DeltaV}_i{}}
\newcommand{\VecThrustPos}{\Vecrho_i{\vphantom{\Vecrho}}}
\newcommand{\VecThrustDir}{\vecformat{\Delta \hat{v}}_i{}}
\newcommand{\AllocatedThrustMax}{\DeltaV\submax{}_{,}{}}
\newcommand{\AllocatedThrustMin}{\DeltaV\submin{}_{,}{}}

\definevector{Vecb}{b}
\definevector{VecF}{F}
\definevector{Vecq}{q}
\definevector{Vecr}{r}
\definevector{Vecs}{s}
\definevector{Vecu}{u}
\definevector{Vecv}{v}
\definevector{Vecw}{w}
\definevector{Vecx}{x}
\definevector{Vecy}{y}
\definevector{Vecz}{z}
\definevector{VecDeltaV}{\Delta v}
\definevector{Vecrho}{\rho}
\definevector{VecThrust}{\DeltaV}
\definevector{VecMoment}{M}

\definematrix{Gramian}{G}
\definematrix{MatA}{A}
\definematrix{MatB}{B}
\definematrix{MatEllipsoid}{E}
\definematrix{MatDeltaV}{\Delta V}
\definematrix{MatDeltaVSlicing}{S}
\definematrix{MatPhi}{\Phi}

\defineset{EllipsoidalSet}{E}
\defineset{ReachableSet}{R}
\defineset{SampleSet}{S}
\defineset{Tree}{T}
\defineset{TimeHorizon}{T}
\defineset{Frontier}{V_{\mathrm{open}}}
\defineset{TreeInterior}{V_{\mathrm{closed}}}
\defineset{Unexplored}{V_{\mathrm{unvisited}}}

\newcommand{\CostThreshold}{\bar{\CostFcn}}
\newcommand{\Dimension}{d}
\newcommand{\EqualityConstraintVec}{\mathbf{h}}
\newcommand{\GoalPosTol}{\epsilon_r}
\newcommand{\GoalVelTol}{\epsilon_v}
\newcommand{\InequalityConstraint}{g}
\newcommand{\InequalityConstraintVec}{\mathbf{g}}

\renewcommand{\Nthrusters}{K}
\renewcommand{\nonnegativereals}{\reals_{\geq 0}}

\newcommand{\PathSet}{\Sigma}
\newcommand{\pvalue}{2}
\newcommand{\ManeuverDuration}{T}
\newcommand{\TotalManeuverDuration}{T_{\mathrm{plan,}}{}}
\newcommand{\NBurns}{N}
\newcommand{\BurnTime}{\tau}
\renewcommand{\Plume}{{\mathcal{P}}_i{}}
\newcommand{\PlumeHalfAngle}{\beta_{\mathrm{plume}}}
\newcommand{\PlumeHeight}{H_{\mathrm{plume}}}
\newcommand{\PlumeAxis}{-\VecThrustDir}
\newcommand{\PlumeAvoidanceTarget}{\spheres_{\mathrm{target}}}

\renewcommand{\TimeHorizon}{\mathcal{T}}
\newcommand{\horizonTime}{T_h}
\newcommand{\FaultIndex}{f}
\newcommand{\FaultTolerance}{F}
\newcommand{\tCAM}{t}
\newcommand{\XsubKOZ}{\Xspace\subKOZ}
\newcommand{\semiaxis}{\rho}
\newcommand{\Vcirc}{-\frac{3}{2} \omegaRef \posCrossTrack}

\newcommand{\GramianEigsUBConst}{M\submax}
\newcommand{\GramianEigsLBConst}{M\submin}

\newcommand{\ClosedBall}{{\mathcal{B}}}

\newcommand{\SubSpace}{{\Xspace}}
\newcommand{\DispersionFcn}{D}
\newcommand{\nSamples}{n}

\newcommand{\SmoothingWeight}{\alpha}
\newcommand{\SmoothingTolerance}{\delta\alpha\submin}
\newcommand{\subsmooth}{_{\mathrm{smooth}}}

\begin{document}

\frontmatter

\begin{abstract}
	This paper presents a sampling-based motion planning algorithm for real-time and \fuel-optimized auto\-nomous spacecraft trajectory generation in near-circular orbits.  Specifically, this paper leverages recent algorithmic advances in the field of robot motion planning to the problem of impulsively-actuated, \fuel-optimized rendezvous and proximity operations under the Clohessy-Wiltshire-Hill (CWH) dynamics model.  The approach calls upon a modified version of the Fast Marching Tree (\FMTstar) algorithm to grow a set of feasible trajectories over a deterministic, low-dispersion set of sample points covering the free state space.  To enforce safety, the tree is only grown over the subset of \emph{actively-safe samples}, from which there exists a feasible one-burn collision avoidance maneuver that can safely circularize the spacecraft orbit along its coasting arc under a given set of potential thruster failures.  Key features of the proposed algorithm include: (i) theoretical guarantees in terms of trajectory safety and performance, (ii) amenability to real-time implementation, and (iii) generality, in the sense that a large class of constraints can be handled directly.  As a result, the proposed algorithm offers the potential for widespread application, ranging from on-orbit satellite servicing to orbital debris removal and autonomous inspection missions.  
\end{abstract}

\maketitle


\section{Introduction}
\label{sec: Intro}

Real-time autonomy for spacecraft proximity operations near circular orbits is an inherently challenging task, particularly for onboard implementation where computational capabilities are limited.  Many effective real-time solutions have been developed for the \emph{unconstrained} case (\eg state transition matrix manipulation \cite{DPD:11}, glideslope methods \cite{HBH-MLT-DJDB:02}, safety ellipses \cite{BJN:05, DEG-BWB:07}, and others \cite{MD-YX-KP-GC:11}).
However, the difficulty of real-time processing increases when there is a need to operate near other objects and/or incorporate some notion of \fuel-optimality or control-effort minimization.  In such cases, care is needed to efficiently handle collision-avoidance, plume impingement, sensor line-of-sight, and other complex guidance constraints, which are often non-convex and may depend on time and a mixture of state and control variables.  State-of-the-art techniques for collision-free spacecraft proximity operations (both with and without optimality guarantees) include artificial potential function guidance \cite{IL-CRMI:95, JDM-NGFC:10}, convexification techniques \cite{MWH-BA:14}, enforcing line-of-sight or approach corridor constraints \cite{LB-JPH:08, RV-FG-EFC:11, PL-XL:13, XL-PL:14}, maintaining relative separation \cite{YZL-LBL-HW-GJT:11}, satisfying Keep-Out-Zone (KOZ) constraints using mixed-integer (MI) programming \cite{AR-TS-JH-EF:02}, and kinodynamic motion planning algorithms \cite{SML-JJK:01, EF:03, JMP-LEK-NB:03, MK-SP:14}.  

Requiring hard assurances of mission safety with respect to a wide variety and number of potential failure modes \cite{MH-GD:14} provides an additional challenge.
Often the concept of \emph{passive safety} (safety certifications on zero-control-effort failure trajectories) over a finite horizon is employed to account for the possibility of control failures, though this potentially neglects mission-saving opportunities and fails to certify safety for all time.  A less-conservative alternative that more readily adapts to infinite horizons, as we will see, is to use \emph{active safety} in the form of positively-invariant set constraints.
For instance, \cite{LB-JPH:08} enforces infinite-horizon active safety for a spacecraft by requiring each terminal state to lie on a collision-free orbit of equal period to the target.  \cite{EF:03} achieves a similar effect by only planning between waypoints that lie on circular orbits (a more restrictive constraint).  Similarly, \cite{AW-MB-CP-RSE-IVK:13} requires that an autonomous spacecraft maintain access to at least one safe forced equilibrium point nearby.  Finally, \cite{JMC-BA-RMM-DGM:08} devises the Safe and Robust Model Predictive Control (MPC) algorithm, which employs invariant feedback tubes about a nominal trajectory (which guarantee resolvability) together with positively-invariant sets (taken about reference safety states) designed to be available at all times over the planning horizon.

The objective of this paper is to design an automated approach to actively-safe spacecraft trajectory optimization for rendezvous and proximity operations near circular orbits, which we model using Clohessy-Wiltshire-Hill (CWH) dynamics.  Our approach is to leverage recent advances from the field of robot motion planning, in particular from the area of \emph{sampling-based} motion planning \cite{SL:06}.  Several decades of research in the robotics community have shown that sampling-based planning algorithms (dubbed ``planners'' throughout this paper) show promise for tightly-constrained, high-dimensional optimal control problems such as the one considered in this paper.  Sampling-based motion planning essentially entails the breakdown of a complex trajectory control problem into a series of many relaxed, simpler Two-Point Boundary Value Problems (2PBVPs, or ``steering'' problems) that are subsequently evaluated \emph{a posteriori} for constraint satisfaction and efficiently strung together into a graph (\ie a tree or roadmap).  In this way, complex constraints like obstacle avoidance or plume impingement are decoupled from the generation of subtrajectories (or graph ``edges'') between graph states (or ``samples''), separating dynamic simulation from constraint checking -- a fact we exploit to achieve real-time capability.  
Critically, this approach avoids the \emph{explicit} construction of the free state space, which is prohibitive in complex planning problems.  As a result, sampling-based algorithms can address a large variety of constraints and can provide significant computational benefits with respect to traditional optimal control methods and mixed-integer programming \cite{SL:06}.  Furthermore, through a property called \emph{asymptotic optimality} (AO), sampling-based algorithms can be designed to provide guarantees of optimality in the limit that the number of samples taken approaches infinity.
This makes sampling-based planners a strong choice for the problem of spacecraft control.

Though the aforementioned works \cite{SML-JJK:01, EF:03, JMP-LEK-NB:03, MK-SP:14} on sampling-based planning for spacecraft proximity operations have addressed several components of the safety-constrained, optimal CWH autonomous rendezvous problem, few have addressed the aspect of real-time implementability in conjunction with both a 2-norm \fuelCost metric and active trajectory safety with respect to control failures.  This paper seeks to fill this gap. 
The paper's central theme is a rigorous proof of asymptotic optimality for a particular sampling-based planner, namely a modified version of the Fast Marching Tree (\FMTstar) algorithm \cite{LJ-ES-AC-ea:15}, under impulsive CWH spacecraft dynamics with hard safety constraints.  First, a description of the problem scenario is provided in \cref{sec: Problem}, along with a formal definition of the 2-norm cost metric that we employ as a proxy for \fuel consumption.  \Cref{sec: Safety} then follows with a thorough discussion of chaser/target vehicle safety, defining precisely how abort trajectories may be designed under CWH dynamics to deterministically avoid for all future times an ellipsoidal region about the CWH frame origin under a prescribed set of control failures.  Next, we proceed in \cref{sec: Approach} to our proposed approach employing the modified \FMTstar algorithm.  The section begins with presentation of a conservative approximation to the \fuelCost reachability set, which characterizes the set of states that are ``nearby'' to a given initial state in terms of \fuel use.  These sets, bounded by unions of ellipsoidal balls, are then used to show that the modified \FMTstar algorithm maintains its (asymptotic) optimality when applied to CWH dynamics under the 2-norm cost function.  From there, in \cref{sec: Smoothing/Robustness}, the paper presents two techniques for improving motion planning solutions: (i) an analytical technique that can be called both during planning and post-processing to merge $\DeltaV$-vectors between any pair of concatenated graph edges, and (ii) a continuous trajectory smoothing algorithm, which can reduce the magnitude of $\DeltaV$-vectors by relaxing the implicit constraint to pass through sample points while still maintaining solution feasibility.  

The combination of these tools into a unified framework provides a flexible, general technique for near-circular orbit spacecraft trajectory generation  that automatically guarantees bounds on run time and solution quality (\fuel cost) while handling a wide variety  of (possibly non-convex) state, time, and control constraints.  The methodology is demonstrated in \cref{sec: Experiments} on a single-chaser, single-target scenario simulating a near-field Low Earth Orbit (LEO) approach with hard constraints on total maneuver duration, relative positioning (including keep-out-zone and antenna interference constraints), thruster plume impingement, individual and net $\DeltaV$-vector magnitudes, and a two-fault thruster stuck-off failure tolerance.  The performances of \FMTstar and the trajectory smoothing techniques are evaluated as a function of sample count and a \fuel cost threshold.

Preliminary versions of this paper appear in \cite{JS-BB-MP:15, JAS-GDM-BWB-MP:16}.  This extended and revised work introduces the following as additional contributions: (i) a more detailed presentation of the \FMTstar optimality proof, 
(ii) improved trajectory smoothing techniques, and (iii) a six-dimensional (3-DOF) numerical example demonstrating non-planar LEO rendezvous. 

\section{Problem Formulation}
\label{sec: Problem}

We begin by defining the problem we wish to solve.  We model the near-field homing phase and approach for a 
spacecraft seeking to maneuver near a target that is moving in a well-defined, circular orbit (see \cref{subfig: CWH Dynamics}).  Let the \emph{state space} ${\Xspace} \subset {\reals}^{\Dimension}$ represent the $\Dimension$-dimensional region in the target's Local Vertical, Local Horizontal (LVLH) frame in which the mission is defined, and define the \emph{obstacle region} or $\Xobs$ as the set of states within $\Xspace$ that result in mission failure (\eg states colliding with the target or which lie outside of a specified approach corridor, for example). We then also define the \emph{free space} or $\Xfree$ as the complement of $\Xobs$, \ie states in $\Xspace$ which lie outside of obstacles.  As illustrated in \cref{subfig: CWH Planning Query}, let $\Vecx\subinit$ represent the chaser spacecraft's initial state relative to the target, and let $\Vecx\subgoal \in \Xgoal$ be a goal state (a new position/velocity near which the chaser can initiate a docking sequence, a survey maneuver, \etc) inside the \emph{goal region} $\Xgoal$.  Finally, define a \emph{state trajectory} (or simply ``trajectory'') as a piecewise-continuous function of time $\map{\Vecx(t)}{\reals}{\Xspace}$, and let $\PathSet$ represent the set of all state trajectories.  Every state trajectory is implicitly generated by a control trajectory $\map{\Vecu(t)}{\reals}{\ControlSet}$, where $\ControlSet$ is the set of controls, through the system dynamics $\dot{\Vecx} = \StateTransitionFcn\left(\Vecx, \Vecu, t\right)$, where $\StateTransitionFcn$ is the system's state transition function.  A state trajectory is called a \emph{feasible} solution to the planning problem $(\Xfree, t\subinit, \Vecx\subinit, \Vecx\subgoal)$ if: (i) it satisfies the boundary conditions $\Vecx(t\subinit) = \Vecx\subinit$ and $\Vecx(t\subfinal) = \Vecx\subgoal$ for some time $t\subfinal > t\subinit$, (ii) it is \emph{collision-free}; that is, $\Vecx(\tau) \in \Xfree$ for all $\tau \in \closedinterval{t\subinit}{t\subfinal}$, and (iii) it obeys all other trajectory constraints, including the system dynamics.  The general motion planning problem can then be defined as follows.
\nomenclature[AX ]{$\Xspace$}{State space}%
\nomenclature[AR ]{$\reals$}{The field of real numbers}%
\nomenclature[Ad#]{$\Dimension$}{State dimension}%
\nomenclature[AX*]{$\Xobs$}{Obstacle space, or states that result in mission failure}%
\nomenclature[AX#]{$\Xfree$}{Feasible (``free'') state space, or states that lie outside of $\Xobs$}%
\nomenclature[Ax.]{$\Vecx\subinit$}{Initial state}%
\nomenclature[Ax.]{$\Vecx\subgoal$}{Goal state}%
\nomenclature[AX']{$\Xgoal$}{Goal region}%
\nomenclature[B18a]{$\PathSet$}{Set of all state trajectories}%
\nomenclature[AU ]{$\ControlSet$}{Admissible control space}%
\nomenclature[Af(]{$\StateTransitionFcn\left(\ldots\right)$}{State transition function (system dynamics)}%
\nomenclature[Ax,]{$\Vecx$}{State vector}%
\nomenclature[Au#]{$\Vecu$}{Control vector}%
\nomenclature[At)]{$t$}{Time}%

\begin{figure}
	\hfill
	\subcaptionbox{%
		\label{subfig: CWH Dynamics}%
		Schematic of CWH dynamics, which models relative guidance near a single target in a circular orbit.
	}[0.47\columnwidth]{%
		\includegraphics[width=\columnwidth]{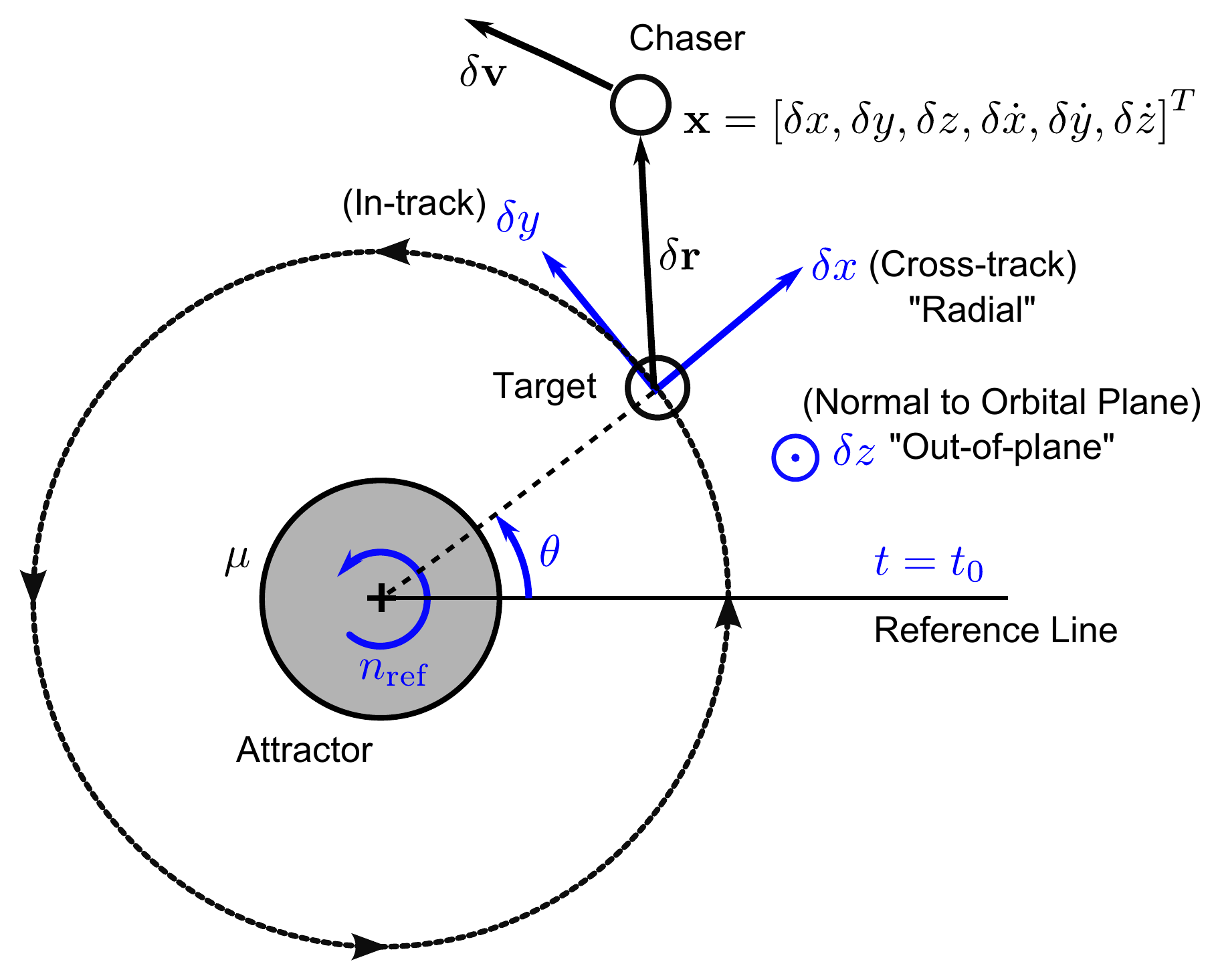}
	}
	\hfill
	\subcaptionbox{%
		\label{subfig: CWH Planning Query}
		A representative motion planning query between feasible states $\Vecx\subinit$ and $\Vecx\subgoal$.
	}[0.48\columnwidth]{%
		\includegraphics[width=\columnwidth]{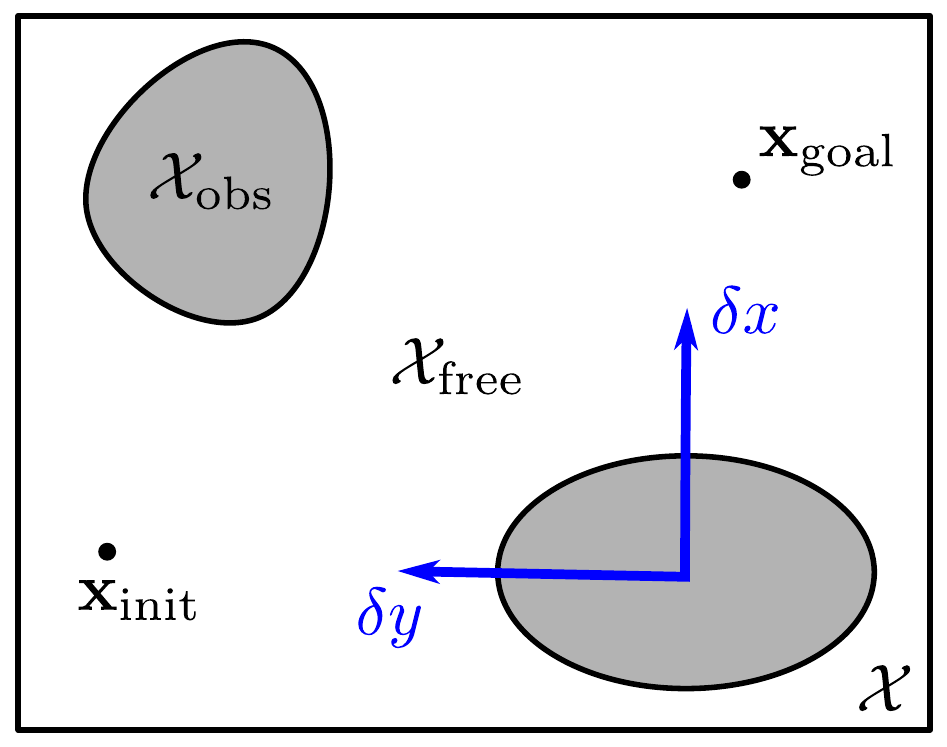}
	}
	\hfill
	\caption{Illustration of the CWH planning scenario.  \iftoggle{AIAAjournal}{}{Here $\omegaRef$ is the mean motion of the target spacecraft orbit, $\theta$ is its mean anomaly, $t$ denotes time, $\delta \Vecr$ and $\delta \Vecv$ are the chaser relative position and velocity, and $(\posCrossTrack, \posInTrack, \posOutofPlane)$ are the LVLH frame coordinates.  The CWH frame rotates with the target at rate $\omegaRef$ as it orbits the gravitational attractor, $\mu$.  Planning takes place in the LVLH frame state space $\Xspace$, in which a safe trajectory is sought between states $\Vecx\subinit$ and $\Vecx\subgoal$ within the feasible (collision-free) subspace $\Xfree$ around nearby obstacles $\Xobs$.}}
	\label{fig: CWH Planning}
\end{figure}
\nomenclature[Ar(]{$\delta \Vecr$}{Position relative to target}%
\nomenclature[Av1]{$\delta \Vecv$}{Velocity relative to target}%
\nomenclature[Ax0]{$\posCrossTrack$}{Cross-track (radial) position}%
\nomenclature[Ay#]{$\posInTrack$}{In-track position}%
\nomenclature[Az#]{$\posOutofPlane$}{Out-of-plane position}%
\nomenclature[Ax1]{$\velCrossTrack$}{Cross-track (radial) velocity}%
\nomenclature[Ay']{$\velInTrack$}{In-track velocity}%
\nomenclature[Az']{$\velOutofPlane$}{Out-of-plane velocity}%
\nomenclature[Anr]{$\omegaRef$}{Target orbit mean motion}%
\nomenclature[B08b ]{$\theta$}{Target orbit mean anomaly}%
\nomenclature[B12b]{$\mu$}{Gravitational parameter}%

\begin{definition}[Optimal Planning Problem]
	\label{def: Optimal Planning Problem}
	Given a planning problem $(\Xfree, t\subinit, \Vecx\subinit, \Vecx\subgoal)$ and a cost functional $\map{\CostFcn}{\PathSet \cross \ControlSet \cross \reals}{\nonnegativereals}$, find a feasible trajectory $\Vecx\supopt(t)$ with associated control trajectory $\Vecu\supopt(t)$ and time span $t = \closedinterval{t\subinit}{t\subfinal}$ for $t\subfinal \in \leftclosedinterval{t\subinit}{\infty}$ such that $\CostFcn\left(\Vecx\supopt(\cdot), \Vecu\supopt(\cdot), t\right) = \min\left\{ \CostFcn\left(\Vecx(\cdot), \Vecu(\cdot), t\right) \suchthat \Vecx(t) \text{ and } \Vecu(t) \text{ are feasible} \right\}$.  If no such trajectory exists, report failure.
\end{definition}
\nomenclature[AJ ]{$\CostFcn$}{Cost functional}%
\nomenclature[AR']{$\nonnegativereals$}{The set of non-negative real numbers}%

For our particular case, we employ a control-effort cost functional $\CostFcn$ that considers only the control trajectory $\Vecu(t)$ and the final time $t\subfinal$, which we represent by the notation $\CostFcn\left(\Vecu(t),t\subfinal\right)$.  Tailoring \cref{def: Optimal Planning Problem} to impulsively-actuated \fuel-optimal motion planning near circular orbits (where here we assume \fuel optimality is measured by the 2-norm metric), the spacecraft motion planning problem we wish to solve is formulated as:
\begin{align}
	\removeequationpadding
	\label{eqn: CWH Optimal Steering}
	\begin{optproblem}
		\given{\text{Initial state } \Vecx\subinit\left(t\subinit\right), \text{Goal region } \Xgoal, \text{Free space } \Xfree} 
		\minimize[\Vecu(t),t\subfinal]{ \CostFcn\left(\Vecu(t),t\subfinal\right) = \int_{t\subinit}^{t\subfinal} \norm[\pvalue]{\Vecu(t)} \D t = \finiteseries[i=1]{\NBurns} \norm[\pvalue]{\VecDeltaV_i} }
		\subjectto
		\constraint[\text{Initial Condition}]{\Vecx(t\subinit) = \Vecx\subinit}
		\constraint[\text{Terminal Condition}]{\Vecx(t\subfinal) \in \Xgoal}
		\constraint[\text{System Dynamics}]{\dot{\Vecx}(t) = \StateTransitionFcn\left(\Vecx(t), \Vecu(t), t\right)}
		\constraint[\text{Obstacle Avoidance}]{\Vecx(t) \in \Xfree \quad\quad\quad\quad\quad\quad\! \text{for all } t \in \closedinterval{t\subinit}{t\subfinal}}
		\constraint[\text{Other Constraints}]{\!\begin{aligned}\InequalityConstraintVec \left(\Vecx(t),\Vecu(t),t\right) \leq 0 \\ \EqualityConstraintVec \left(\Vecx(t),\Vecu(t),t\right) = 0 \end{aligned}
		\quad\quad\quad\!\! \text{for all } t \in \closedinterval{t\subinit}{t\subfinal} 
		}
		\constraint[\text{Active Safety}]{\exists\, \text{ safe } \Vecx\subCAM(\tau), \tau > t \quad\: \text{for all } \Vecx(t)}
	\end{optproblem}
\end{align}
where $t\subinit$ and $t\subfinal$ are the initial and final times, $\Vecx\subCAM(\tau)$ refers to an infinite-horizon Collision-Avoidance-Maneuver (CAM), and we restrict our attention to impulsive control laws $\Vecu(t) = \finiteseries[i=1]{\NBurns} \VecDeltaV_i \delta\left(t - \BurnTime_i\right)$, where $\delta(\cdot)$ denotes the Dirac delta function, which represent a finite sequence of instantaneous translational burns $\VecDeltaV_i$ fired at discrete times $\BurnTime_i$ (note that the number of burns $\NBurns$ is not fixed \emph{a priori}).  Though one could consider the set of all possible control laws, it is both theoretically and computationally simpler to optimize over the finite-dimensional search space enabled by using $\DeltaV$-vectors; furthermore, such control laws represent the most common forms of propulsion systems used on-orbit, including high-impulse cold-gas and liquid bi-propellant thrusters, and can at least in theory approximate continuous control trajectories in the limit that $\NBurns \rightarrow \infty$.
\nomenclature[Ag ]{$\InequalityConstraintVec$}{Inequality constraint functions vector}%
\nomenclature[Ah#]{$\EqualityConstraintVec$}{Equality constraint functions vector}%
\nomenclature[Ati]{$t\subinit$}{Initial time of motion plan}%
\nomenclature[Atf]{$t\subfinal$}{Final time of motion plan}%
\nomenclature[Ax.]{$\Vecx\subCAM$}{Collision-Avoidance-Maneuver trajectory}%
\nomenclature[B04b]{$\delta(\cdot)$}{Dirac-delta function}%
\nomenclature[Av,]{$\VecDeltaV$}{Impulsive velocity change or ``burn''}%
\nomenclature[B19b]{$\BurnTime$}{Time of impulsive burn}%
\nomenclature[AN ]{$\NBurns$}{Number of impulsive burns}%

We now elaborate on the objective function and each constraint in turn.

\subsection{Cost Functional}
\label{subsec: Cost Functional}
A critical component of the spacecraft rendezvous problem is the choice of the cost function.  Consistent with the conclusions of \cite{IMR:04}, we define our cost as the $L^1$-norm of the $\ell_p$-norm of the control.  The best choice for $p \geq 1$ depends on the propulsion system geometry, and on the frame within which $\Vecu(t) = \finiteseries[i=1]{\NBurns} \VecDeltaV_i \delta\left(t - \BurnTime_i\right)$ in $\CostFcn$ is resolved.  Minimizing \fuel directly requires resolving $\VecDeltaV_i$ into the spacecraft body-fixed frame; unfortunately, without relaxations, this requires spacecraft attitude $\Vecq$ to be included in our state $\Vecx$.  To avoid this, a standard used throughout the literature and routinely in practical applications is to employ $p = 2$ so that each $\VecDeltaV_i$ is as short as possible, and optimally allocate the commanded $\VecDeltaV_i$ to thrusters later (online, once the attitude is known).  Though this moves \fuel minimization to a separate control allocation step (which we discuss in more detail in \cref{subsec: Traj Constraints}), it greatly simplifies the problem in a practical way without neglecting attitude.  Because the cost of $\DeltaV$-allocation can only grow due to the need to satisfy torque constraints or impulse bounds (\eg necessitating counter-opposing thrusters to achieve the same net $\DeltaV$-vector), we are in effect minimizing the best-case, unconstrained \fuel use of the spacecraft.  As we will show in our numerical experiments, this does not detract significantly from the technique; the coupling of $\CostFcn$ with $p = 2$ to the actual \fuel use through the minimum control-effort thruster $\DeltaV$ allocation problem seems to promote low \fuelCost solutions.  Hence (in practice) $\CostFcn$ serves as a good proxy to \fuel use, with the added benefit of independence from propulsion system geometry.
\nomenclature[Aq ]{$\Vecq$}{Spacecraft attitude (\eg represented by a quaternion vector)}%

\subsection{Boundary Conditions}
Sampling-based motion planners generally assume a known initial state $\Vecx\subinit$ and time $t\subinit$ from which planning begins (\eg the current state of the spacecraft), and define one or more goal regions $\Xgoal$ to which guidance is sought.  In this paper, we assume the chaser targets only one goal state $\VecxTranspose\subgoal = \left[\,\delta \VecrTranspose\subgoal\ \delta \VecvTranspose\subgoal\,\right]$ at a time (``exact convergence,'' $\Xgoal = \left\{\Vecx\subgoal\right\}$), where $\delta \Vecr\subgoal$ is the goal position and $\delta \Vecv\subgoal$ is the goal velocity.  During numerical experiments, however, we sometimes permit termination at any state whose position and velocity lie within Euclidean balls $\ball[]{\delta \Vecr\subgoal}{\GoalPosTol}$ and $\ball[]{\delta \Vecv\subgoal}{\GoalVelTol}$, respectively (``inexact convergence,'' $\Xgoal = \ball[]{\Vecr\subgoal}{\GoalPosTol} \cross \ball[]{\Vecv\subgoal}{\GoalVelTol}$), where the notation $\ball[]{\Vecr}{\epsilon} = \left\{\Vecx \in \Xspace \suchthat \norm[]{\Vecr - \Vecx} \leq \epsilon \right\}$ denotes a ball with center $\Vecr$ and radius $\epsilon$.
\nomenclature[AB#]{$\ball[]{\Vecr}{\epsilon}$}{Euclidean ball with center $\Vecr$ and radius $\epsilon$}%

\subsection{System Dynamics}\label{sec:sysDyn}
Because spacecraft proximity operations incorporate significant drift, spatially-dependent external forces, and changes on fast timescales, any realistic solution must obey dynamic constraints; we cannot assume straight-line trajectories. In this paper, we employ the classical Clohessy-Wiltshire-Hill (CWH) equations \cite{WHC-RSW:60, GWH:78} for impulsive linearized motion about a circular reference orbit at radius $\rRef$ about an inverse-square-law gravitational attractor with parameter $\mu$.  This model provides a first-order approximation to a chaser spacecraft's motion relative to a rotating target-centered coordinate system (see \cref{fig: CWH Planning}).  The linearized equations of motion for this scenario as resolved in the Local Vertical, Local Horizontal (LVLH) frame of the target are given by:
\begin{subequations}
	\label{eqn: CWH Equations}
	\begin{align}
		\delta\ddot{x} - 3\omegaRef^2\posCrossTrack - 2\omegaRef\velInTrack &= \frac{\VecFMag_x}{m} \\
		\delta\ddot{y} + 2\omegaRef\velCrossTrack &= \frac{\VecFMag_y}{m} \\
		\delta\ddot{z} + \omegaRef^2\posOutofPlane &= \frac{\VecFMag_z}{m}
	\end{align}
\end{subequations}
where $\omegaRef = \sqrt{\frac{\mu}{\rRef^3}}$ is the orbital frequency (mean motion) of the reference spacecraft orbit, $m$ is the spacecraft mass, $\VecF = \left[\VecFMag_x, \VecFMag_y, \VecFMag_z\right]$ is some applied force, and ${(\posCrossTrack, \posInTrack, \posOutofPlane)}$ and $(\velCrossTrack, \velInTrack, \velOutofPlane)$ represent the cross-track (``radial''), in-track, and out-of-plane relative position and relative velocity, respectively.  The CWH model is quite common, and used often for rendezvous and proximity operations in Low Earth Orbit (LEO) and for leader-follower formation flight dynamics.
\nomenclature[AF ]{$\VecF$}{Applied force vector}%

Defining the state $\Vecx$ as $\transpose{\left[\posCrossTrack, \posInTrack, \posOutofPlane, \velCrossTrack, \velInTrack, \velOutofPlane\right]}$ and the applied force-per-unit-mass $\Vecu$ as $\transpose{\frac{1}{m}\VecF}$,
the CWH equations can be described by the \emph{linear time-invariant} (LTI) system:
\nomenclature[Am|]{$m$}{Spacecraft mass}%
\begin{equation}
	\label{eqn: CWH State Dynamics}
	\dot{\Vecx} = \StateTransitionFcn\left(\Vecx,\Vecu,t\right) = \MatA \Vecx + \MatB \Vecu
\end{equation}
where the dynamics matrix $\MatA$ and input matrix $\MatB$ are given by:%
\nomenclature[AA ]{$\MatA$}{System dynamics matrix}%
\nomenclature[AB ]{$\MatB$}{System input matrix}%
\begin{subequations}
	\label{eqn: CWH Matrices}
	{ \setlength{\arraycolsep}{3pt}
	\begin{align*}
		\MatA &= \left[ \begin{array}{cccccc} 0 & 0 & 0 & 1 & 0 & 0 \\
					0 & 0 & 0 & 0 & 1 & 0 \\
					0 & 0 & 0 & 0 & 0 & 1 \\
					3\omegaRef^2 & 0 & 0 & 0 & 2\omegaRef & 0 \\
					0 & 0 & 0 & -2\omegaRef & 0 & 0 \\
					0 & 0 & -\omegaRef^2 & 0 & 0 & 0 \end{array} \right]
		& \MatB &= \left[ \begin{array}{ccc} 0 & 0 & 0 \\
					0 & 0 & 0 \\
					0 & 0 & 0 \\
					1 & 0 & 0 \\
					0 & 1 & 0 \\
					0 & 0 & 1 \end{array} \right].
	\end{align*}%
	}%
\end{subequations}
As for any LTI system, we can express the solution to \cref{eqn: CWH State Dynamics} for any time $t \geq \ti$ using superposition and the convolution integral as $\Vecx(t) = e^{\MatA \left(t - \ti\right)} \Vecx(\ti) + \int_{\ti}^{t} e^{\MatA \left(t - \BurnTime\right)} \MatB \Vecu\left(\BurnTime\right) \D \BurnTime.$
The expression $\MatPhi\left(t,\BurnTime\right) \triangleq e^{\MatA \left(t - \BurnTime\right)}$ is called the \emph{state transition matrix}, which importantly provides an analytical mechanism for computing state trajectories that we rely heavily upon in simulations.  Note, throughout this work, we shall sometimes represent $\MatPhi\left(t,\BurnTime\right)$ as $\MatPhi$ for brevity when its arguments are understood.
\nomenclature[B21a]{$\MatPhi$}{State transition matrix}%

We now specialize the above to the case of $\NBurns$ impulsive velocity changes at times $\ti \leq \BurnTime_i \leq \tf$, for $i \in \orderedlist[1]{\NBurns}$, in which case $\Vecu(\BurnTime) = \sum_{i=1}^\NBurns \VecDeltaV_i \delta(\BurnTime - \BurnTime_i)$, where $\delta(y) = \left\{ 1 \text{ where } y = 0, \text{ or } 0 \text{ otherwise} \right\}$ signifies the Dirac-delta distribution.  Substituting for $\MatPhi$ and $\Vecu\left(\BurnTime\right)$, this yields:
\begin{align*}
	\Vecx(t) &= \MatPhi\left(t,\ti\right) \Vecx(\ti) + \int_{\ti}^{t} \MatPhi\left(t,\BurnTime\right) \MatB \left(\sum_{i=1}^\NBurns \VecDeltaV_i \delta(\BurnTime - \BurnTime_i)\right) \D \BurnTime \\
		&= \MatPhi\left(t,\ti\right) \Vecx(\ti) + \sum_{i=1}^\NBurns \int_{\ti}^{t} \MatPhi\left(t,\BurnTime\right) \MatB \VecDeltaV_i \delta(\BurnTime - \BurnTime_i) \D \BurnTime,
\end{align*}
where on the second line we used the linearity of the integral operator.  By the sifting property of $\delta$, denoting $\NBurns_t$ as the number of burns applied from $\ti$ up to time $t$, we have for all times $t \geq \ti$ the following expression for the impulsive solution to \cref{eqn: CWH State Dynamics}:
\begin{subequations}
	\label{eqn: CWH Impulsive Solution}
	\begin{align}
		\Vecx(t) &= \MatPhi\left(t,\ti\right) \Vecx(\ti) + \sum_{i=1}^{\NBurns_t} \MatPhi\left(t,\BurnTime_i\right) \MatB \VecDeltaV_i
			\\ &= \MatPhi\left(t,\ti\right) \Vecx(\ti) + \underbrace{\left[ \begin{array}{ccc} \MatPhi\left(t,\BurnTime_1\right)\MatB & \dots & \MatPhi\left(t,\BurnTime_{\NBurns_t}\right)\MatB \end{array} \right]}_{\triangleq \MatPhi_v\left(t,\left\{\BurnTime_i\right\}_i\right)} \underbrace{\left[\begin{array}{c} \VecDeltaV_1 \\ \vdots \\ \VecDeltaV_{\NBurns_t} \end{array} \right]}_{\triangleq \MatDeltaV}
			\\ &= \MatPhi\left(t,\ti\right) \Vecx(\ti) + \MatPhi_v\left(t,\left\{\BurnTime_i\right\}_i\right) \MatDeltaV.
	\end{align}
\end{subequations}
\nomenclature[B21a]{$\MatPhi_v$}{Aggregation of impulse state transition matrices}%
\nomenclature[AV)]{$\MatDeltaV$}{Stacked impulse vector}%
Throughout this paper, the notations $\MatDeltaV$ for the stacked $\DeltaV$-vector and $\MatPhi_v\left(t,\left\{\BurnTime_i\right\}_i\right)$ for the aggregated impulse state transition matrix (or simply $\MatPhi_v$ for short, when the parameters $t$ and $\{\BurnTime_i\}_i$ are clear) implicitly imply only those burns $i$ occurring before time $t$.

\subsection{Obstacle Avoidance}
Obstacle avoidance is imposed by requiring the spacecraft trajectory $\Vecx(t)$ to stay within $\Xfree$ (or, equivalently, outside of the obstacle region $\Xobs$) -- typically a difficult non-convex constraint.  For CWH proximity operations, $\Xobs$ might include those states that result in a collision with a neighboring object, all positions which lie outside of a given approach corridor, all velocities violating a given relative guidance law, \etc.  In our numerical experiments, we assume $\Xobs$ includes an ellipsoidal Keep Out Zone (KOZ) enclosing the target spacecraft centered at the origin and a conical nadir-pointing region that approximates its antenna beam pattern -- this both enforces collision-avoidance and prevents the chaser from interfering with the target's communications.

Note that according to the definition of $\Xfree$, this also requires the solution $\Vecx(t)$ to stay within the confines of $\Xspace$ (CWH system dynamics do not guarantee that state trajectories will lie in $\Xspace$ despite the fact that their endpoints do).  Though not strictly necessary in practice, if $\Xfree$ is defined to mark the extent of reliable sensor readings or the boundary inside which CWH equation linearity assumptions hold, then this can be a useful constraint to enforce.

\subsection{Other Trajectory Constraints}
\label{subsec: Traj Constraints}

Many other types of constraints may be included to encode additional restrictions on state and control trajectories, which we represent here by a set of inequality constraints $\InequalityConstraintVec$ and equality constraints $\EqualityConstraintVec$ (note that $\InequalityConstraintVec$ and $\EqualityConstraintVec$ denote vector functions).  To illustrate the flexibility of the sampling-based planning approach, we encode the following into constraints $\InequalityConstraintVec$ (for brevity, we omit their exact representation, which is straightforward based on vector geometry):
\begin{align*}
	&\TotalManeuverDuration\submin \leq t\subfinal - t\subinit \leq \TotalManeuverDuration\submax
		& &
		& & \text{Plan Duration Bounds}
	\\ &\VecDeltaV_i \in {\ControlSet}\left(\Vecx(\BurnTime_i)\right)
		& & \text{ for all } i = \orderedlist{\NBurns}
		& & \text{Control Feasibility}
	\\ \smashoperator[r]{\Union_{k \in \finitelist[1]{\Nthrusters}}} & \Plume_k\left(\PlumeAxis_k, \PlumeHalfAngle, \PlumeHeight\right) \intersect \PlumeAvoidanceTarget = \nullset
		& & \text{ for all } i = \orderedlist{\NBurns} 
		& & \text{Plume Impingement}
\end{align*}
\nomenclature[Ai ]{$i$}{Impulse index}%
\nomenclature[Ak ]{$k$}{Thruster index}%
\nomenclature[AK ]{$\Nthrusters$}{Number of thrusters}%
\nomenclature[AT(]{$\TotalManeuverDuration\submin, \TotalManeuverDuration\submax$}{Duration bounds for the motion plan}%
\nomenclature[AP ]{$\Plume$}{Exhaust plume}%
\nomenclature[B02b]{$\PlumeHalfAngle$}{Exhaust plume half-angle}%
\nomenclature[AH ]{$\PlumeHeight$}{Exhaust plume height}%
\nomenclature[ASt]{$\PlumeAvoidanceTarget$}{Target spacecraft circumscribing sphere}%
%
\noindent Here $0 \leq \TotalManeuverDuration\submin < \TotalManeuverDuration\submax$ represent minimum and maximum motion plan durations, $\ControlSet\left(\Vecx(\BurnTime_i)\right)$ is the admissible \emph{control set} corresponding to state $\Vecx(\BurnTime_i)$%
, $\Plume_k$ is the exhaust plume emanating from thruster $k$ of the chaser spacecraft while executing burn $\VecDeltaV_i$ at time $\BurnTime_i$, and $\PlumeAvoidanceTarget$ is the target spacecraft circumscribing sphere.
We motivate each constraint in turn.

\paragraph{Plan Duration Bounds}
Plan duration bounds facilitate the inclusion of rendezvous windows based on the epoch of the chaser at $\Vecx\subinit\left(t\subinit\right)$; such windows might be determined by illumination requirements, grounds communication opportunities, or mission timing restrictions, for example.  $\TotalManeuverDuration\submax$ may also be used to ensure the errors incurred by our linearized CWH model, which grow with time, do not exceed acceptable tolerances.  

\paragraph{Control Feasibility}
\label{para: Control Feasibility}
Control set constraints are intended to encapsulate limitations on control authority imposed by propulsive actuators and their geometric distribution about the spacecraft.  For example, given the maximum burn magnitude $0 < \VecDeltaVMag\submax$, the constraint:
\begin{align}
	\label{eqn: Net Delta-V Bounds}
		\norm[\pvalue]{\VecDeltaV_i} \leq \VecDeltaVMag\submax
		& \text{ for all } i = \orderedlist{\NBurns}
\end{align}
\nomenclature[Av.]{$\VecDeltaVMag\submax$}{Maximum (net) impulse vector magnitude}%
might be used to represent an upper bound on the impulse range achievable by a gimbaled thruster system that is able to direct thrust freely in all directions. 
In our case, we use ${\ControlSet}\left(\Vecx(\BurnTime_i)\right)$ to represent all commanded net $\DeltaV$-vectors that (i) satisfy the constraint \cref{eqn: Net Delta-V Bounds} above, and also (ii) can be successfully allocated to thrusters along trajectory $\Vecx(t)$ at time $\BurnTime_i$ according to a simple minimum-control effort thruster allocation problem (a straightforward linear program (LP) \cite{MB:02}).
%
\bgroup
\renewcommand{\subnet}{_i}
To keep the paper self-contained, we repeat the problem here and in our notation.  Let $\resolveat{\VecThrust\subnet}{\BodyFixedFrame}$ and $\resolveat{\VecMoment\subnet}{\BodyFixedFrame}$ be the desired net $\DeltaV$ and moment vectors at burn time $\BurnTime_i$, resolved in the body-fixed frame according to attitude $\Vecq\left(\BurnTime_i\right)$ (we henceforth drop the bar, for brevity).  Note the attitude $\Vecq\left(\BurnTime_i\right)$ must either be included in the state $\Vecx\left(\BurnTime_i\right)$ or be derived from it, as is assumed in this paper by imposing (along nominal trajectories) a nadir-pointing attitude profile for the chaser spacecraft.  Let $\AllocatedThrust_k = \norm[2]{\VecAllocatedThrust_k}$ be the $\DeltaV$-magnitude allocated to thruster $k$, which generates an impulse in direction $\VecThrustDir_k$ at position $\VecThrustPos_k$ from the spacecraft center-of-mass (both are constant vectors if resolved in the body-fixed frame).  Finally, to account for the possibility of on or off thrusters, let $\ThrustAvailable_k$ be equal to $1$ if thruster $k$ is available for burn $i$, or $0$ otherwise.  Then the minimum-effort control allocation problem can be represented as:
\nomenclature[Av0]{$\VecThrust\subnet$}{Net impulse vector commanded for $\DeltaV$ allocation}%
\nomenclature[AMi]{$\VecMoment\subnet$}{Net moment vector commanded for $\DeltaV$ allocation}%
\nomenclature[Av-]{$\AllocatedThrust_k$}{Impulse vector magnitude allocated to thruster $k$}%
\nomenclature[Av-]{$\VecThrustDir_k$}{Unit-vector direction of impulse allocated to thruster $k$}%
\nomenclature[B17b]{$\VecThrustPos_k$}{Position vector of thruster $k$ relative to the center-of-mass}%
\nomenclature[B07b]{$\ThrustAvailable_k$}{(On-off) boolean flag indicating availability of thruster $k$}%
\begin{align}
	\begin{optproblem}
		\given{
			\text{On-off flags } \ThrustAvailable_k, \text{thruster positions } \VecThrustPos_k, \text{thruster axes } \VecThrustDir_k,
		}
		& & \mathrlap{\text{commanded } \DeltaV\text{-vector } \VecThrust\subnet, \text{and commanded moment vector } \VecMoment\subnet} \\
		\minimize[\AllocatedThrust_k]{\sum_{k=1}^{\Nthrusters} \AllocatedThrust_k}
		\subjectto
		\constraint[\text{Net } \DeltaV \text{-Vector Allocation}]{ \sum_{k=1}^{\Nthrusters} \VecThrustDir_k \left(\ThrustAvailable_k \AllocatedThrust_k\right) = \VecThrust\subnet }
		\constraint[\text{Net Moment Allocation}]{ 
			\sum_{k=1}^{\Nthrusters} \left( \VecThrustPos_k \cross \VecThrustDir_k \right) \left(\ThrustAvailable_k \AllocatedThrust_k\right)
			= \VecMoment\subnet
		}
		\constraint[\text{Thruster } \DeltaV \text{ Bounds}]{ \AllocatedThrustMin_k \leq 
			\AllocatedThrust_k \leq \AllocatedThrustMax_k }
	\end{optproblem}
	\label{eqn: Minimum-Effort Delta-V Allocation}
\end{align}
where $\AllocatedThrustMin_k$ and $\AllocatedThrustMax_k$ represent minimum and maximum impulse limits on thruster $k$ (due to actuator limitations, minimum impulse bit, pulse-width constraints, or maximum on-time restrictions, for example).
Because $\DeltaV$ is directly-proportional to thrust through the Tsiokolvsky rocket equation, the formulation above is directly analogous to minimum-\fuel consumption; as discussed in \cref{subsec: Cost Functional}, by using control trajectories that minimize commanded $\DeltaV$-vector lengths $\norm{\VecThrust\subnet}$, we can drive \fuel use downwards as much as possible subject to our thrust bounds and net torque constraints.  In this work, we set $\VecMoment\subnet = \VecZeros$ to enforce torque-free burns and minimize the disturbance to our assumed attitude trajectory $\Vecq(t)$.
\nomenclature[Av/]{$\AllocatedThrustMin_k$, $\AllocatedThrustMax_k$}{Bounds on impulse magnitudes of thruster $k$}%

Note that we do not consider a minimum norm constraint in \cref{eqn: Net Delta-V Bounds} for $\VecThrust\subnet$.  As discussed in \cref{subsec: Cost Functional}, $\|\VecThrust\subnet\|$ is only a proxy for the true \fuel cost computed from the thrust allocation problem (\cref{eqn: Minimum-Effort Delta-V Allocation}).  The value of the norm bound $\VecDeltaVMag\submax$ may be computed from the thruster limits $\AllocatedThrustMin_k, \AllocatedThrustMax_k$ and knowledge of the thruster configuration.

\paragraph{Plume Impingement}
Impingement of thruster exhaust on neighboring spacecraft can lead to dire consequences, including destabilizing effects on attitude caused by exhaust gas pressure, degradation of sensitive optical equipment and solar arrays, and unexpected thermal loading \cite{GD:91}.  To take this into account, we generate representative exhaust plumes at the locations of each thruster firing.  For burn $i$ occurring at time $\BurnTime_i$, a right circular cone is generated with axis $\PlumeAxis_k$, half-angle $\PlumeHalfAngle$, and height $\PlumeHeight$ at each \emph{active} thruster $k$ ($\ThrustAvailable_k = 1$) for which its allocated thrust $\AllocatedThrust_k\supopt$ is non-zero, as determined by the solution to \cref{eqn: Minimum-Effort Delta-V Allocation}.  Intersections are checked with the target spacecraft circumscribing sphere, $\PlumeAvoidanceTarget$, which is used as a more efficiently-verified, conservative approximation to the exact target geometry.  For an illustration, see \cref{subfig: Plume Impingement}.

\begin{figure}
		\centering
		\begin{tikzpicture}
			\node[anchor=south west,inner sep=0] (image) {\includegraphics[width=0.55\columnwidth]{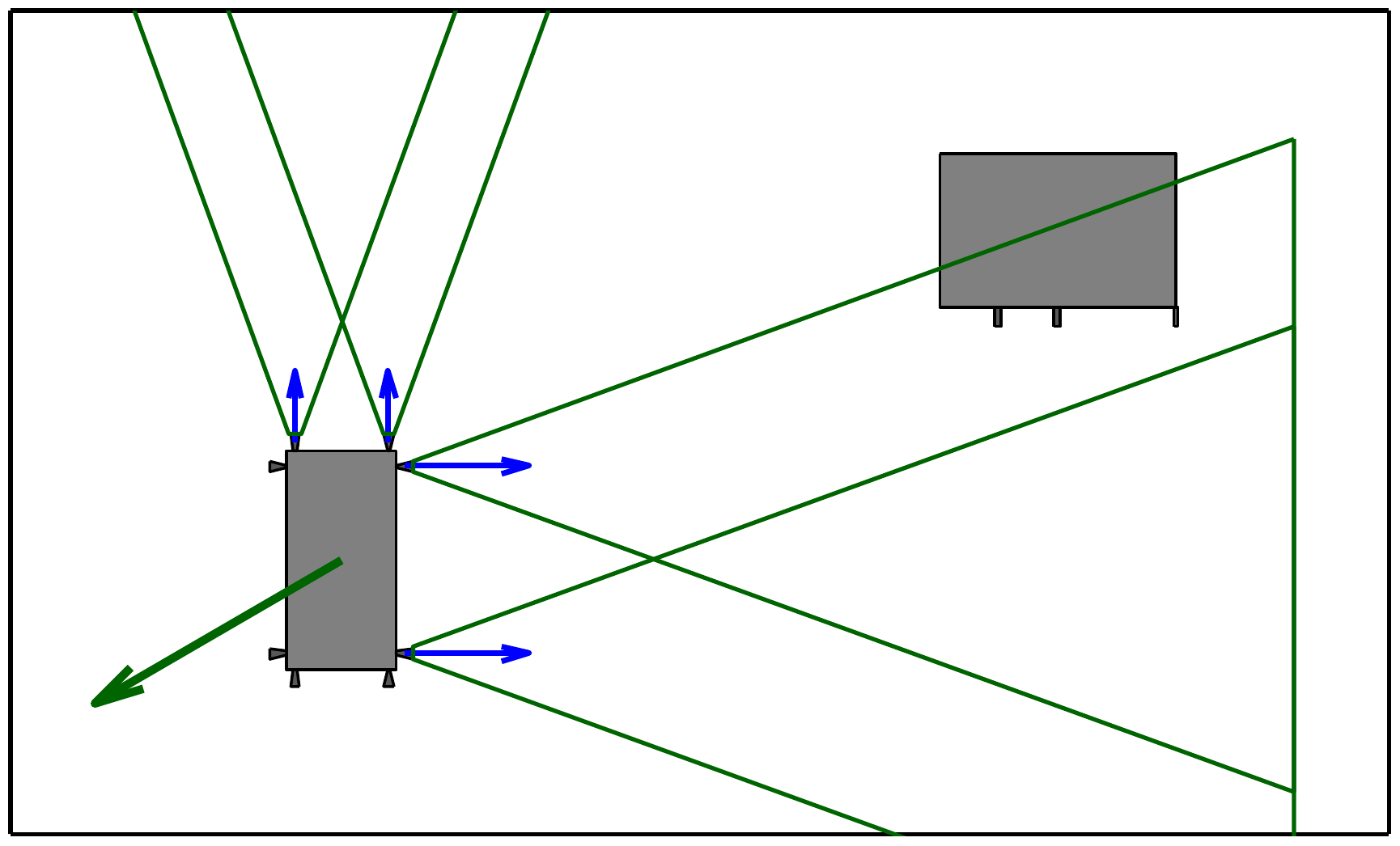}};
			\begin{scope}[x={(image.south east)},y={(image.north west)}]
				\coordinate (targetspacecraft) at (0.76,0.72);											
				\node [anchor=south east] (targetcircumsphere) at 
					([shift={(-0.09,0.02)}]targetspacecraft.center) {$\PlumeAvoidanceTarget$};			
				\node [anchor=west] (plumehalfangle) at (0.49,0.51) {$\PlumeHalfAngle$}; 				
				\node [anchor=north west] (deltaVnet) at (0.09,0.20) {$\VecDeltaV_i$}; 					
				\node [anchor=west] (deltaV) at (0.37,0.22) {$-\AllocatedThrust_k \VecThrustDir_k$}; 	
				\node [anchor=north] (plumeheight) at (0.73,0.39) {$\PlumeHeight$}; 					
				\draw [-,thick,dashed,black] (0.39,0.45) to (1.0,0.45); 								
				\draw [very thick,solid,black] (0.49,0.45) to[out=90, in=-30] (0.47,0.56); 				
				\draw [stealth-stealth,thick,solid,black] (0.29,0.4) to (0.92,0.4); 					
				\draw [-,thick,solid,black] (0.29,0.37) to (0.29,0.43); 								
				\draw [-,thick,solid,black] (0.92,0.37) to (0.92,0.43); 								
			\end{scope}
			\draw [dotted,thick,blue] (targetspacecraft.center) circle (0.06\columnwidth);				
		\end{tikzpicture}
		\caption{Illustration of exhaust plume impingement from thruster firings.  \iftoggle{AIAAjournal}{}{Given commanded $\VecDeltaV_i$, the spacecraft must successfully allocate the impulse to thrusters while simultaneously avoiding impingement of neighboring object(s).}}
		\label{subfig: Plume Impingement}
\end{figure}
\egroup

\paragraph{Other Constraints}
Other constraints may easily be added.  Solar array shadowing, pointing constraints, approach corridor constraints, \etc, all fit within the framework, and may be represented as additional inequality or equality constraints.  For additional examples, refer to \cite{JAS-BA-IADN-MP:15}.

\subsection{Active Safety}
\label{subsec: Active Safety}
An additional feature we include in our work is the concept of \emph{active safety}, in which we require the target spacecraft to maintain a feasible Collision Avoidance Maneuver (CAM) to a safe higher or lower circular orbit from every point along its solution trajectory in the event that any mission-threatening control degradations such as stuck-off thrusters (as in \cref{fig: Control Allocation}) take place.  This reflects our previous work \cite{JS-BB-MP:15}, and is detailed more thoroughly in \cref{sec: Safety}.  

\begin{figure}
	\subcaptionbox{%
		Thruster allocation without stuck-off failures.
	}[0.46\columnwidth]{%
		\includegraphics[height=0.20\textheight]{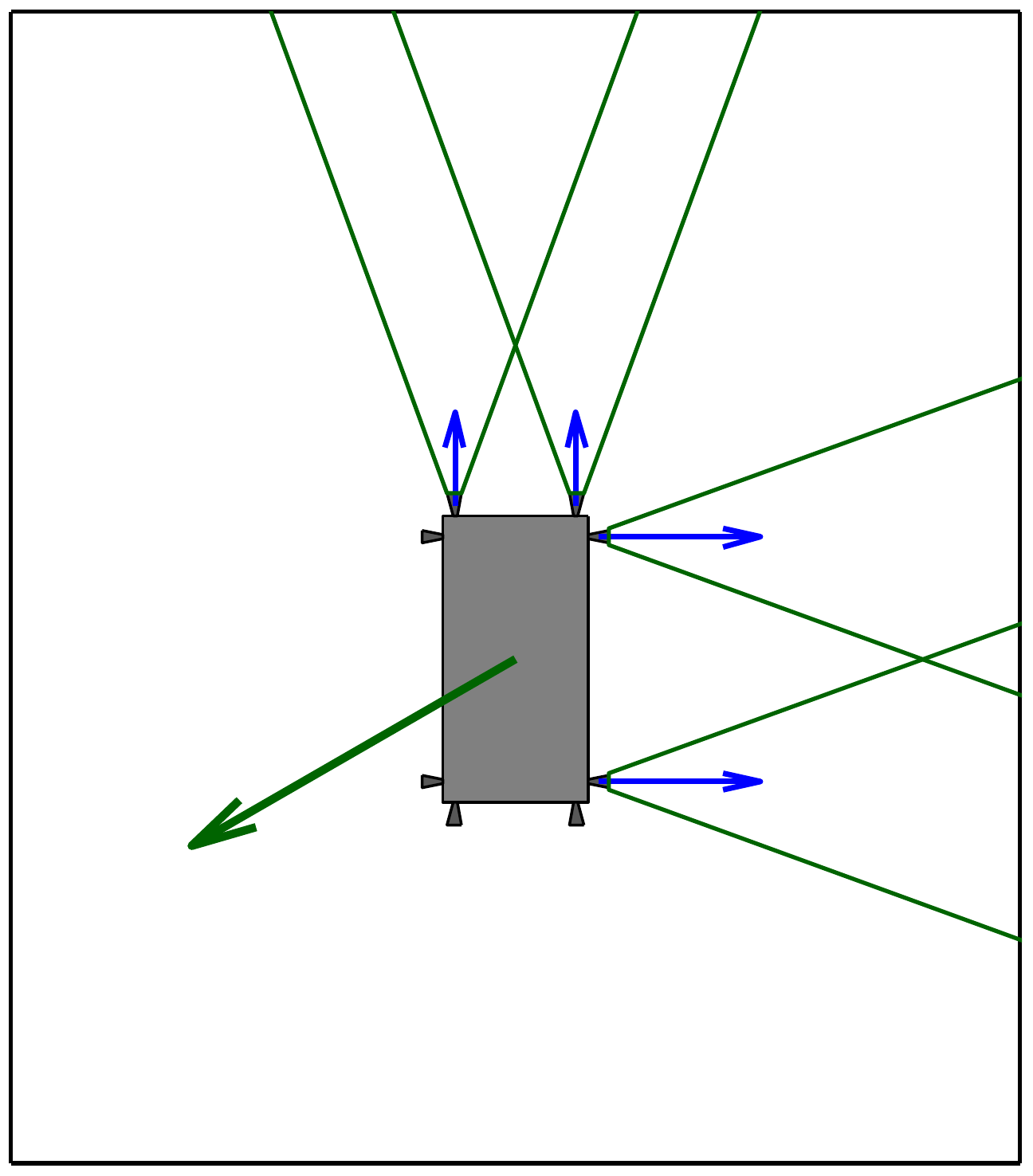}
	}
	\hspace{1em}
	\subcaptionbox{%
		The same allocation problem, with both upper-right thrusters ``stuck off.''
	}[0.46\columnwidth]{%
		\begin{tikzpicture}
			\node[anchor=south west,inner sep=0] (image) {\includegraphics[height=0.20\textheight]{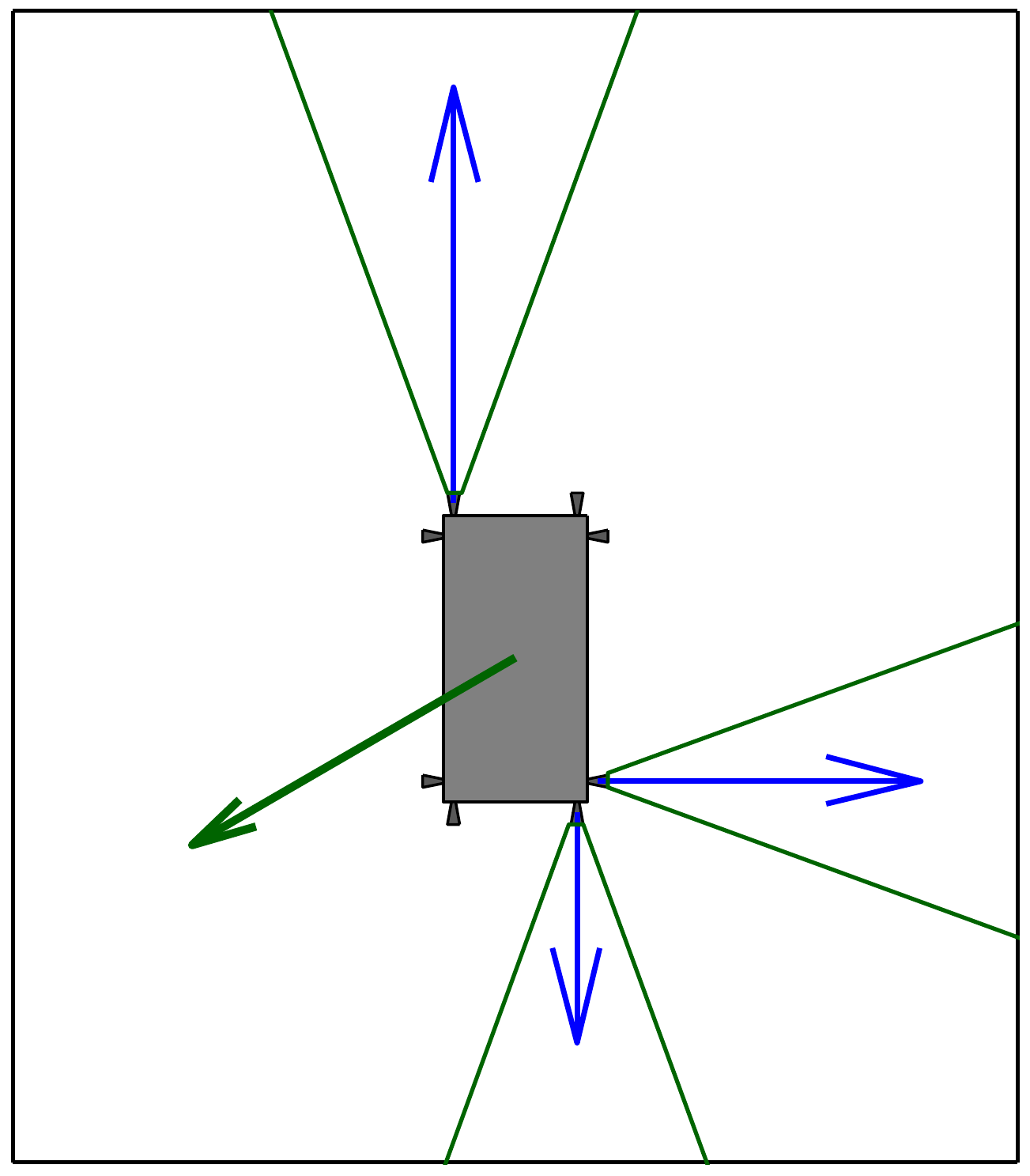}};
			\begin{scope}[x={(image.south east)},y={(image.north west)}]
				\draw[red,ultra thick,rounded corners] (0.54,0.53) rectangle (0.60,0.59);
				\draw[red,ultra thick,solid] (0.55,0.535) to (0.59,0.58);
			\end{scope}
		\end{tikzpicture}
	}
	\caption{Changes to torque-free control allocation in response to thruster failures.  \iftoggle{AIAAjournal}{}{As can be seen, for the same net $\VecDeltaV$ vector (the large green arrow), thruster configuration changes can have a profound impact on thruster $\DeltaV$ locations and magnitudes (blue arrows), and hence on plume impingement satisfaction and thereby the safety of proposed $\VecDeltaV$ trajectories.}}
	\label{fig: Control Allocation}
\end{figure}

\section{Vehicle Safety under CWH Dynamics}
\label{sec: Safety}

In this section, we describe a general strategy for handling the active safety constraints introduced in \cref{eqn: CWH Optimal Steering,subsec: Active Safety}, whereby we seek to deterministically certify our solution trajectory with respect to thruster ``stuck-off'' failures up to a given fault-tolerance but without compromising real-time implementability.  As will be motivated, the idea is to embed the escape trajectory generation process together with positively-invariant set safety constraints into the sampling routine of sampling-based motion planners.  Due to the \fuel-limited nature of many spacecraft proximity operations missions, emphasis is placed on finding minimum-$\DeltaV$ escape maneuvers in order to improve mission reattempt opportunities.  We prioritize \emph{active safety} measures (which allow actuated Collision Avoidance Maneuvers or CAMs) over passive safety guarantees (which shut off all thrusters and restrict the system to zero control) in order to broaden the search space for abort trajectories and thereby allow better performance during the nominal mission.  We emulate the rendezvous design process taken by Barbee \etal \cite{BWB-JRC-SH-FLM-MM-BJN-JVE:11}, but numerically optimize \fuel consumption and remove much of its reliance on user intuition by automating the satisfaction of safety constraints.

Consistent with the notions proposed by Schouwenaars \etal \cite{TS-JPH-EF:04a}, Fehse \cite[4.1.2]{WF:03}, and Fraichard \cite{TF:07}, the definition of vehicle safety in this paper is taken as the following:
\begin{definition}[Vehicle Safety]
	\label{def: Vehicle Safety}
	A vehicle state is \emph{safe} if and only if there exists, under the worst-possible environment and failure conditions, a collision-free, dynamically-feasible trajectory satisfying the constraints that navigates the vehicle to a set of states in which it can remain indefinitely.
\end{definition}
Note \emph{indefinitely} (or sufficiently-long for all practical purposes under the accuracy of the dynamics model) is a critical component of the definition.  Trajectories without infinite-horizon safety guarantees can ultimately violate constraints \cite{LB-JPH:08}, thereby posing a risk that can defeat the purpose of using a hard guarantee in the first place.  For this reason, we impose safety constraints over an infinite-horizon (or, as we will show using invariant sets, an \emph{effectively} infinite horizon).

Consider the scenario described in \cref{sec: Problem} for a spacecraft with nominal state trajectory $\Vecx(t) \in \Xspace$ and control trajectory $\Vecu(t) \in \ControlSet{\left(\Vecx(t)\right)}$ evolving over time $t$ in time span $\TimeHorizon = \leftclosedinterval{t\subinit}{\infty}$.  Let $\TimeHorizon\subfail \subseteq \TimeHorizon$ represent the set of potential failure times we wish to certify (for instance, a set of prescribed burn times $\left\{\BurnTime_i\right\}$, the final approach phase $\TimeHorizon_{\mathrm{approach}}$, or the entire maneuver span $\TimeHorizon$).  When a failure occurs, control authority is lost through a reduction in actuator functionality, negatively impacting system controllability.  Let $\ControlSet\subfail{\left(\Vecx\right)} \subset \ControlSet{\left(\Vecx\right)}$ represent the new control set, where we assume that $\VecZeros \in \ControlSet\subfail$ for all $\Vecx$ (\ie we assume that no actuation is always a feasible option).
\nomenclature[AT ]{$\TimeHorizon$}{Time span}%
\nomenclature[AT#]{$\TimeHorizon\subfail$}{Set of potential failure times}%
\nomenclature[AU#]{$\ControlSet\subfail$}{Admissible control set after a failure (``failure modes'')}%
	Mission safety is commonly imposed in two different ways \cite[4.4]{WF:03}:

\begin{itemize}
	\item \emph{Passive Safety}: For all $t\subfail \in \TimeHorizon\subfail$, ensure that $\Vecx\subCAM{\left(\tCAM\right)}$ satisfies \cref{def: Vehicle Safety} with $\Vecu\subCAM{\left(\tCAM\right)} = \VecZeros$ for all $\tCAM \geq t\subfail$.  For spacecraft, this means its coasting arc from the point of failure must be safe for all future time (though practically this is imposed only over a finite horizon).
	\item \emph{Active Safety}: For all $t\subfail \in \TimeHorizon\subfail$ and failure modes $\ControlSet\subfail$, design actuated collision avoidance maneuvers $\Vecx\subCAM\left(\tCAM\right)$ to satisfy \cref{def: Vehicle Safety} with $\Vecu\subCAM{\left(\tCAM\right)} \in \ControlSet\subfail$ for all $\tCAM \geq t\subfail$, where $\Vecu\subCAM{\left(\tCAM\right)}$ is not necessarily restricted to $\VecZeros$.
\end{itemize}
\nomenclature[Atf]{$t\subfail$}{Time of failure}%
In much of the literature, only passive safety is considered out of a need for tractability (to avoid verification over a combinatorial explosion of failure mode possibilities) and in order to capture the common failure mode in which control authority is lost completely.  Though considerably simpler to implement, this approach potentially neglects many mission-saving control policies.

\subsection{Positively-Invariant Set Constraints}
Instead of evaluating trajectory safety for all future times $\tCAM \geq t\subfail$, it is generally more practical to consider finite-time solutions starting at $\Vecx\left(t\subfail\right)$ that terminate at a point inside a safe positively invariant set $\Xinvariant$.  If the maneuver is safe and the invariant set is safe for all time, then safety of the spacecraft is assured.

\begin{definition}[Positively Invariant Set]
	\label{def: Positively Invariant Set}
	A set $\Xinvariant$ is positively-invariant with respect to the autonomous system $\dot{\Vecx}\subCAM = \StateTransitionFcn(\Vecx\subCAM)$ if and only if
	$\Vecx\subCAM\left(t\subfail\right) \in \Xinvariant$ implies $\Vecx\subCAM\left(\tCAM\right) \in \Xinvariant$ for all $\tCAM \geq t\subfail$.
\end{definition}
\nomenclature[AX(]{$\Xinvariant$}{Positively-invariant set}%

This yields the following definition for finite-time verification of trajectory safety:

\begin{definition}[Finite-Time Trajectory Safety Verification]
	\label{def: Finite-Time Traj Safety}
	For all $t\subfail \in \TimeHorizon\subfail$ and for all $\ControlSet\subfail\left(\Vecx(t\subfail)\right) \subset \ControlSet\left(\Vecx(t\subfail)\right)$, there exists $\left\{\Vecu(\tCAM), \tCAM \geq t\subfail\right\} \in \ControlSet\subfail\left(\Vecx(t\subfail)\right)$ and $\horizonTime > t\subfail$ such that $\Vecx(\tCAM)$ is feasible for all $t\subfail \leq \tCAM \leq \horizonTime$ and $\Vecx(\horizonTime) \in \Xinvariant \subseteq \Xfree$, 
\end{definition}
\nomenclature[AT']{$\horizonTime$}{(Finite) horizon time used for safety verification}%
\noindent where $\horizonTime$ is some finite safety verification horizon time. Though in principle any \emph{safe} positively-invariant set $\Xinvariant$ is acceptable, not just any will do in practice; in real-world scenarios, unstable trajectories caused by model uncertainties could cause state divergence towards configurations whose safety has not been verified.  Hence care must be taken to use only \emph{stable} positively-invariant sets.

Combining \cref{def: Finite-Time Traj Safety} with our constraints in \cref{eqn: CWH Optimal Steering} from \cref{sec: Problem}, spacecraft trajectory safety after a failure at $\Vecx\left(t\subfail\right) = \Vecx\subfail$ can be expressed in its full generality as the following optimization problem in decision variables $\horizonTime \in \leftclosedinterval{t\subfail}{\infty}$, $\Vecx\subCAM\left(t\right)$, and $\Vecu\subCAM\left(t\right)$, for $t \in \closedinterval{t\subfail}{\horizonTime}$:
\begin{equation}
	\label{eqn: Finite-Time Traj Safety OptProb}
	\begin{optproblem}
		\given{\text{Failure state } \Vecx\subfail\left(t\subfail\right), \text{failure control set } \ControlSet\subfail\left(\Vecx\subfail\right), \text{the free space } \Xfree,}
		& & \mathrlap{\text{a safe, stable invariant set } \Xinvariant, \text{and a fixed number of impulses } \NBurns} \\
		\minimize[\begin{subarray}{c}%
			\Vecu\subCAM(t) \in \ControlSet\subfail\left(\Vecx\subfail\right), \\
			\horizonTime, \Vecx\subCAM(t)
		\end{subarray}]{ \CostFcn\left(\Vecx\subCAM\left(t\right), \Vecu\subCAM\left(t\right), t\right) = \int_{t\subfail}^{\horizonTime} \norm[\pvalue]{\Vecu\subCAM(t)} \D t = \finiteseries[i=1]{\NBurns} \norm[\pvalue]{\VecDeltaV\subCAM{}_{,}{}_i} }
		\subjectto
		\constraint[\text{System Dynamics}]{\dot{\Vecx}\subCAM\left(t\right) = \StateTransitionFcn\left( \Vecx\subCAM\left(t\right), \Vecu\subCAM\left(t\right), t \right)}
		\constraint[\text{Initial Condition}]{ \Vecx\subCAM\left(t\subfail\right) = \Vecx\subfail }
		\constraint[\text{Safe Termination}]{ \Vecx\subCAM\left(\horizonTime\right) \in \Xinvariant }
		\constraint[\text{Obstacle Avoidance}]{\Vecx\subCAM(t) \in \Xfree \quad\quad\quad\quad\quad\;\, \text{for all } t \in \closedinterval{t\subfail}{\horizonTime}}
		\constraint[\text{Other Constraints}]{\!\begin{aligned}
			\InequalityConstraintVec \left(\Vecx\subCAM,\Vecu\subCAM,t\right) \leq 0 \\
			\EqualityConstraintVec \left(\Vecx\subCAM,\Vecu\subCAM,t\right) = 0
		\end{aligned}
		\quad\quad\quad\!\! \text{for all } t \in \closedinterval{t\subfail}{\horizonTime} 
		}
	\end{optproblem}
\end{equation}
\nomenclature[Ax.]{$\Vecx\subfail$}{State vector at which a failure occurs}%
This is identical to \cref{eqn: CWH Optimal Steering}, except that now under failure mode $\ControlSet\subfail\left(\Vecx\subfail\right)$ we abandon the attempt to terminate at a goal state in $\Xgoal$ and instead replace it with a constraint to terminate at a safe, stable positively-invariant set $\Xinvariant$. We additionally neglect any timing constraints encoded in $\InequalityConstraintVec$ as we are no longer concerned with our original rendezvous.  Typically any feasible solution is sought following a failure, in which case one may use $\CostFcn = 1$.  However, to enhance the possibility of mission recovery, we assume the same minimum-\fuel 2-norm cost function as before, but with the exception that here, as we will motivate, we use a single-burn strategy with $\NBurns = 1$.

\subsection{Fault-Tolerant Safety Strategy}
\label{subsec: Safety Strategy}
The difficulty of solving the finite-time trajectory safety problem lies in the fact that a feasible solution must be found for \emph{all} possible failure times (typically assumed to be any time during the mission) as well as for \emph{all} possible failures.  To illustrate, for an $\FaultTolerance$-fault tolerant spacecraft with $\Nthrusters$ control components (thrusters, momentum wheels, CMGs, \etc) that we each model as either ``operational'' or ``failed,'' this yields a total of 
$N\subfail = \sum\limits_{\FaultIndex = 0}^{\FaultTolerance} \left(\begin{smallmatrix} \Nthrusters \\ \FaultIndex \end{smallmatrix}\right)
	= \sum\limits_{\FaultIndex = 0}^{\FaultTolerance} \frac{\Nthrusters!}{\left(\Nthrusters - \FaultIndex\right)! \FaultIndex!}$
possible optimization problems that must be solved for every time $t\subfail$ along the nominal trajectory.
By any standard, this is intractable, and hence explains why so often \emph{passive} safety guarantees are selected (requiring only one control configuration check instead of $N\subfail$, since we prescribe $\Vecu\subCAM = \VecZeros$ which must lie in $\ControlSet\subfail$ given our assumption.  This is analogous to setting $\FaultIndex = \Nthrusters$ with $\FaultTolerance \triangleq \Nthrusters$).  One idea for simplifying the problem while still satisfying safety (the constraints of \cref{eqn: Finite-Time Traj Safety OptProb}) consists of the following strategy:
\nomenclature[Af']{$\FaultIndex$}{Fault index}%
\nomenclature[AF#]{$\FaultTolerance$}{Maximum fault tolerance}%
\nomenclature[AN']{$N\subfail$}{Number of potential failure modes}%

\begin{definition}[Fault-Tolerant Active Safety Strategy]
	\label{def: Fault-Tolerant Active Safety}
	As a conservative solution to the optimization problem in \cref{eqn: Finite-Time Traj Safety OptProb}, it is sufficient (but not necessary) to implement the following procedure:%
	\vspace*{-0.25em}%
	\begin{enumerate}
		\item \label{step: Prescribe a CAM}
		From each $\Vecx(t\subfail)$, prescribe a Collision-Avoidance Maneuver (CAM) policy $\Pi\subCAM$ that gives a horizon time $\horizonTime$ and escape control sequence $\Vecu\subCAM = \Pi\subCAM\left(\Vecx(t\subfail)\right)$ designed to automatically satisfy $\Vecu\subCAM(\tau) \subset \ControlSet$ for all $t\subfail \leq \tau \leq \horizonTime$ and $\Vecx\left(\horizonTime\right) \in \Xinvariant$.
		\item \label{step: Determine CAM Feasibility}
		For each failure mode $\ControlSet\subfail\left(\Vecx(t\subfail)\right) \subset \ControlSet\left(\Vecx(t\subfail)\right)$ up to tolerance $\FaultTolerance$, determine if the control law is feasible; that is, see if $\Vecu\subCAM = \Pi\subCAM\left(\Vecx(t\subfail)\right) \subset \ControlSet\subfail$ for the particular failure in question.
	\end{enumerate}
\end{definition}
\nomenclature[B16a]{$\Pi\subCAM$}{Collision-Avoidance-Maneuver policy}%

This effectively removes decision variables $\Vecu\subCAM$ from \cref{eqn: Finite-Time Traj Safety OptProb}, allowing simple numerical integration for satisfaction of the dynamic constraints and a straightforward \emph{a posteriori} verification of the other trajectory constraints (inclusion in $\Xfree$, and satisfaction of constraints $\InequalityConstraintVec$ and $\EqualityConstraintVec$).  This checks if the prescribed CAM, guaranteed to provide a safe escape route, can actually be accomplished in the given failure situation.  The approach is conservative due to the fact that the control law is imposed and not derived; however, the advantage is a greatly simplified optimal control problem with difficult-to-handle constraints relegated to \emph{a posteriori} checks --- exactly identical to the way that steering trajectories are derived and verified during the planning process of sampling-based planning algorithms.  Note that formal definitions of safety require that this be satisfied for all possible failure modes of the spacecraft; we do not avoid the combinatorial explosion of $N\subfail$.  However, each instance of problem \cref{eqn: Finite-Time Traj Safety OptProb} is greatly simplified, and with $\FaultTolerance$ typically at most $3$, the problem remains tractable.  The difficult part, then, lies in computing $\Pi\subCAM$, but this can easily be generated in an offline fashion.  Hence, the strategy should work well for vehicles with difficult, non-convex objective functions and constraints, as is exactly the case for CWH proximity operations.

Note, it is always possible to reduce this approach to the (more-conservative) definition of ``passive safety'' that has traditionally been seen in the literature by choosing some finite horizon $\horizonTime$ and setting $\Vecu\subCAM = \Pi\subCAM\left(\Vecx(t\subfail)\right) = \VecZeros$ for all potential failure times $t\subfail \in \TimeHorizon\subfail$.

\subsection{Safety in CWH Dynamics}
\label{sec: CWH Safety}

We now specialize these ideas to proximity operations under impulsive CWH dynamics.  Because many missions require stringent avoidance (prior to final approach and docking phase, for example), it is quite common for a ``Keep-Out Zone'' (KOZ) $\XsubKOZ$, typically ellipsoidal in shape, to be defined about the target in the CWH frame.
Throughout its approach, the chaser must certify that it will not enter this KOZ under any circumstance up to a specified thruster fault tolerance $\FaultTolerance$.
\nomenclature[AX)]{$\XsubKOZ$}{State space Keep-Out-Zone}%
	Per \cref{def: Finite-Time Traj Safety}, this necessitates a search for a safe invariant set for finite-time escape along with, as outlined by \cref{def: Fault-Tolerant Active Safety}, the definition of an escape policy $\Pi\subCAM$, which we describe next.

\subsubsection{CAM Policy}
\label{subsec: CWH CAM Policy Intro}
	For mission safety following a failure under CWH dynamics, \cref{def: Finite-Time Traj Safety} requires us to find a terminal state in an invariant set $\Xinvariant$ entirely contained within the free state space $\Xfree$. 
As will be motivated, we choose for $\Xinvariant$ the set of circularized orbits whose planar projections lie outside of the radial band spanned by the KOZ.  Circular orbits are stable (assuming Keplerian motion, which neglects perturbations that \emph{can} create unstable orbital changes -- likely not an unreasonable assumption as presumably the \emph{difference} in chaser/target perturbation responses matters more), accessible (given the proximity of the chaser to the target's circular orbit), and passively safe (once reached, provided that they do not intersect the KOZ).  In the planar case, this set of safe circularized orbits can fortunately be identified by inspection; as can be seen in \cref{fig: CWH Target KOZ and Circularization RIC}, we require only that the orbital radius lie outside of the KOZ radial band; otherwise circularization would result in an eventual collision, either in the short-term or after nearly one full synodic period -- a violation of \cref{def: Vehicle Safety}.
Such a region is called a zero-thrust ``Region of Inevitable Collision (RIC),'' which we denote as $\Xric$, whose complement $\Xinvariant$ can be used to terminate planar escape maneuvers. 
For the non-planar case, we can conservatively extend this result by considering only those circular orbits whose planar orbit projection avoids $\Xric$.
See \crefrange{eqn: CWH KOZ}{eqn: CWH Invariant Set} for a mathematical description:
\nomenclature[AX+]{$\Xric$}{State space Region of Inevitable Collision (with zero thrust)}%

\begin{align}
	\label{eqn: CWH KOZ}
	\XsubKOZ 	&= \left\{ \Vecx \suchthat \VecxTranspose \MatEllipsoid \Vecx \geq 1\text{, where } \MatEllipsoid = 
	\diag{\semiaxis_{\posCrossTrack}^{-2}, \semiaxis_{\posInTrack}^{-2}, \semiaxis_{\posOutofPlane}^{-2}, 0, 0, 0}
	\text{, with } \semiaxis_i \text{ representing } \right.
	\\ \nonumber &\qquad \left. \text{the ellipsoidal KOZ semi-axis in the $i$-th LVLH frame axis direction.} \right\}
	\\
	\label{eqn: CWH RIC}
	\Xric 		&= \left\{ \Vecx \suchthat \abs{\posCrossTrack} < \semiaxis_{\posCrossTrack} \text{, } \velCrossTrack = 0 \text{, } \velInTrack = \Vcirc \right\} \supset \XsubKOZ
	\\
	\label{eqn: CWH Invariant Set}
	\Xinvariant &= \left\{ \Vecx \suchthat \abs{\posCrossTrack} \geq \semiaxis_{\posCrossTrack} \text{, } \velCrossTrack = 0 \text{, } \velInTrack = \Vcirc \right\} = \setcomplement{\Xric}
\end{align}
\nomenclature[AE ]{$\MatEllipsoid$}{Ellipsoid matrix}%
\nomenclature[B17b]{$\semiaxis_{\posCrossTrack}, \semiaxis_{\posInTrack}, \semiaxis_{\posOutofPlane}$}{Ellipsoid semi-axes in the radial, in-track, and out-of-plane directions}%

\begin{figure}
	\hfill
	\subcaptionbox{%
		\label{subfig: CWH Planar Zero-Thrust RIC}
		Safe circularization burn zones $\Xinvariant$ for planar CWH dynamics.  \iftoggle{AIAAjournal}{}{Any circularization attempts inside $\Xric$ will result in eventual penetration of the client KOZ, as indicated by the arrows.}
	}[0.46\columnwidth]{%
		\includegraphics[width=0.95\columnwidth]{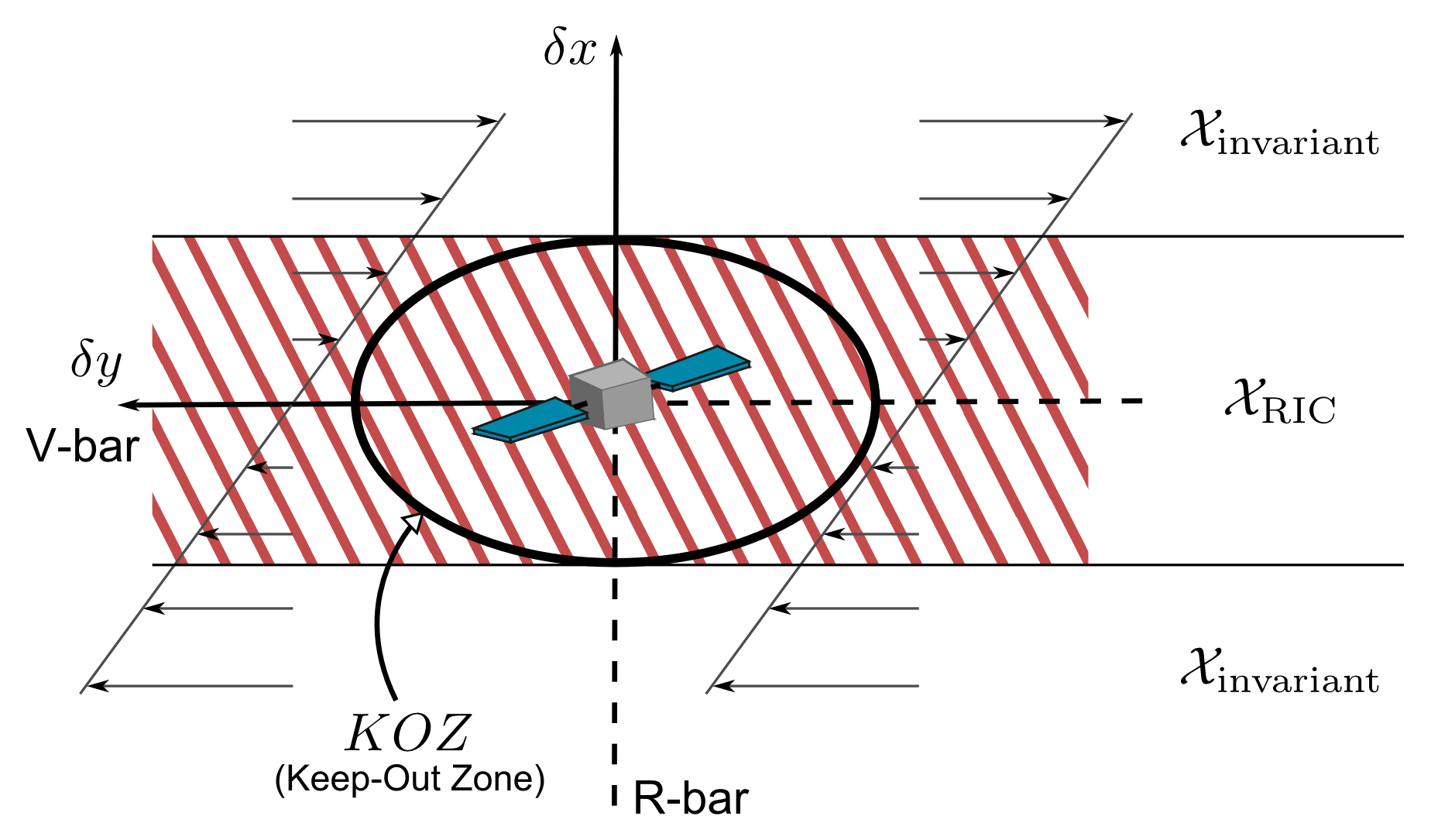}
	}
	\hfill
	\subcaptionbox{%
		\label{subfig: Target KOZ and Circularization RIC}
		Inertial view of the radial band spanned by the KOZ that defines the unsafe RIC.  \iftoggle{AIAAjournal}{}{Its complement shows the invariant positions in $\Xinvariant$ used for safe trajectory escape maneuver targeting.}
	}[0.42\columnwidth]{%
		\includegraphics[width=0.95\columnwidth]{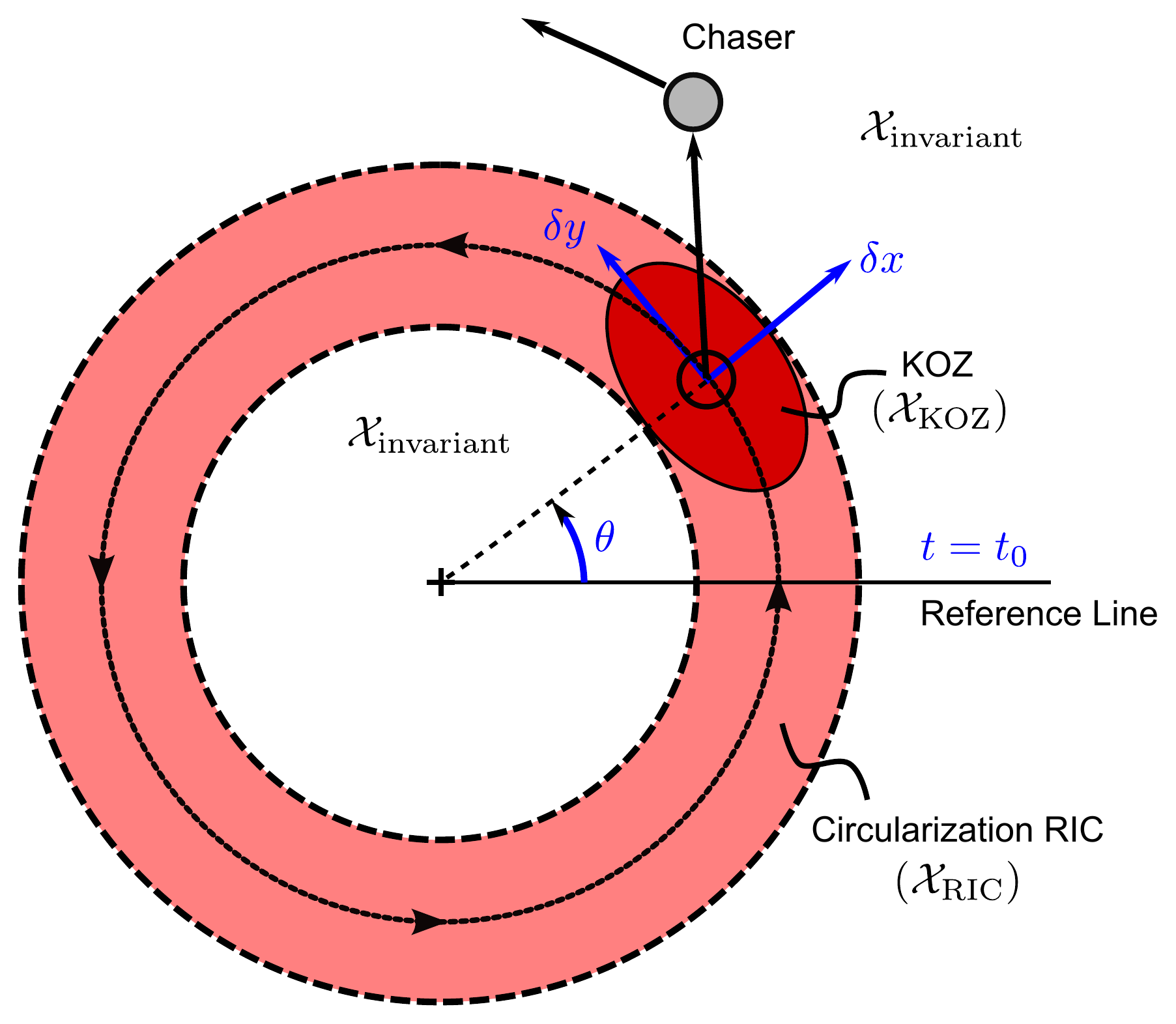}
	}
	\hfill
	\caption{Visualizing the safe and unsafe circularization regions used by the CAM safety policy.}
	\label{fig: CWH Target KOZ and Circularization RIC}
\end{figure}

%
In short, our CAM policy to safely escape from a state $\Vecx$ at which the spacecraft arrives (possibly under failures) at time $t\subfail$ consists of the following:
\vspace*{-0.5em}
\begin{enumerate}
	\setlength{\itemsep}{0pt}
	\item Coast from $\Vecx(t\subfail)$ to some new $\horizonTime > t\subfail$ such that $\Vecx\subCAM(\horizonTime\supminus)$ lies at a position in $\Xinvariant$.
	\item Circularize the orbit
	at $\Vecx\subCAM(\horizonTime)$ such that $\Vecx\subCAM(\horizonTime\supplus) \in \Xinvariant$.
	\item Coast along the new orbit (horizontal drift along the in-track axis in the CWH relative frame) in $\Xinvariant$ until allowed to continue the mission (\eg after approval from ground operators).
\end{enumerate}

\subsubsection{Determining the Circularization Time, $\horizonTime$}
\label{subsec: CWH CAM Policy Optimization}
In the event of a thruster failure at state $\Vecx(t\subfail)$ that requires an emergency CAM, the time $\horizonTime > t\subfail$ at which to attempt a circularization maneuver after coasting from $\Vecx(t\subfail)$ becomes a degree of freedom.  As we intend to maximize the recovery chances of the chaser after a failure, we choose $\horizonTime$ so as to minimize the cost of the circularization burn $\VecDeltaV\subcirc$, whose magnitude we denote $\DeltaV\subcirc$.  Details on the approach, which can be solved analytically, can be found in \cref{appendix: Optimal Circularization}.
\nomenclature[Av,]{$\VecDeltaV\subcirc$}{Circularization burn vector applied during Collision-Avoidance-Maneuvers}%

\subsubsection{Verifying CAM Policy Feasibility}
\label{subsec: CWH CAM Policy Verification}
Once the circularization time $\horizonTime$ is determined, feasibility of the escape trajectory under every possible failure configuration at $\Vecx(t\subfail)$
must be assessed in order to declare a particular CAM as actively-safe.
To show this, the constraints of \cref{sec: Problem} must be evaluated under every combination of ``stuck-off'' thrusters (up to fault tolerance $\FaultTolerance$), \emph{with the exception of KOZ avoidance} as this is embedded into the CAM design process.  Fortunately, given our particular constraints encoded in \cref{eqn: Finite-Time Traj Safety OptProb} and described in \cref{subsec: Traj Constraints} (static in the CWH LVLH frame and independent of the arrival time $t\subfail$), assuming additionally that the position of the target is known \emph{a priori} (fixed at the origin) and that the attitude $\Vecq(t)$ of the chaser is specified as a function of our state $\Vecx(t)$, we may evaluate CAM trajectory feasibility (control allocation feasibility, plume impingement, antenna lobe avoidance, \etc) in an offline fashion.  Better still, we need only evaluate the safety of arriving at $\Vecx$ once; this means the \emph{active} safety of a particular state $\Vecx$ can be cached --- a fact we will make extensive use of in the design of our planning algorithm.

\section{Planning Algorithm and Theoretical Characterization}
\label{sec: Approach}

With the proximity operations scenario established, we are now in position to describe our approach.  As previously described, the constraints that must be satisfied in \cref{eqn: CWH Optimal Steering} are diverse, complex, and difficult to satisfy numerically.  In this section, we propose a guidance algorithm to solve this problem, followed by a detailed proof of its optimality with regard to the 2-norm \fuelCost metric $\CostFcn$ under impulsive CWH dynamics.  As will be seen, the proof relies on an understanding of: (i) the steering connections between sampled points assuming no obstacles or other trajectory constraints, and (ii) the nearest-neighbors or \emph{reachable} states from a given state.  We hence start by characterizing these two concepts, in \cref{subsec: steering,subsec: Reachability Sets} respectively.  We then proceed to the algorithm presentation (\cref{subsec: Algorithm}) and its theoretical characterization (\cref{subsec: Theoretical Characterization}).

\subsection{The Steering Problem}
\label{subsec: steering}%
\begingroup%
\setlength{\arraycolsep}{3pt}%
In this section, we consider the \emph{unconstrained} minimal \fuel 2-point boundary value problem (2PBVP) or ``steering problem'' between an initial state $\Vecx\subnaught$ and a final state $\Vecx\subf$ within the CWH dynamics model.  Solutions to these steering problems provide the local building blocks from which we construct solutions to the more complicated problem formulation in \cref{eqn: CWH Optimal Steering}.  Steering solutions serve two main purposes: (i) they represent a class of short-horizon controlled trajectories that are filtered online for constraint satisfaction and efficiently strung together into a state space-spanning graph (\ie a tree or roadmap), and (ii) the costs of steering trajectories are used to inform the graph construction process by identifying the unconstrained ``nearest neighbors'' as edge candidates.  Because these problems can be expressed independently of the arrival time $t\subnaught$ (as will be shown), our solution algorithm does not need to solve these problems \emph{online}; the solutions between every pair of samples can be precomputed and stored prior to receiving a motion query.  Hence the 2PBVP presented here need not be solved quickly.  However, we mention techniques here for speed-ups due to the reliance of our smoothing algorithm (\cref{alg: CWH Traj Smoothing}) on a fast solution method.
\nomenclature[Ax-]{$\Vecx\subnaught$}{Initial state for steering problem}%
\nomenclature[Ax.]{$\Vecx\subf$}{Final state for steering problem}%
\nomenclature[At0]{$t\subnaught$}{Initial time for steering problem}%
\nomenclature[Atf]{$\tf$}{Final time for steering problem}%

Substituting our boundary conditions into \cref{eqn: CWH Impulsive Solution}, evaluating at $t = \tf$, and rearranging, we seek a stacked burn vector $\MatDeltaV$ such that:
\begin{equation}
	\label{eqn: CWH 2PBVP}
	\MatPhi_v\left(\tf,\left\{\BurnTime_i\right\}_i\right) \MatDeltaV = \Vecx\subf - \MatPhi\left(\tf,\ti\right) \Vecx\subnaught,
\end{equation}
for some number $\NBurns$ of burn times $\BurnTime_i \in \closedinterval{\ti}{\tf}$.  Formulating this as an optimal control problem that minimizes our 2-norm cost functional (as a proxy for the actual \fuel consumption, as described in \cref{subsec: Cost Functional}), we wish to solve:
\begin{align}
	\removeequationpadding
	\label{eqn: CWH Optimal 2PBVP}
	\begin{optproblem}
		\given{\text{Initial state } \Vecx\subnaught, \text{final state } \Vecx\subf, \text{burn magnitude bound } \VecDeltaVMag\submax,}
		& & \mathrlap{\text{and maneuver duration bound }  \ManeuverDuration\submax} \\
		\minimize[\VecDeltaV_i,\BurnTime_i,\tf,\NBurns]{\finiteseries[i=1]{\NBurns} \norm[\pvalue]{\VecDeltaV_i}}
		\subjectto
		\constraint[\text{Dynamics/Boundary Conditions}]{ \MatPhi_v\left(\tf,\left\{\BurnTime_i\right\}_i\right) \MatDeltaV = \Vecx\subf - \MatPhi\left(\tf,\ti\right) \Vecx\subnaught.}
		\constraint[\text{Maneuver Duration Bounds}]{ 0 \leq \tf - \ti \leq \ManeuverDuration\submax }
		\constraint[\text{Burn Time Bounds}]{ \ti \leq \BurnTime_i \leq \tf \quad\quad\quad\quad\!\! \text{for burns } i}
		\constraint[\text{Burn Magnitude Bounds}]{ 
			\norm[\pvalue]{\VecDeltaV_i} \leq \VecDeltaVMag\submax \quad \text{for burns } i}
	\end{optproblem}
\end{align}
Notice that this is a relaxed version of the original problem presented as \cref{eqn: CWH Optimal Steering}, with only its boundary conditions, dynamic constraints, and control norm bound.  As it stands, due to the nonlinearity of the dynamics with respect to $\BurnTime_i$, $\tf$ and $\NBurns$, \cref{eqn: CWH Optimal 2PBVP} is non-convex and inherently difficult to solve.  However, we can make the problem tractable if we make a few assumptions.  
Given that we plan to string many steering trajectories together to form our overall solution, let us ensure they represent the most primitive building blocks possible such that their concatenation will adequately represent any arbitrary trajectory.  Set $\NBurns = 2$ (the smallest number of burns required to transfer between any pair of arbitrary states, as it makes $\MatPhi_v\left(\tf,\left\{\BurnTime_i\right\}_i\right)$ square) and select burn times $\BurnTime_1 = \ti$ and $\BurnTime_2 = \tf$ (which automatically satisfy our burn time bounds).  This leaves $\VecDeltaV_1 \in {\reals}^{\Dimension/2}$ (an intercept burn applied just after $\Vecx\subnaught$ at time $\ti$), $\VecDeltaV_2 \in {\reals}^{\Dimension/2}$ (a rendezvous burn applied just before $\Vecx\subf$ at time $\tf$), and $\tf$ as our only remaining decision variables.  If we conduct a search for $\tf\supopt \in \closedinterval{\ti}{\ti + \ManeuverDuration\submax}$, the relaxed-2PBVP can now be solved iteratively as a relatively simple bounded one-dimensional nonlinear minimization problem, where at each iteration one computes:
\begin{equation*}	
	\MatDeltaV\supopt\left(\tf\right) = \MatPhiInv_v\left(\tf, \{\ti, \tf\}\right) \left(\Vecx\subf - \MatPhi\left(\tf,\ti\right) \Vecx\subnaught\right),
\end{equation*}
where the argument $\tf$ is shown for $\MatDeltaV\supopt$ to highlight its dependence.  By uniqueness of the matrix inverse (provided $\MatPhiInv_v$ is non-singular, discussed below), we need only check that the resulting impulses $\VecDeltaV_i\supopt\left(\tf\right)$ satisfy the magnitude bound to declare the solution to an iteration feasible.  Notice that because $\MatPhi$ and $\MatPhiInv_v$ depend only on the difference between $\tf$ and $\ti$, we can equivalently search over various $\tf - \ti \in \closedinterval{0}{\ManeuverDuration\submax}$ instead, using the expression:
\begin{equation}
	\label{eqn: CWH Optimal 2PBVP with N=2}
	\MatDeltaV\supopt\left(\tf - \ti\right) = \MatPhiInv_v\left(\tf - \ti, \{0, \tf - \ti\}\right) \left(\Vecx\subf - \MatPhi\left(\tf - \ti, 0\right) \Vecx\subnaught\right),
\end{equation}
which reveals that our impulsive steering problem depends only on the maneuver duration $\ManeuverDuration = \tf - \ti$ (provided $\Vecx\subf$ and $\Vecx\subnaught$ are given).  This will be indispensable for precomputation, as it allows steering trajectories to be generated and stored \emph{offline}.
%
%
Regarding singularities, our steering solution $\MatDeltaV\supopt = \argmin_{\tf}{\MatDeltaV\supopt\left(\tf - \ti\right)}$ requires that $\MatPhi_v$ be invertible, \ie that $\left(\tf - \BurnTime_1\right) - \left(\tf - \BurnTime_2\right) = \tf - \ti$
avoids certain values (such as zero and certain values longer than one period \cite{KA-SRV-PG-JH-LB:09}, including orbital period multiples) which we achieve by restricting $\ManeuverDuration\submax$ to be shorter than one orbital period.  To handle $\tf - \ti = 0$ exactly, note a solution to the 2PBVP exists if and only if $\Vecx\subnaught$ and $\Vecx\subf$ differ in velocity only; in such cases, we take this velocity difference as $\VecDeltaV_2\supopt$ (with $\VecDeltaV_1\supopt = \VecZeros$) to be the solution.
%

\subsection{Reachability Sets}
\label{subsec: Reachability Sets}
In keeping with \cref{eqn: CWH Optimal 2PBVP with N=2}, since $\MatDeltaV\supopt = \argmin_{\ManeuverDuration}{\MatDeltaV\supopt\left(\ManeuverDuration\right)}$ only depends on $\Vecx\subf$ and $\Vecx\subnaught$, we henceforth refer to the cost of a steering trajectory by the notation 
$\CostFcn(\Vecx\subnaught, \Vecx\subf)$.
We then define the forward \emph{reachability set} from a given state $\Vecx\subnaught$ as follows:
\nomenclature[AJt]{$\CostThreshold$}{Fuel cost threshold used in reachability sets}%
%
\begin{definition}[Forward Reachable Set]
	The forward reachable set $\ReachableSet$ from state $\Vecx\subnaught$ is the set of all states $\Vecx\subf$ that can be reached from $\Vecx\subnaught$ with a cost $\CostFcn\left(\Vecx\subnaught, \Vecx\subf\right)$ below a given cost threshold $\CostThreshold$, \ie%
	\begin{equation*}
		\ReachableSet\left(\Vecx\subnaught,  \CostThreshold \right) \triangleq \left\{ \Vecx\subf \in \Xspace \suchthat \CostFcn\left(\Vecx\subnaught,\Vecx\subf\right) < \CostThreshold \right\}.
	\end{equation*}%
	\label{def:reach}%
\end{definition}%
\nomenclature[AR#]{$\ReachableSet$}{Forward-reachable set}%
Recall from \cref{eqn: CWH Optimal 2PBVP with N=2} in \cref{subsec: steering} that the steering cost may be written as:
\begin{equation}
	\label{eqn: CWH Cost Function}
	\CostFcn\left(\Vecx\subnaught, \Vecx\subf\right) = \norm{\VecDeltaV_1} + \norm{\VecDeltaV_2} = \norm{\MatDeltaVSlicing_1 \MatDeltaV} + \norm{\MatDeltaVSlicing_2 \MatDeltaV}
\end{equation}
where $\MatDeltaVSlicing_1 = [\,\Identity_{d/2 \times d/2}\ \MatZeros_{d/2 \times d/2}\,]$, $\MatDeltaVSlicing_2 = [\,\MatZeros_{d/2 \times d/2}\ \Identity_{d/2 \times d/2}\,]$, and $\MatDeltaV$ is given by:
\nomenclature[AS ]{$\MatDeltaVSlicing$}{Slicing matrix used to isolate individual impulse vectors from $\MatDeltaV$}%
\begin{align*}
	\MatDeltaV\left(\Vecx\subnaught, \Vecx\subf\right) = \begin{bmatrix} \VecDeltaV_1 \\ \VecDeltaV_2 \end{bmatrix} = \MatPhiInv_v\left(\tf,\{\ti,\tf\}\right) \left(\Vecx\subf - \MatPhi(\tf,\ti) \Vecx\subnaught\right).
\end{align*}
The cost function $\CostFcn\left(\Vecx\subnaught, \Vecx\subf\right)$ is difficult to gain insight on directly; however, as we shall see, we can work with its bounds much more easily. 

\begin{restatable}[Fuel Burn Cost Bounds]{lemma}{LemmaFuelCostBounds}
	\label{lem: CWH Cost Bounds}
	For the cost function in \cref{eqn: CWH Cost Function}, we have the following upper and lower bounds:%
	\begin{equation*}%
		\norm{\MatDeltaV} \leq \CostFcn\left(\Vecx\subnaught, \Vecx\subf\right) \leq \sqrt{2} \norm{\MatDeltaV}.
	\end{equation*}%
\end{restatable}
\begin{proof}
	For the proof, see \cref{appendix: Useful Results}.
\end{proof}
Now, observe that $\norm{\MatDeltaV} = \sqrt{\transpose{(\Vecx\subf - \MatPhi(\tf,\ti) \Vecx\subnaught)} \GramianInv (\Vecx\subf - \MatPhi(\tf,\ti) \Vecx\subnaught)}$ where $\GramianInv = \MatPhiInvTranspose_v \MatPhiInv_v$, \ie the expression for an ellipsoid $\EllipsoidalSet\left(\Vecx\subf\right)$ resolved in the LVLH frame with matrix $\GramianInv$ and center $\MatPhi(\tf, \ti) \Vecx\subnaught$ (the state $\ManeuverDuration = \tf - \ti$ time units ahead of $\Vecx\subnaught$ along its coasting arc).  Combined with \cref{lem: CWH Cost Bounds}, we see that for a fixed maneuver time $\ManeuverDuration$ and \fuel cost threshold $\CostThreshold$, the spacecraft at $\Vecx\subnaught$ can reach all states inside an area under-approximated by an ellipsoid with matrix $\sidefrac{\GramianInv}{\CostThreshold^2}$ and over-approximated by an ellipsoid of matrix $\sidefrac{\sqrt{2}\GramianInv}{\CostThreshold^2}$.  The forward reachable set for impulsive CWH dynamics under the 2-norm metric is therefore bounded by the union over all maneuver times of these under- and over-approximating ellipsoidal sets, respectively.  See \cref{fig: CWH Reachability} for visualization.
\nomenclature[AG ]{$\GramianInv$}{Gramian inverse with respect to $\MatPhi_v$}%
\nomenclature[AE ]{$\EllipsoidalSet$}{Ellipsoidal set}%

\begin{figure}
	\hfill
	\subcaptionbox{%
		The set of reachable positions $\delta\Vecr\subf$ within duration $\ManeuverDuration\submax$ and \fuel cost $\DeltaV\submax$.
	}[0.45\columnwidth]{%
		\includegraphics[width=0.93\columnwidth]{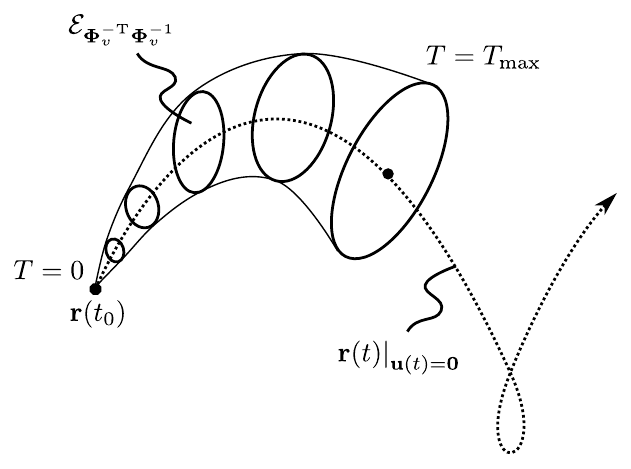}
	}
	\hfill
	\subcaptionbox{%
		The set of reachable velocities $\delta\Vecv\subf$ within duration $\ManeuverDuration\submax$ and \fuel cost $\DeltaV\submax$.
	}[0.45\columnwidth]{%
		\includegraphics[width=\columnwidth]{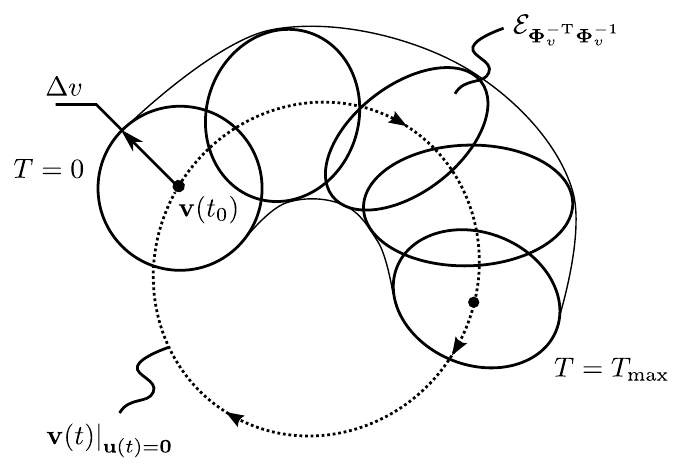}
	}
	\hfill
	\caption{Bounds on reachability sets from initial state $\Vecx(\ti)$ under \fuel cost threshold $\CostThreshold$.  \iftoggle{AIAAjournal}{}{For cost measured as a sum of $\VecDeltaV$ 2-norms, these are bounded by unions of ellipsoidal balls in position-velocity (phase) space over the set of all permissible maneuver durations $\ManeuverDuration \in \interval{0}{\ManeuverDuration\submax}$.  For the special case of $\ManeuverDuration = 0$, the reachable set of states is a 3-dimensional velocity ball embedded in ${\reals}^6$ (with volume 0) corresponding to states with position $\delta\Vecr\subf = \delta\Vecr(\ti)$ and velocity $\delta\Vecv\subf \in \ball{\delta\Vecv(\ti)}{\DeltaV\submax}$.}}
	\label{fig: CWH Reachability}
\end{figure}

\subsection{Algorithm}
\label{subsec: Algorithm}
As mentioned in \cref{sec: Intro}, we apply a modified version of the Fast Marching Tree (\FMTstar) sampling-based planning algorithm to solve the problem in \cref{eqn: CWH Optimal Steering}.  Sampling-based planning \cite{SML-JJK:01, SML:06} 
essentially breaks down a continuous trajectory optimization problem into a series of relaxed, local steering problems (as in \cref{subsec: steering}) between intermediate waypoints (called \emph{samples}) before piecing them together to form a global solution to the original problem.  This framework can yield significant computational benefits if: (i) the relaxed subproblems are simple enough, and (ii) the \emph{a posteriori} evaluation of trajectory constraints is fast compared to a single solution of the full-scale problem.  Furthermore, provided samples are sufficiently dense in the free state-space $\Xfree$ and graph exploration is spatially symmetric, sampling-based planners can closely approximate global optima without fear of convergence to local minima.  Though many candidate planners could be used here, we rely on the asymptotically-optimal (AO) \FMTstar algorithm for its efficiency (see \cite{LJ-ES-AC-ea:15} for details on the advantages of \FMTstar over its state-of-the-art counterparts) and its compatibility with \emph{deterministic} (as opposed to random) sampling sequences \cite{LJ-BI-MP:15}, which leads to a number of algorithmic simplifications (including use of offline knowledge).

The \FMTstar algorithm, tailored to our application, is presented as \cref{alg: FMTstar} (we shall henceforth refer to our modified version of \FMTstar as simply \FMTstar, for brevity).  \FMTstar efficiently expands a tree of feasible trajectories from an initial state $\Vecx\subinit$ to a goal state $\Vecx\subgoal$ around nearby obstacles.  It begins by taking a set of samples distributed in the free state space $\Xfree$ using the \Call{SampleFree}{} routine, which restricts state sampling to actively-safe, collision-free samples (which lie outside of $\Xobs$ and have access to a safe Collision-Avoidance Maneuver (CAM) as described in \cref{sec: CWH Safety}).  In our implementation, we assume samples are taken using a particular deterministic, low-dispersion sequence called the Halton sequence \cite{JHH:60}, though any deterministic, \emph{low-dispersion} sampling sequence may be used \cite{LJ-BI-MP:15}.  Selecting $\Vecx\subinit$ first for further expansion as the minimum cost-to-come node $\Vecz$, the algorithm then proceeds to look at reachable samples or ``neighbors'' (samples that can be reached with less than a given \fuel cost threshold $\CostThreshold$, as described in the previous subsection) and attempts connections (using \Call{Steer}{}) to those with cheapest cost-to-come back to the tree.  The cost threshold $\CostThreshold$ is a free parameter whose value can have a significant effect on performance; see \cref{th: fmt asymptotic optimality} for a theoretical characterization and \cref{sec: Experiments} for a representative numerical trade study. Those trajectories satisfying the constraints of \cref{eqn: CWH Optimal Steering}, as determined by \Call{CollisionFree}{}, are saved.  As feasible connections are made, the algorithm relies on adding and removing nodes (saved waypoint states) from three sets: a set of unexplored samples $\Unexplored$ not yet connected to the tree, a frontier $\Frontier$ of nodes likely to make efficient connections to unexplored neighbors, and an interior $\TreeInterior$ of nodes that are no longer useful for exploring the state space $\Xspace$.  Details on \FMTstar can be found in its original work \cite{LJ-ES-AC-ea:15}.
\nomenclature[Az ]{$\Vecz$}{Minimum cost-to-come node in an \FMTstar frontier set}%
\nomenclature[AV(]{$\Unexplored$}{Unexplored set of samples}%
\nomenclature[AV']{$\Frontier$}{Frontier set of \FMTstar nodes}%
\nomenclature[AV#]{$\TreeInterior$}{Interior set of expanded \FMTstar nodes}%
%

To make \FMTstar amenable to a real-time implementation, we consider an online-offline approach that relegates as much computation as possible to a pre-processing phase.  To be specific, the sample set $\SampleSet$ (\cref{line: Sample Xfree}), nearest-neighbor sets (used in \cref{line: Expand Tree Loop Begin,line: Nearest Neighbor Selection}), and steering trajectory solutions (\cref{line: Locally-Optimal Steering}) may be entirely pre-processed, assuming the planning problem satisfies the following conditions:
\begin{enumerate}
	\item the state space $\Xspace$ is known \emph{a priori}, as is typical for most LEO missions (a luxury we do not generally have for the obstacle space $\Xobs$, which must be identified online using onboard sensors once the spacecraft arrives at $\Xspace$), \label{item: online-offline known state space}
	\item steering solutions are independent of sample arrival times $t\subnaught$, as we show in \cref{subsec: steering}.\label{item: online-offline arrival time independence}
\end{enumerate}
Here \cref{item: online-offline known state space} allows samples to be precomputed, while \cref{item: online-offline arrival time independence} enables steering trajectories to be stored onboard or uplinked from the ground up to the spacecraft, since their values remain relevant regardless of the times at which the spacecraft actually follows them during the mission.  This leaves only collision-checking, graph construction, and termination checks as parts of the online phase, greatly improving the online run time and leaving the more intensive work to offline resources where running time is less important.  
This breakdown into online and offline components (inspired by \cite{RA-MP:16}) is a valuable technique for imbuing kinodynamic motion planning problems with real-time online solvability using fast batch-planners like \FMTstar.


\begin{figure}
	\begin{algorithm}[H]
		\caption{The Fast Marching Tree Algorithm (\FMT). Computes a minimal-cost trajectory from an initial state $\Vecx(\protect\ti) = \Vecx\subinit$ to a target state $\Vecx\subgoal$ through a fixed number $\nSamples$ of samples ${\SampleSet}$.}
		\label{alg: FMTstar}
		\begin{algorithmic}[1]
			\State Add $\Vecx\subinit$ to the root of the tree $\Tree$, as a member of the frontier set $\Frontier$ \label{line: Tree Initialization}
			\State Generate samples ${\SampleSet} \leftarrow \Call{SampleFree}{{\Xspace}, \protect\nSamples, \protect\ti}$
			and add them to the unexplored set $\Unexplored$ \label{line: Sample Xfree}
			\State Set the minimum cost-to-come node in the frontier set as $\Vecz \leftarrow \Vecx\subinit$
			\While{\textbf{true}}
				\For{each neighbor $\Vecx$ of $\Vecz$ in $\Unexplored$} \label{line: Expand Tree Loop Begin}
					\State Find the neighbor $\Vecx\submin$ in $\Frontier$ of cheapest cost-to-go to $\Vecx$ \label{line: Nearest Neighbor Selection}
					\State Compute the trajectory between them as $\left[\Vecx\left(t\right), \Vecu\left(t\right), t\right] \leftarrow \Call{Steer}{\Vecx\submin, \Vecx}$ (\emph{see \cref{subsec: steering}}) \label{line: Locally-Optimal Steering}
					\If{\Call{CollisionFree}{$\Vecx\left(t\right), \Vecu\left(t\right), t$}}
						\State Add the trajectory from $\Vecx\submin$ to $\Vecx$ to tree $\Tree$ \label{line: Expand Tree Loop End}
					\EndIf
				\EndFor
				\State Remove all $\Vecx$ from the unexplored set $\Unexplored$
				\State Add any new connections $\Vecx$ to the frontier $\Frontier$
				\State Remove $\Vecz$ from the frontier $\Frontier$ and add it to $\TreeInterior$
				\If{$\Frontier$ is empty}
					\State \Return Failure
				\EndIf
				\State Reassign $\Vecz$ as the node in $\Frontier$ with smallest cost-to-come from the root ($\Vecx\subinit$)
				\If{$\Vecz$ is in the goal region $\Xgoal$}
					\State \Return Success, and the unique trajectory from the root ($\Vecx\subinit$) to $\Vecz$
				\EndIf
			\EndWhile
		\end{algorithmic}
	\end{algorithm}
\end{figure}
\nomenclature[An)]{$\nSamples$}{Number of samples}%
\nomenclature[AT ]{$\Tree$}{Tree data structure used by the \FMTstar algorithm}%
\nomenclature[AS ]{$\SampleSet$}{Sample set}%

\subsection{Theoretical Characterization}
\label{subsec: Theoretical Characterization}
It remains to show that \FMTstar provides similar asymptotic optimality and convergence rate guarantees under the 2-norm \fuelCost metric and impulsive CWH dynamics (which enter into \cref{alg: FMTstar} under \crefrange{line: Nearest Neighbor Selection}{line: Locally-Optimal Steering}), as it does for kinematic (straight-line path planning) problems \cite{LJ-ES-AC-ea:15}.
\input{optimality}

\nomenclature[AD ]{$\DispersionFcn\left(\cdot\right)$}{Dispersion of a set of points}%
\nomenclature[AO ]{$\BigO{\cdot}$}{Asymptotic upper-bound}%
\nomenclature[Ao ]{$\LittleO{\cdot}$}{Strict asymptotic upper-bound}%
\nomenclature[B24b]{$\LittleOmega{\cdot}$}{Strict asymptotic lower-bound}%
\nomenclature[AN#]{$\naturals$}{The field of natural numbers}%
\nomenclature[B03b]{$\gamma$}{Positive constant in $\DispersionFcn\left(\SampleSet\right)$ upper bound}%
\nomenclature[B04b]{$\delta$}{Positive constant defining a trajectory's minimum clearance from $\Xobs$}%
\nomenclature[AJ ]{$\CostFcn\supopt$}{Optimal solution cost}%
\nomenclature[AJn]{$\CostFcn_\nSamples$}{Solution cost using $\nSamples$ samples}%
\nomenclature[Ay ]{$\Vecy$}{Sample waypoint}%
\nomenclature[B05b]{$\epsilon$}{Small multiplicative factor between the optimal and \FMTstar solution costs}%

\section{Trajectory Smoothing}
\label{sec: Smoothing/Robustness}
Due to the discreteness caused by using a \emph{finite} number of samples, sampling-based solutions will necessarily be approximations to true optima.  In an effort to compensate for this limitation, we offer in this section two techniques to improve the quality of solutions returned by our planner from \cref{subsec: Algorithm}.  We first describe a straightforward method for reducing the sum of $\DeltaV$-vector magnitudes along concatenated sequences of edge trajectories that can also be used to improve the search for \fuel-efficient trajectories in the feasible state space $\Xfree$.  We then follow with a fast post-processing algorithm for further reducing \fuel cost after a solution has been reported.

The first technique removes unnecessary $\DeltaV$-vectors that occur when joining sub-trajectories (edges) in the planning graph.  Consider merging two edges at a node with position $\delta \Vecr\left(t\right)$ and velocity $\delta \Vecv\left(t\right)$ as in \cref{subfig: CWH Built-In Smoothing}.  A naive concatenation would retain both $\VecDeltaV_2\left(t\supminus\right)$ (the rendezvous burn added to the incoming velocity $\Vecv\left(t\supminus\right)$) and $\VecDeltaV_1\left(t\right)$ (the intercept burn used to achieve the outgoing velocity $\Vecv\left(t\supplus\right)$) individually within the combined control trajectory.  Yet, because these impulses occur at the same time, a more realistic approach should merge them into a single net $\DeltaV$-vector $\VecDeltaV\subnet\left(t\supminus\right)$.  By the triangle inequality, we have that:
\nomenclature[At-]{$t\supminus$}{Time immediately before an impulse at time $t$}%
\nomenclature[At+]{$t\supplus$}{Time immediately after an impulse at time $t$}%
\begin{align*}
	\norm{\VecDeltaV_{\mathrm{net}}\left(t\supminus\right)}
		&= \norm{\VecDeltaV_2\left(t\supminus\right) + \VecDeltaV_1\left(t\right)}
		\leq \norm{\VecDeltaV_2\left(t\supminus\right)} + \norm{\VecDeltaV_1\left(t\right)}.
\end{align*}
Hence, merging edges in this way guarantees $\DeltaV$ savings for solution trajectories under the 2-norm \fuel metric.  Furthermore, incorporating net $\DeltaV$'s into the cost-to-come during graph construction can make exploration of the search space more efficient; the cost-to-come $\CostToCome(\Vecz)$ for a given node $\Vecz$ would then reflect the cost to rendezvous with $\Vecz$ from $\Vecx\subinit$ through a series of intermediate intercepts rather than a series of rendezvous maneuvers (as a trajectory designer might normally expect).  Note, on the other hand, that two edges as in \cref{subfig: CWH Built-In Smoothing} that are merged in this fashion no longer achieve velocity $\Vecv\left(t\right)$; state $\Vecx\left(t\right)$ is skipped altogether.  This can be problematic for our active safety policy from \cref{sec: CWH Safety} for states along the incoming edge which relies on rendezvousing with the endpoint $\Vecx = \left[\,\Vecr(t)\ \Vecv(t)\,\right]$ \emph{exactly} before executing our one-burn escape maneuver.  To compensate for this, care must be taken to ensure that the burn $\VecDeltaV_2\left(t\supminus\right)$ that is eliminated during merging is appropriately appended to the front of the escape control trajectory and verified for all possible failure configurations.  Hence we see the price of smoothing in this way is that our original one-burn policy now requires an extra burn, which may not be desirable in some applications. 

\begin{figure}
	\subcaptionbox{%
		\label{subfig: CWH Built-In Smoothing}
		Smoothing during graph construction (merges $\DeltaV$-vectors at edge endpoints).
	}[0.4\columnwidth]{%
		\includegraphics[width=0.9\textwidth]{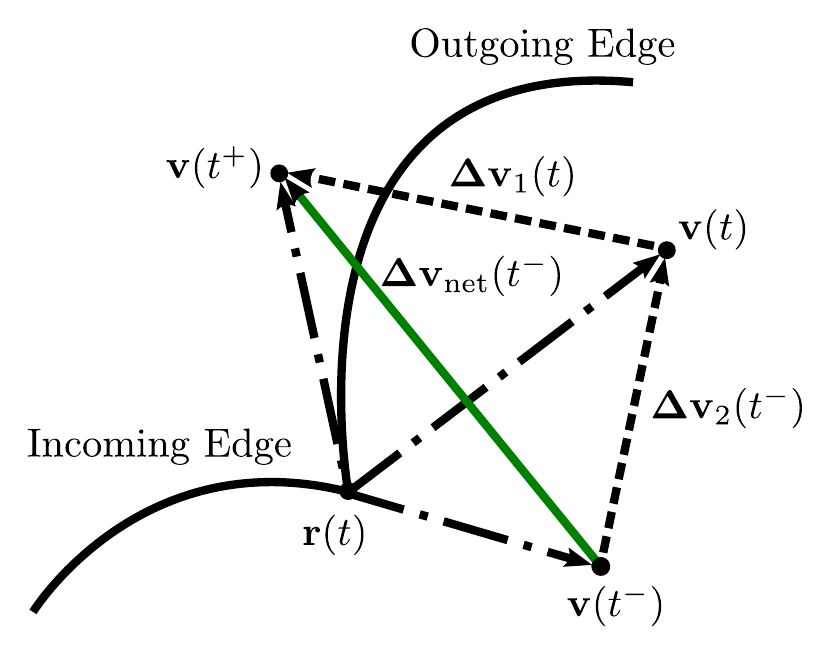}
	}
	\hspace{1em}
	\subcaptionbox{%
		\label{subfig: CWH Post-Processing Smoothing}
		Smoothing during post-processing (see \cref{alg: CWH Traj Smoothing}).
	}[0.4\columnwidth]{%
		\includegraphics[width=0.9\textwidth]{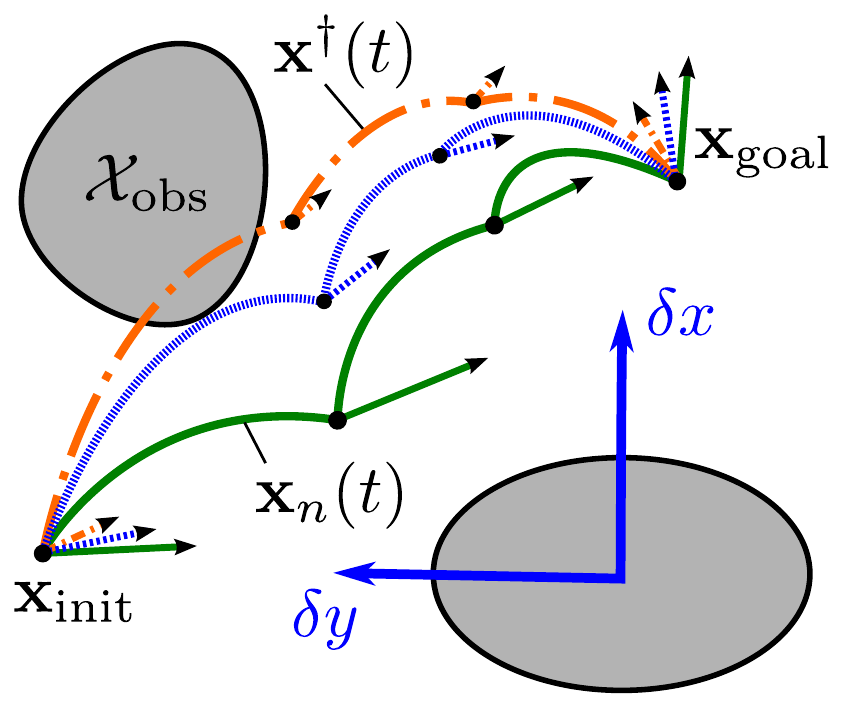}
	}
	\caption{Improving sampling-based solutions under minimal-\fuel impulsive dynamics.  \iftoggle{AIAAjournal}{}{\Cref{subfig: CWH Built-In Smoothing} can be used to merge $\DeltaV$-vectors between edge endpoints during and after graph construction, while \cref{subfig: CWH Post-Processing Smoothing} illustrates the post-processing smoothing algorithm given in \cref{alg: CWH Traj Smoothing} (the original trajectory $\Vecx_\nSamples(t)$ is solid, the approximate unconstrained optimum $\Vecx^\dagger(t)$ is dash-dotted, and the resulting ``smoothed'' trajectory derived from their combination is shown dashed).}}
	\label{fig: CWH Planning with Smoothing}
\end{figure}
%
The second technique attempts to reduce solution cost by adjusting the magnitudes of $\DeltaV$-vectors in the trajectory returned by \FMTstar, denoted by $\Vecx_\nSamples(t)$ with associated stacked impulse vector $\MatDeltaV_\nSamples$.  By relaxing \FMTstar's constraint to pass through state samples, strong cost improvements may be gained.  The main idea is to deform our low-cost, feasible solution $\Vecx_\nSamples(t)$ as much as possible towards the unconstrained minimum-\fuel solution $\Vecx\supopt(t)$ between $\Vecx\subinit$ and $\Vecx\subgoal$, as determined by the 2-point Boundary Value Problem (\cref{eqn: CWH Optimal 2PBVP}) solution from \cref{subsec: steering} (in other words, use a homotopic transformation from $\Vecx_\nSamples(t)$ to $\Vecx\supopt(t)$).  However, a naive attempt to solve \cref{eqn: CWH Optimal 2PBVP} in its full generality would be too time-consuming to be useful, and would threaten the real-time capability of our approach.  Assuming our sampling-based trajectory is near-optimal (or at least, in a low-cost solution homotopy), we can relax \cref{eqn: CWH Optimal 2PBVP} by keeping the number of burns $\NBurns$, end time $\tf \coloneqq t\subfinal$, and burn times $\BurnTime_i$ fixed from our planning solution, and solve for an approximate unconstrained minimum-\fuel solution $\MatDeltaV^\dagger$ with associated state trajectory $\Vecx^\dagger(t)$ via:
\nomenclature[SUP01]{$\ast$}{Related to the optimal solution}%
\nomenclature[SUBn]{$\nSamples$}{Related to the \FMTstar solution using $\nSamples$ samples}%
\nomenclature[SUP04]{$\dagger$}{Related to the approximate unconstrained minimal-\fuel solution}%
\begin{equation}
	\begin{optproblem}
		\minimize[\VecDeltaV_i]{\finiteseries[i=1]{\NBurns} \norm[\pvalue]{\VecDeltaV_i}}
		\subjectto
		\constraint[\text{Dynamics/Boundary Conditions}]{ \MatPhi_v\left(t\subfinal,\left\{\BurnTime_i\right\}_i\right) \MatDeltaV = \Vecx\subgoal - \MatPhi\left(t\subfinal,t\subinit\right) \Vecx\subinit }
		\constraint[\text{Burn Magnitude Bounds}]{ \norm[\pvalue]{\VecDeltaV_i} \leq \VecDeltaVMag\submax \quad\text{for all burns } i }
	\end{optproblem}
	\label{eqn: CWH Solution Used for Traj Smoothing}
\end{equation}
(see \cref{sec:sysDyn} for definitions).  It can be shown that \cref{eqn: CWH Solution Used for Traj Smoothing} is a second-order cone program (SOCP), and hence quickly solved using standard convex solvers.
%
As the following proof shows explicitly, we can safely deform the trajectory $\Vecx_\nSamples(t)$ towards $\Vecx^\dagger(t)$ without violating our dynamics and boundary conditions if we use a convex combination of our two control trajectories $\MatDeltaV_\nSamples$ and $\MatDeltaV^\dagger$.  This follows from the principle of superposition, given that the CWH equations are Linear, Time-Invariant (LTI), and the fact that both solutions already satisfy the boundary conditions.
\nomenclature[Ax-]{$\Vecx\supopt$}{Optimal solution state trajectory}%
\nomenclature[Ax/]{$\Vecx_\nSamples$}{Solution state trajectory using $\nSamples$ samples}%
\nomenclature[Ax-]{$\Vecx^\dagger$}{Unconstrained minimal-\fuel state trajectory}%
\nomenclature[AV)]{$\MatDeltaV\supopt$}{Optimal stacked $\DeltaV$-vector}%
\nomenclature[AV)]{$\MatDeltaV_\nSamples$}{Solution stacked $\DeltaV$-vector using $\nSamples$ samples}%
\nomenclature[AV)]{$\MatDeltaV^\dagger$}{Unconstrained minimal-\fuel stacked $\DeltaV$-vector}%
\begin{theorem}[Dynamic Feasibility of CWH Trajectory Smoothing]
	Suppose $\Vecx_\nSamples(t)$ and $\Vecx^\dagger(t)$ with respective control vectors $\MatDeltaV_\nSamples$ and $\MatDeltaV^\dagger$ are two state trajectories which satisfy the impulsive CWH steering problem \cref{eqn: CWH 2PBVP} between states $\Vecx\subinit$ and $\Vecx\subgoal$.  Then the trajectory $\Vecx(t)$ generated by the convex combination of $\MatDeltaV_\nSamples$ and $\MatDeltaV^\dagger$ is itself a convex combination of $\Vecx_\nSamples(t)$ and $\Vecx^\dagger(t)$, and hence also satisfies \cref{eqn: CWH 2PBVP}.
\end{theorem}
\begin{proof}
	Let $\MatDeltaV = \SmoothingWeight \MatDeltaV_\nSamples + \left(1 - \SmoothingWeight\right) \MatDeltaV^\dagger$ for some value $\SmoothingWeight \in \closedinterval{0}{1}$.  From our dynamics equation,
	\begin{align*}
		\Vecx(t) &= \MatPhi\left(t,t\subinit\right) \Vecx\subinit + \MatPhi_v\left(t,\left\{\BurnTime_i\right\}_i\right) \MatDeltaV
			\\ &= \left[\SmoothingWeight + \left(1-\SmoothingWeight\right)\right] \MatPhi\left(t,t\subinit\right) \Vecx\subinit
				+ \MatPhi_v\left(t,\left\{\BurnTime_i\right\}_i\right) \left[ \SmoothingWeight \MatDeltaV_n + \left(1 - \SmoothingWeight\right) \MatDeltaV^\dagger \right]
			\\ &= \SmoothingWeight \left[ \MatPhi\left(t,t\subinit\right) \Vecx\subinit + \MatPhi_v\left(t,\left\{\BurnTime_i\right\}_i\right) \MatDeltaV_n \right]
				+ \left(1-\SmoothingWeight\right) \left[ \MatPhi\left(t,t\subinit\right) \Vecx\subnaught + \MatPhi_v\left(t,\left\{\BurnTime_i\right\}_i\right) \MatDeltaV^\dagger \right]
			\\ &= \SmoothingWeight \Vecx_n(t) + \left(1-\SmoothingWeight\right) \Vecx^\dagger(t)
	\end{align*}
	which is a convex combination, as required.  Substituting $t = t\subinit$ or $t = t\subgoal$, we see that $\Vecx(t)$ satisfies the boundary conditions given that $\Vecx_\nSamples(t)$ and $\Vecx^\dagger(t)$ do.  This completes the proof.
\end{proof}
\nomenclature[B01b ]{$\SmoothingWeight$}{Weighting factor used during trajectory smoothing}%
We take advantage of this fact for trajectory-smoothing.  Our algorithm, reported as \cref{alg: CWH Traj Smoothing} and illustrated in \cref{subfig: CWH Post-Processing Smoothing}, computes the approximate unconstrained minimum-\fuel solution $\Vecx^\dagger(t)$ and returns it (if feasible) or otherwise conducts a bisection line search on $\SmoothingWeight$, returning a convex combination of our original planning solution $\Vecx_\nSamples(t)$ and $\Vecx^\dagger(t)$ that comes as close to $\Vecx^\dagger(t)$ as possible without violating trajectory constraints.  Note because $\MatDeltaV_\nSamples$ lies in the feasible set of \cref{eqn: CWH Solution Used for Traj Smoothing}, the algorithm can only improve the final \fuel cost.  By design, \cref{alg: CWH Traj Smoothing} is geared towards reducing our original solution \fuelCost as quickly as possible while maintaining feasibility;  
the most expensive computational components are the calculation of $\MatDeltaV^\dagger$ and collision-checking (consistent with our sampling-based algorithm).  Fortunately, the number of collision-checks is limited by the maximum number of iterations $\ceil{\log_2\left(\frac{1}{\SmoothingTolerance}\right)} + 1$, given tolerance $\SmoothingTolerance \in \openinterval{0}{1}$.  As an added bonus, for strictly time-constrained applications that require a solution in a fixed amount of time, the algorithm can be easily modified to return the $\SmoothingWeight_{\ell}$-weighted trajectory $\Vecx\subsmooth(t)$ when time runs out, as the feasibility of this trajectory is maintained as an algorithm invariant.

\begin{figure}
	\begin{algorithm}[H]
		\caption{``Trajectory smoothing'' algorithm for impulsive CWH dynamics.  Given a trajectory $\Vecx_\nSamples(t), t \in \closedinterval{t\subinit}{t\subgoal}$ between initial and goal states $\Vecx\subinit$ and $\Vecx\subgoal$ satisfying \cref{eqn: CWH Optimal Steering} with impulses $\MatDeltaV_\nSamples$ applied at times $\left\{\BurnTime_i\right\}_i$, returns another feasible trajectory with reduced 2-norm \fuelCost.}
		\label{alg: CWH Traj Smoothing}
		\begin{algorithmic}[1]
			\State Initialize the smoothed trajectory $\Vecx\subsmooth(t)$ as $\Vecx_\nSamples(t)$, with $\MatDeltaV\subsmooth = \MatDeltaV_\nSamples$
			\State Compute the unconstrained optimal control vector $\MatDeltaV^\dagger$ by solving \cref{eqn: CWH Solution Used for Traj Smoothing}
			\State Compute the unconstrained optimal state trajectory $\Vecx^\dagger(t)$ using \cref{eqn: CWH Impulsive Solution} \emph{(See \cref{sec:sysDyn})}
			\State Initialize weight $\SmoothingWeight$ and its lower and upper bounds as $\SmoothingWeight \leftarrow 1$, $\SmoothingWeight_{\ell} \leftarrow 0$, $\SmoothingWeight_u \leftarrow 1$
			\While{true}
				\State $\Vecx(t) \leftarrow \left(1 - \SmoothingWeight\right) \Vecx_\nSamples(t) + \SmoothingWeight \Vecx^\dagger(t)$
				\State $\MatDeltaV \leftarrow \left(1 - \SmoothingWeight\right) \MatDeltaV_\nSamples + \SmoothingWeight \MatDeltaV^\dagger$
				\If{\Call{CollisionFree}{$\Vecx(t), \MatDeltaV, t$}}
					\State $\SmoothingWeight_{\ell} \leftarrow \SmoothingWeight$
					\State Save the smoothed trajectory $\Vecx\subsmooth(t)$ as $\Vecx(t)$ and control $\MatDeltaV\subsmooth$ as $\MatDeltaV$
				\Else
					\State $\SmoothingWeight_u \leftarrow \SmoothingWeight$
				\EndIf
				\If{$\SmoothingWeight_u - \SmoothingWeight_{\ell}$ is less than tolerance $\SmoothingTolerance \in \openinterval{0}{1}$}
					\State \textbf{break}
				\EndIf
				\State $\SmoothingWeight \leftarrow \sidefrac{\left(\SmoothingWeight_{\ell} + \SmoothingWeight_u\right)}{2}$
			\EndWhile
			\State \Return the smoothed trajectory $\Vecx\subsmooth(t)$, with $\MatDeltaV\subsmooth$ 
		\end{algorithmic}
	\end{algorithm}
\end{figure}
\nomenclature[B01b']{$\SmoothingTolerance$}{Trajectory smoothing convergence tolerance}%
\nomenclature[B01b#]{$\SmoothingWeight_\ell$, $\SmoothingWeight_u$}{Weighting factor bounds during trajectory smoothing}%

\section{Numerical Experiments}
\label{sec: Experiments}

Consider the two scenarios shown in \cref{fig: CWH Motion Planning Query}, here modeling near-field approaches of a chaser spacecraft in close proximity \footnote{\emph{Close proximity} in this context implies that any higher-order terms of the linearized relative dynamics are negligible, \eg within a few percent of the target orbit mean radius.} to a target moving in a circular LEO trajectory (as in \cref{fig: CWH Planning}).  We imagine the chaser, which starts in a circular orbit of lower radius, must be repositioned through a sequence of pre-specified CWH waypoints (\eg required for equipment checks, surveying, \etc) to a coplanar position located radially above the target, arriving with zero relative velocity in preparation for a final radial (``R-bar'') approach.  Throughout the maneuver, as described in detail in \cref{sec: Problem}, the chaser must avoid entering the elliptic target KOZ, enforce hard safety constraints with regard to a two-fault tolerance to stuck-off thruster failures, and otherwise avoid interfering with the target.  This includes avoiding the target's nadir-pointing communication lobes (represented by truncated half-cones), and preventing exhaust plume impingement on its surfaces.  For context, we use the Landsat-7 spacecraft and orbit as a reference \cite[3.2]{SNG-JGM-DLW-JRI-RJT:01} (see \cref{fig: LandSat-7 Mission Scenario}).

\begin{figure}
	\subcaptionbox{%
		\label{subfig: LandSat-7 Schematic}
		Landsat-7 schematic (Nadir (-$\posCrossTrack$ direction) points down, while the in-track ($+\posInTrack$) direction points left).
	}[0.55\columnwidth]{%
		\includegraphics[height=0.18\textheight]{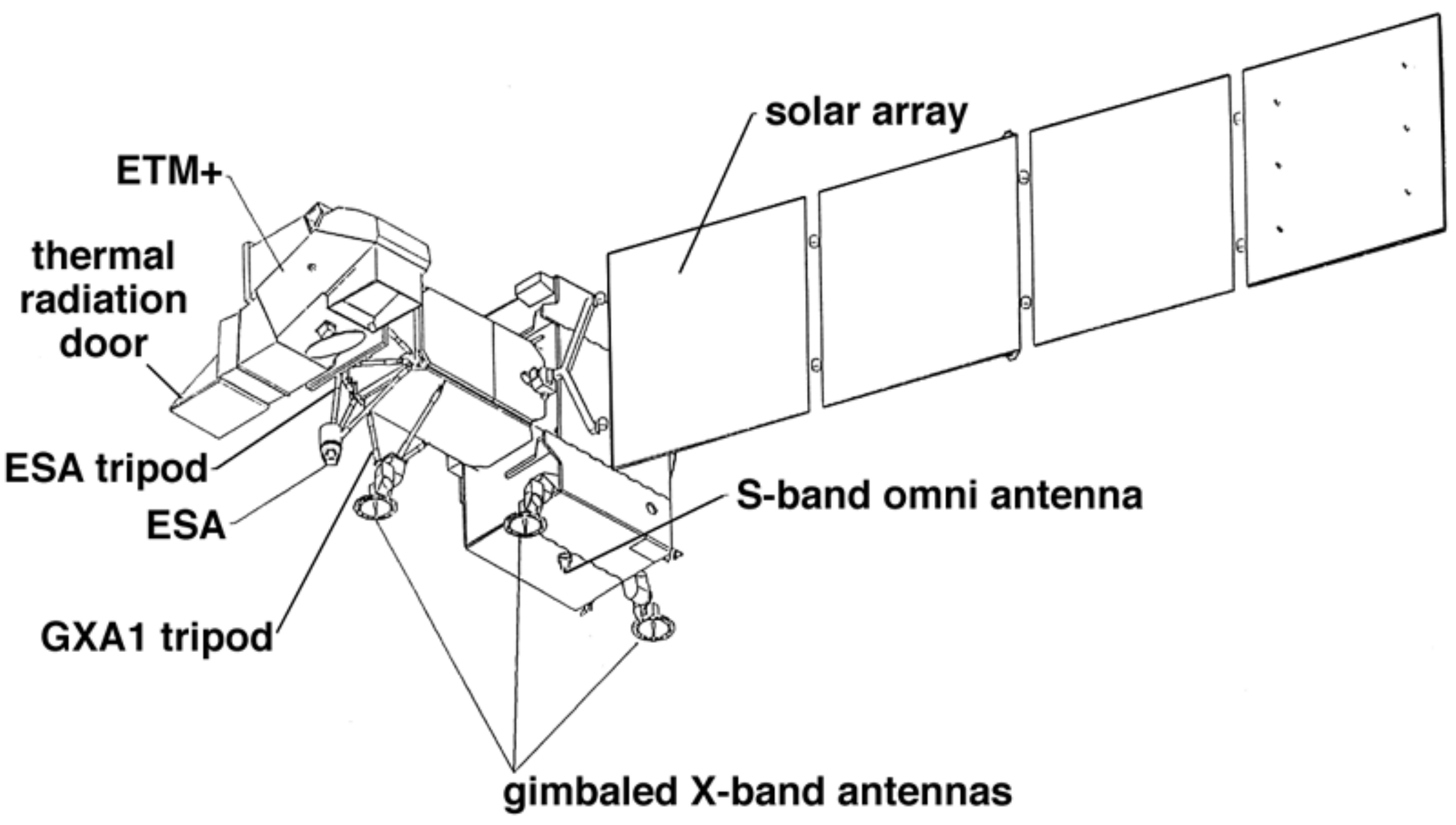}
	}
	\hspace{0.01\columnwidth}
	\subcaptionbox{%
		\label{subfig: LandSat-7 Orbit}
		Landsat-7 orbit (Courtesy of the \href{http://landsathandbook.gsfc.nasa.gov/orbit_coverage/}{Landsat-7 Handbook}).
	}[0.40\columnwidth]{%
		\includegraphics[height=0.18\textheight]{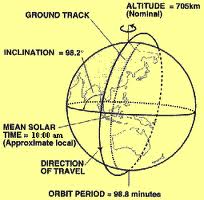}
	}
	\caption{Target spacecraft geometry and orbital scenario used in numerical experiments.}
	\label{fig: LandSat-7 Mission Scenario}
\end{figure}

\begin{figure}
	\subcaptionbox{%
		\label{subfig: CWH Planar Motion Planning Query.}
		Planar motion planning query
	}[0.36\columnwidth]{%
		\includegraphics[height=0.22\textheight]{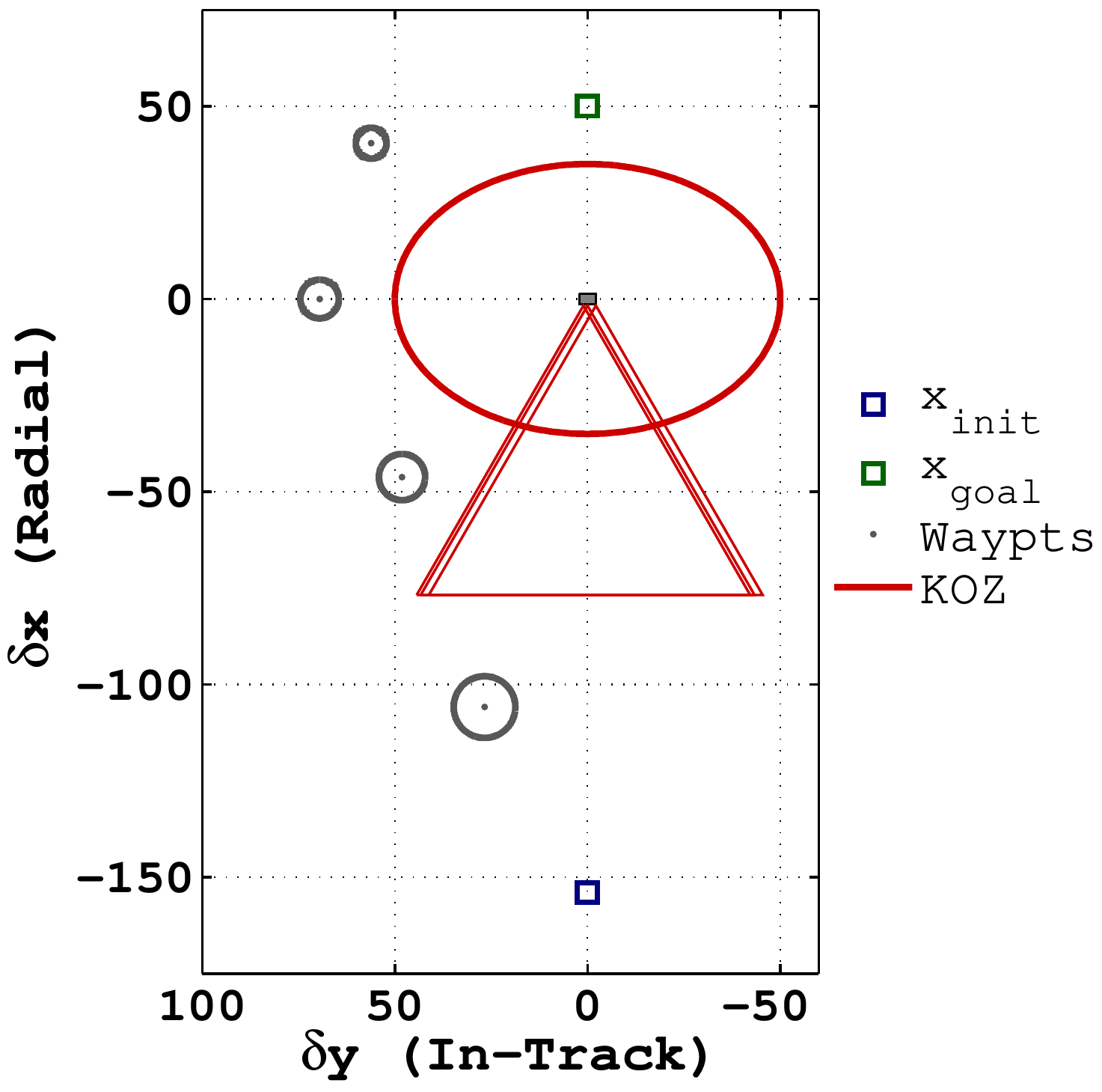}
	}
	\hspace{0.05\columnwidth}
	\subcaptionbox{%
		\label{subfig: CWH 3D Motion Planning Query}
		3-dimensional motion planning query.
	}[0.51\columnwidth]{%
		\includegraphics[height=0.22\textheight]{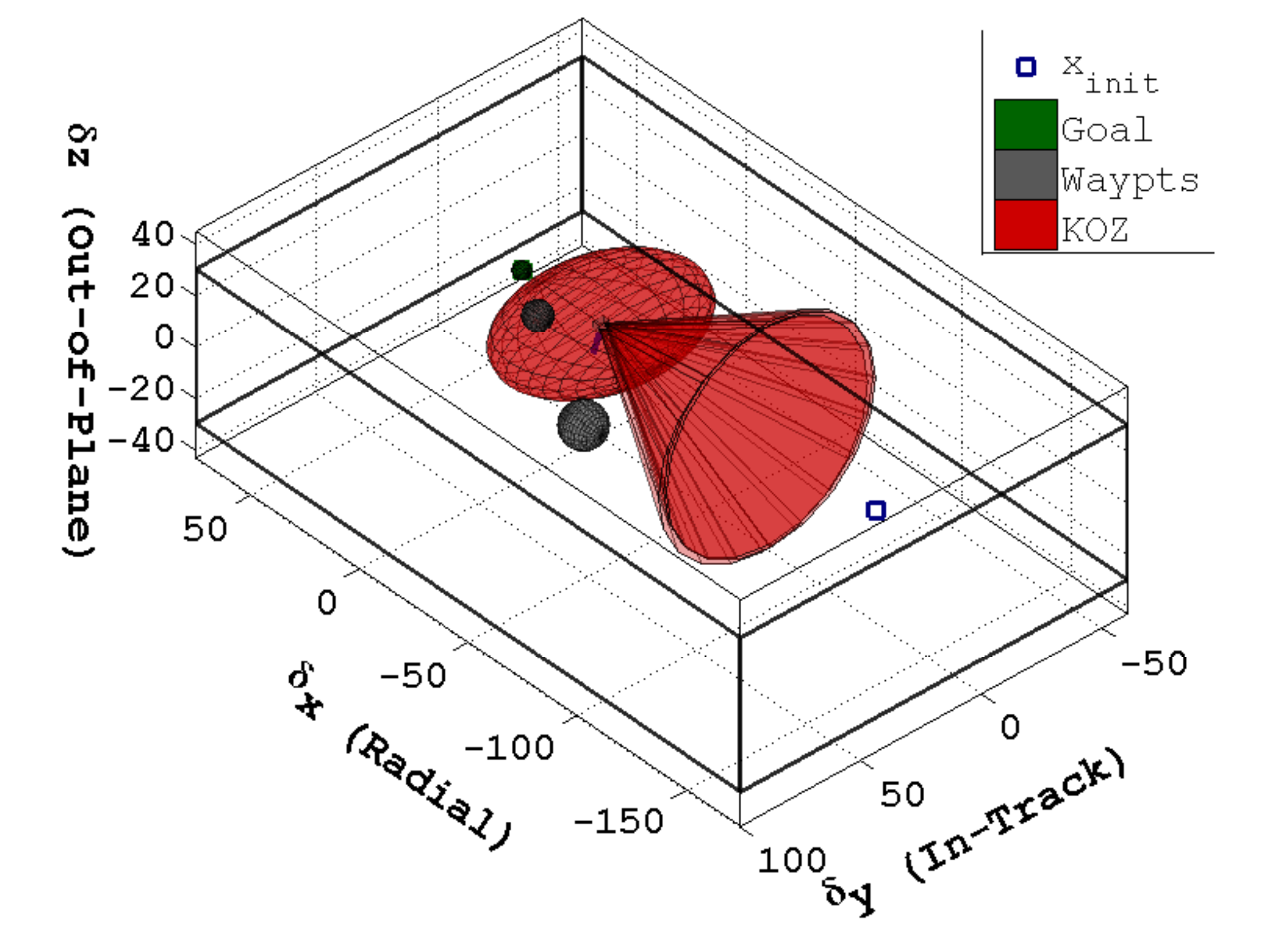}
	}
	\caption{Illustrations of the planar and 3D motion plan queries in the LVLH frame.  \iftoggle{AIAAjournal}{}{The spacecraft must track a series of guidance waypoints to the final goal state, which is located radially above the client.  Positional goal tolerances are visualized as circles around each waypoint, which successively shrink in size.}}
	\label{fig: CWH Motion Planning Query}
\end{figure}

\begin{figure}
	\subcaptionbox{%
		\label{subfig: Chaser spacecraft}
		Chaser spacecraft
	}[0.45\columnwidth]{%
		\includegraphics[height=0.2\textheight]{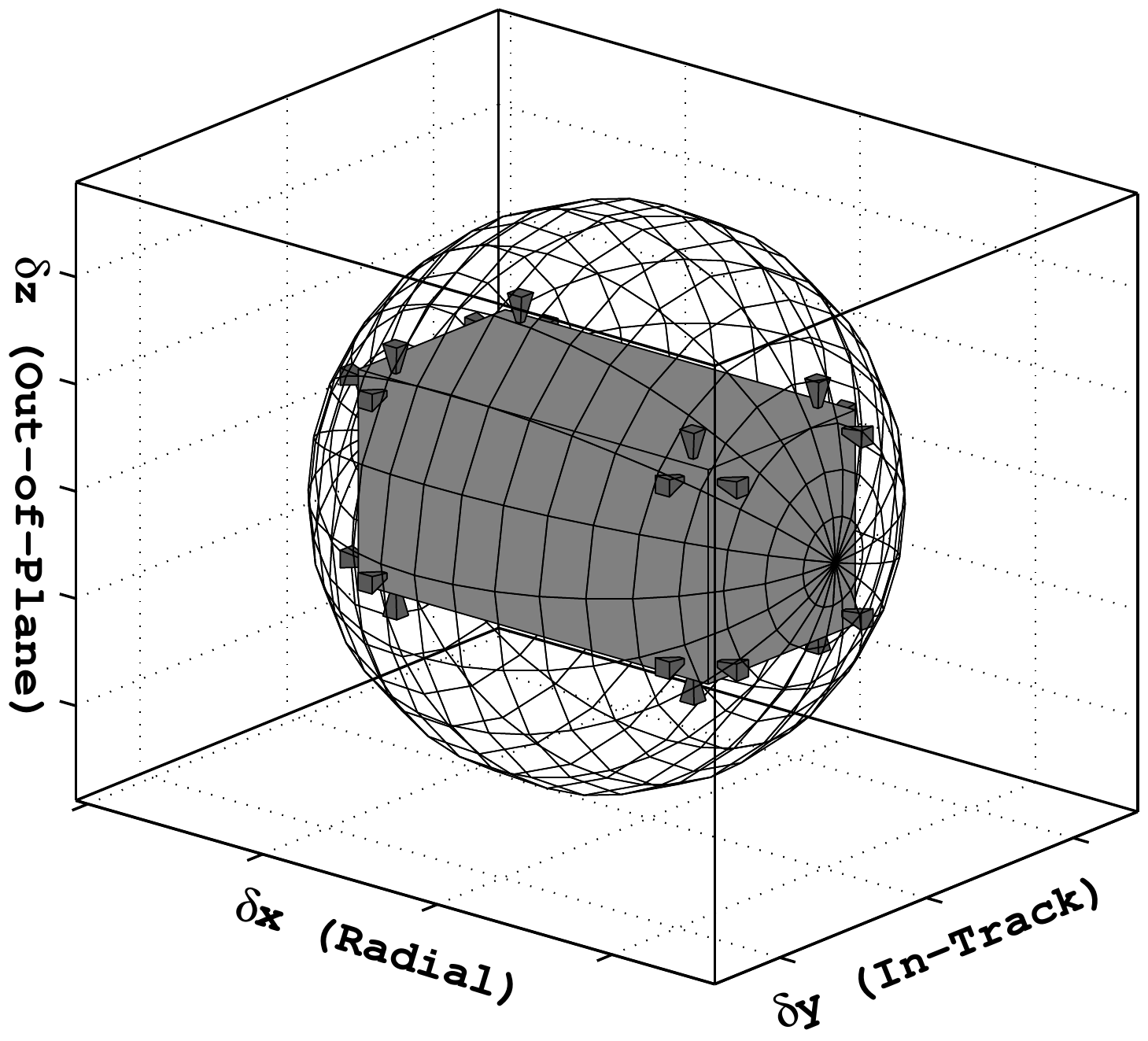}
	}
	\hspace{0.01\columnwidth}
	\subcaptionbox{%
		\label{subfig: Target spacecraft}
		Target spacecraft
	}[0.45\columnwidth]{%
		\includegraphics[height=0.2\textheight]{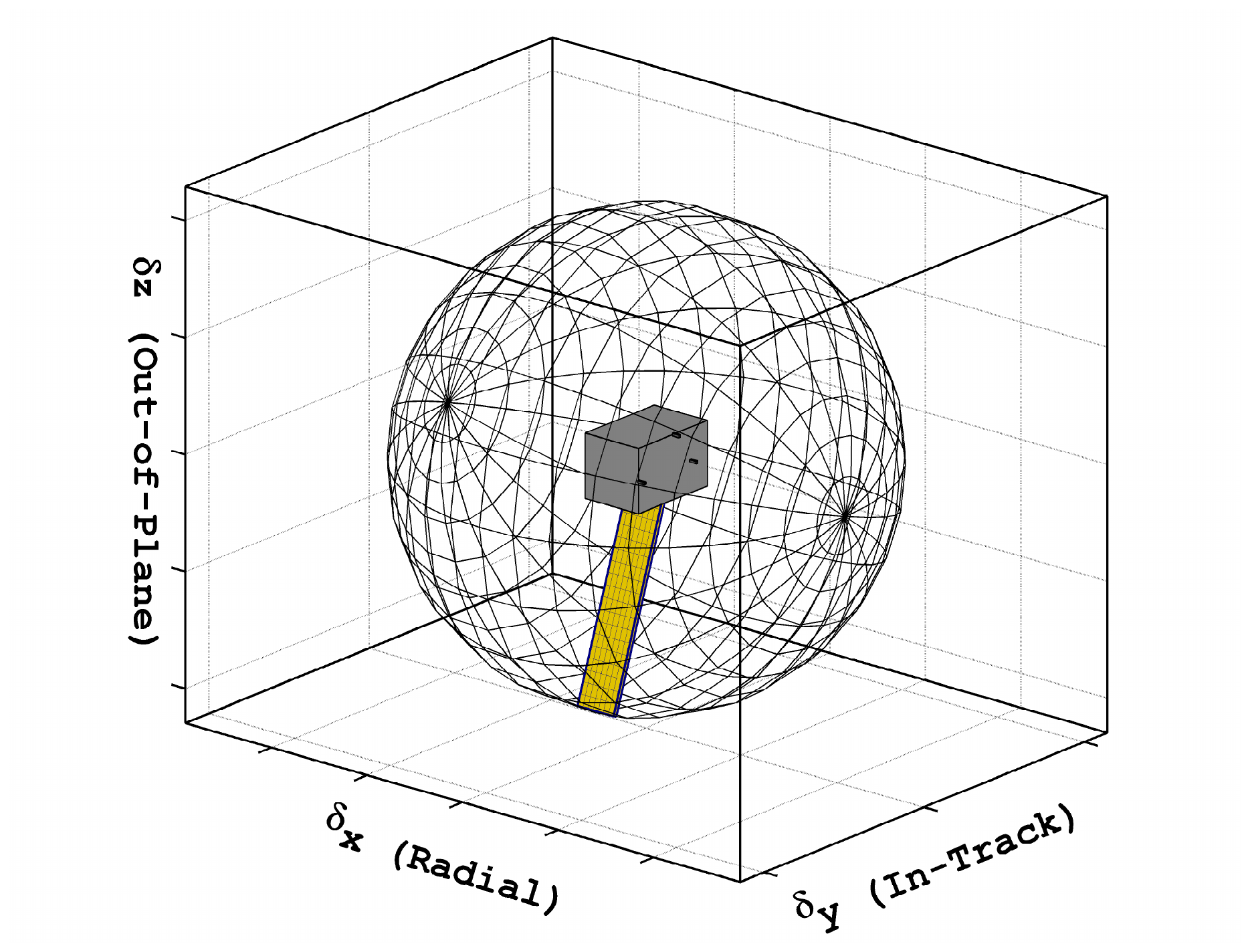}
	}
	\caption{Schematics of the chaser and target, together with their circumscribing spheres.} 
	\label{fig: Spacecraft Geometries}
\end{figure}

If we take the waypoints in the guidance sequence one at a time as individual goal points $\Vecx\subgoal$, we can solve the given scenario as a series of motion planning problems (or ``subplans''), calling \FMTstar from \cref{sec: Approach} once for each instance, linking them together to form an overall solution to the problem.  As our steering controller from \cref{subsec: steering} is attitude-independent, states $\Vecx \in {\reals}^\Dimension$ will be either $\VecxTranspose = \left[\begin{smallmatrix} \posCrossTrack & \posInTrack & \velCrossTrack & \velInTrack \end{smallmatrix}\right]$ with $\Dimension = 4$ (planar case) or $\VecxTranspose = \left[\begin{smallmatrix} \posCrossTrack & \posInTrack & \posOutofPlane & \velCrossTrack & \velInTrack & \velOutofPlane \end{smallmatrix}\right]$ with $\Dimension = 6$ (non-planar case).  We omit the attitude $\Vecq$ from the state during planning by assuming the existence of an attitude policy (as well as a stable attitude-trajectory-following controller) that produces $\Vecq(t)$ from the state trajectory $\Vecx(t)$; for illustration purposes, a simple nadir-pointing attitude profile is chosen to represent a mission that requires constant communication with the ground throughout the maneuver (note this is not enforced along actively-safe escape trajectories, which for each failure mode execute a simple ``turn-burn-turn'' policy that orients the closest available thruster as quickly as possible in the direction required to implement the necessary circularization burn\iftoggle{AIAAjournal}{}{-- see \cite{JS-BB-MP:15} for full details}).
Given the hyper-rectangular shape of the state-space, we call upon the deterministic, low-dispersion $d$-dimensional Halton sequence \cite{JHH:60} to sample positions and velocities.  To improve sample densities, each subplan uses its own sample space defined around only its respective initial and goal waypoints, with some arbitrary threshold space added around them.  Additionally, extra samples $\nSamples\subgoal$ are taken inside each waypoint ball to facilitate convergence.  For this multi-plan problem, we define the solution cost as the sum of individual subplan costs (if using trajectory smoothing, the endpoints between two plans will be merged identically to two edges within a plan, as described in \cref{sec: Smoothing/Robustness}).

Before we proceed to the results, we make note of a few implementation details.  First, for clarity we list the simulation parameters used in \cref{table: Simulation Parameters}.  
Second, all position-related constraint-checking regard the chaser spacecraft as a point at its center of mass, with all other obstacles artificially inflated by the radius of its circumscribing sphere.  Third and finally, all trajectory collision-checking is implemented by point-wise evaluation with a fixed time-step resolution $\Delta t$, using the analytic state transition equations 
\cref{eqn: CWH Impulsive Solution}
together with steering solutions from \cref{subsec: steering} to propagate graph edges; for speed, the line segments between points are excluded.  Except near very sharp obstacle corners, this approximation is generally not a problem in practice (obstacles can always be inflated further to account for this).  To improve performance, each obstacle primitive (ellipsoid, right-circular cone, hypercube, \etc) employs hierarchical collision-checking using hyper-spherical and/or hyper-rectangular bounding volumes to quickly prune points from consideration.

\begin{table}
	\centering
	\caption{List of parameters used during numerical experiments.}
	\label{table: Simulation Parameters}
	\begin{tabularx}{0.75\columnwidth}{l >{\raggedright\arraybackslash}X}
		\toprule
		Chaser plume half-angle, $\PlumeHalfAngle$ & $10\degrees$
		\\ Chaser plume height, $\PlumeHeight$ & 16 m
		\\ Chaser thruster fault tolerance, $\FaultTolerance$ & 2
		\\ Cost threshold, $\CostThreshold$ & 0.1--0.4 m/s
		\\ Dimension, $\Dimension$ & 4
		\\ Goal sample count, $\nSamples\subgoal$ & $0.04\nSamples$
		\\ Goal position tolerance, $\GoalPosTol$ & 3--8 m
		\\ Goal velocity tolerance, $\GoalVelTol$ & 0.1--0.5 m/s
		\\ Max.\ allocated thruster $\DeltaV$ magnitude, $\AllocatedThrustMax_{k}$ & $\infty$ m/s
		\\ Max.\ commanded $\DeltaV$-vector magnitude $\norm[]{\VecDeltaV_i}$, $\DeltaV\submax$ & $\infty$ m/s
		\\ Max.\ plan duration, $\TotalManeuverDuration\submax$ & $\infty$ s
		\\ Min.\ plan duration, $\TotalManeuverDuration\submin$ & 0 s
		\\ Max.\ steering maneuver duration, $\ManeuverDuration\submax$ & $0.1 \cdot \left(\sidefrac{2\pi}{\omegaRef}\right)$
		\\ Min.\ steering maneuver duration, $\ManeuverDuration\submin$ & 0 s
		\\ Sample count, $\nSamples$ & 50--400 per plan 
		\\ Simulation timestep, $\Delta t$ & $0.0005 \cdot \left(\sidefrac{2\pi}{\omegaRef}\right)$
		\\ Target antenna lobe height & 75 m
		\\ Target antenna beamwidth & $60\degrees$
		\\ Target KOZ semi-axes, $\left[\semiaxis_{\posCrossTrack}, \semiaxis_{\posInTrack}, \semiaxis_{\posOutofPlane}\right]$ & $\begin{bmatrix} 35 & 50 & 15 \end{bmatrix}$ m
		\\ \bottomrule
	\end{tabularx}
\end{table}

\subsection{Planar Motion Planning Solution}
A representative solution to the posed planning scenario, both with and without the trajectory smoothing algorithm (\cref{alg: CWH Traj Smoothing}), is shown in \cref{fig: Solution}.  
As shown, the planner successfully finds safe trajectories within each subplan, which are afterwards linked to form an overall solution.  The state space of the first subplan shown at the bottom is essentially obstacle-free, as the chaser at this point is too far away from the target for plume impingement to come into play.  This means every edge connection attempted here is added; so the first subplan illustrates well a discrete subset of the reachable states around $\Vecx\subinit$ and the unrestrained growth of \FMTstar.  As the second subplan is reached, the effects of the Keep-Out-Zone position constraints come in to play, and we see edges begin to take more leftward loops.  In subplans 3 and 4, plume impingement begins to play a role.  Finally, in subplan 5 at the top, where it becomes very cheap to move between states (as the spacecraft can simply coast to the right for free), we see the initial state connecting to nearly every sample in the subspace, resulting in a straight shot to the final goal.  As is evident, straight-line path planning would not approximate these trajectories well, particularly near coasting arcs which our dynamics allow the spacecraft to transfer to for free.

\begin{figure}
	\centering\includegraphics[width=0.45\columnwidth,trim=1.5cm 1.75cm 0cm 3.15cm, clip]{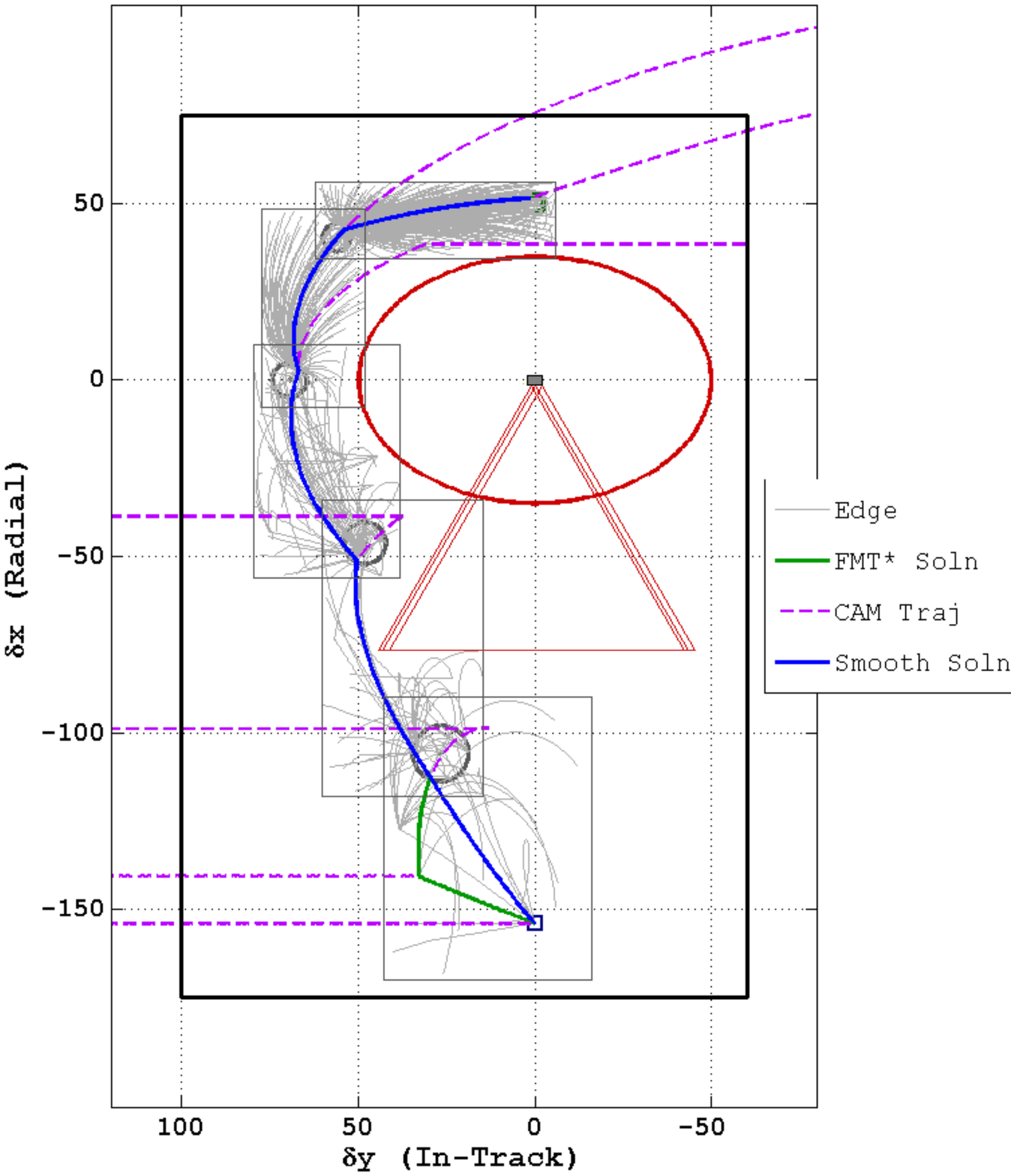}
	\caption{Representative planar motion planning solution using the \FMTstar algorithm (\cref{alg: FMTstar}) with $\nSamples = 2000$ (400 per subplan), $\CostThreshold = 0.3$ m/s, and relaxed waypoint convergence.  \iftoggle{AIAAjournal}{}{The output from \FMTstar is shown in green, while the trajectory combined with post-processing smoothing is shown in blue.  Explored trajectories found to be safe are shown in grey.  Actively-safe minimum-\fuel abort trajectories are shown as purple dashed lines (one for each burn $\VecDeltaV_i$ along the trajectory).}}
	\label{fig: Solution}
\end{figure}

To understand the smoothing process, examine \cref{fig: Smoothing}.  Here we see how the discrete trajectory sequence from our sampling-based algorithm may be smoothly and continuously deformed towards the unconstrained minimal-\fuel trajectory (as outlined in \cref{sec: Smoothing/Robustness}) until it meets trajectory constraints; if these constraints happen to be inactive, then the exact minimal-\fuel trajectory is returned, as \cref{subfig: Traj Smoothing Before/After} shows.  This computational approach is generally quite fast, assuming a well-implemented convex solver is used, as will be seen in the results of the next subsection.

\begin{figure}
	\subcaptionbox{%
		\label{subfig: Traj Smoothing Before/After}
		Paths before and after smoothing ($\nSamples = 2000$, $\CostThreshold = 0.2$ m/s, exact waypoint convergence).
	}[0.30\columnwidth]{%
		\includegraphics[width=0.95\columnwidth,trim=1.45cm 3.0cm 3.45cm 4.90cm, clip]{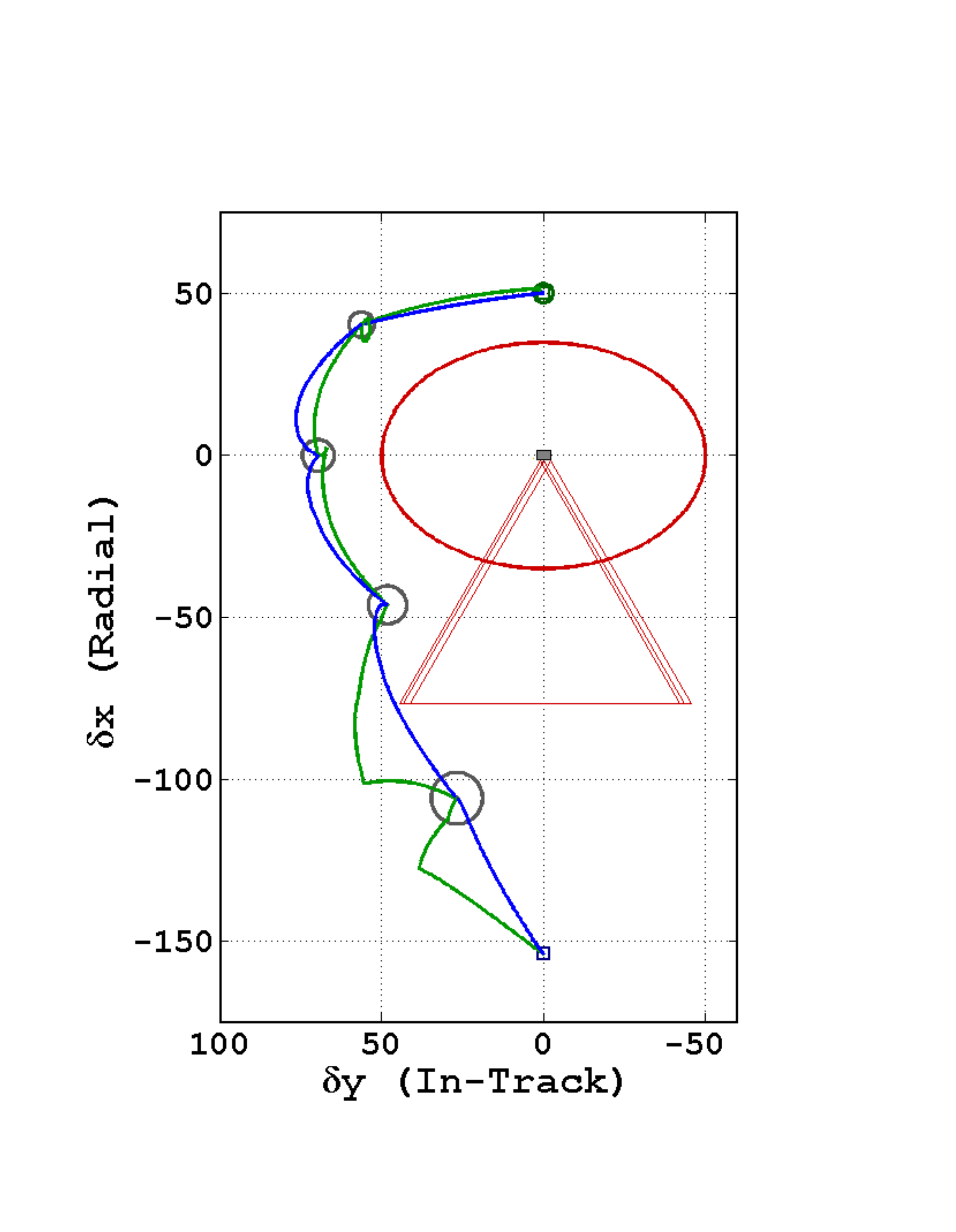}
	}
	\hspace{3em}
	\subcaptionbox{%
		\label{subfig: Traj Smoothing Iterations}
		Smoothing algorithm iterations ($\nSamples = 1500$, $\CostThreshold = 0.3$ m/s, inexact waypoint convergence).
	}[0.30\columnwidth]{%
		\includegraphics[width=0.95\columnwidth,trim=1cm 0cm 0cm 0cm, clip]{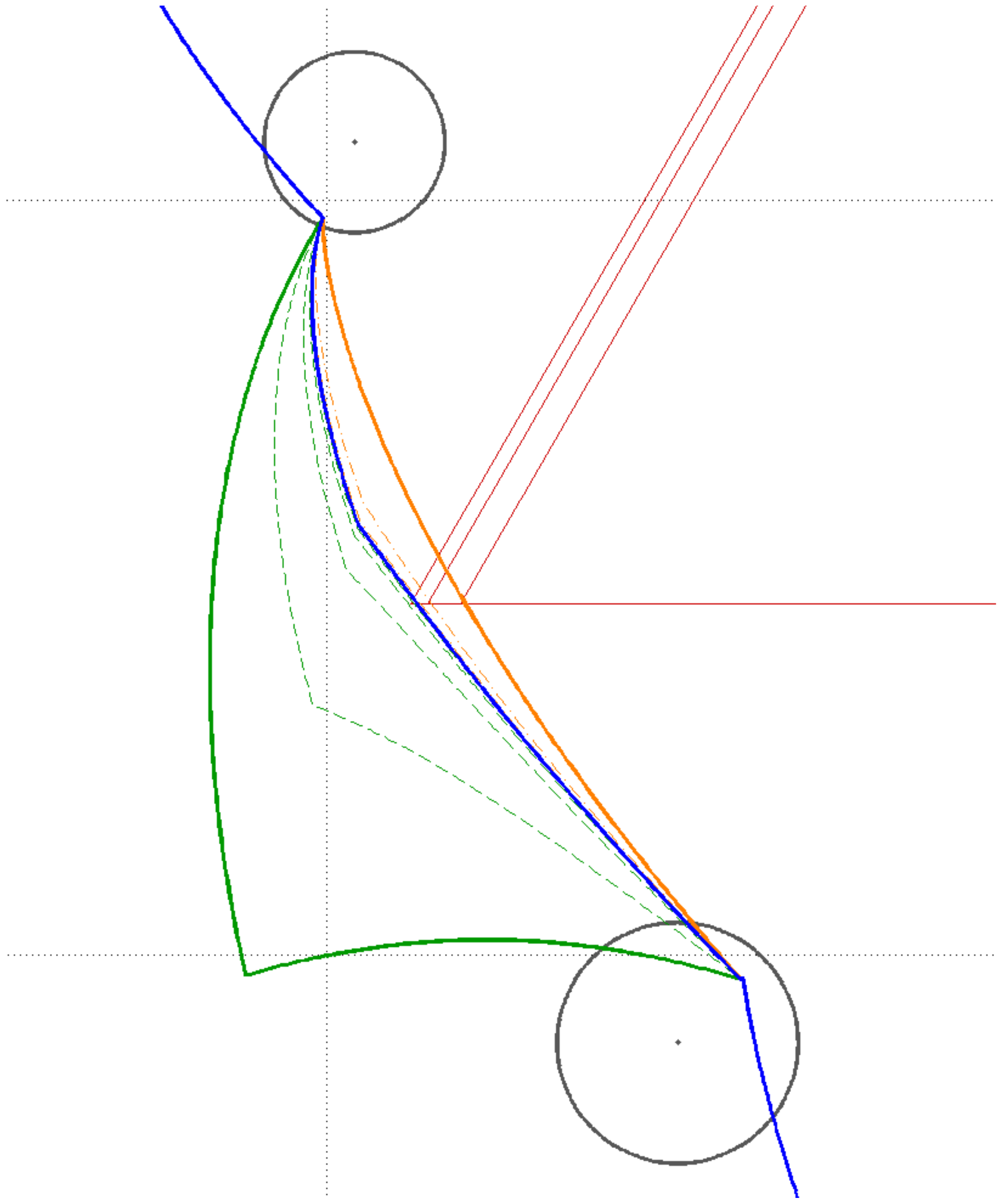}
	}
	\caption{Visualizing trajectory smoothing (\cref{alg: CWH Traj Smoothing}) for the solution shown in \cref{fig: Solution}\iftoggle{AIAAjournal}{.}{, zoomed in on the second plan.  The original plan is shown in green (towards the bottom-left), along with various iterates attempted while converging to the smoothed trajectory shown in blue (in the center).  Invalid trajectories, including the lower \fuelCost trajectory used to guide the process, are shown in orange (towards the right).}}
	\label{fig: Smoothing}
\end{figure}

The 2-norm $\DeltaV$ costs of the two reported trajectories in this example come to 0.835 m/s (unsmoothed) and 0.811 m/s (smoothed).  Compare this to 0.641 m/s, the cost of the unconstrained direct solution that intercepts each of the goal waypoints on its way to rendezvousing with $\Vecx\subgoal$ (this trajectory exits the state-space along the positive in-track direction, a violation of our proposed mission; hence its cost represents an under-approximation to the true optimal cost $\CostFcn\supopt$ of the constrained problem).  This suggests that our solutions are quite close to the constrained optimum, and certainly on the right order of magnitude.  Particularly with the addition of smoothing at lower sample counts, the approach appears to be a viable one for spacecraft planning.

If we compare the 2-norm $\DeltaV$ costs to the actual measured \fuel consumption given by the sum total of all allocated thruster $\DeltaV$ magnitudes, which equal 1.06 m/s (unsmoothed) and 1.01 m/s (smoothed) respectively, we find increases of 27.0\% and 24.5\%; as expected, our 2-norm cost metric under-approximates the true \fuel cost.  For point-masses with isotropic control authority (\eg a steerable or gimbaled thruster that is able to point freely in any direction), our cost metric would be exact.  However, for our distributed attitude-dependent propulsion system (see \cref{subfig: Chaser spacecraft}), it is clearly a reasonable proxy for allocated \fuel use, returning values on the same order of magnitude.  Though we cannot make a strong statement about our proximity to the \fuel-optimal solution without directly optimizing over thruster $\DeltaV$ allocations, our solution clearly seems to promote low \fuel consumption.  

\subsection{Non-Planar Motion Planning Solution}
For the non-planar case, representative smoothed and unsmoothed \FMTstar solutions can be found in \cref{fig: Non-Planar Solution}.  Here the spacecraft is required to move out-of-plane to survey the target from above before reaching the final goal position located radially above the target.  The first subplan involves a long reroute around the conical region spanned by the target spacecraft's communication lobes.  Because the chaser begins in a coplanar circular orbit at $\Vecx\subinit$, most steering trajectories require a fairly large cost to maneuver out-of-plane to the first waypoint.  Consequently, relatively few edges are added that both lie in the reachable set of $\Vecx\subinit$ and safely avoid the large conical obstacles.  As we progress to the second and third subplans, the corresponding trees become denser (more steering trajectories are both safe and within our cost threshold $\CostThreshold$) as the state space becomes freer.  Compared with the planar case, the extra degree-of-freedom associated with the out-of-plane dimension appears to allow more edges ahead of the target in the in-track direction than before, likely because now the exhaust plumes generated by the chaser are well out-of-plane from the target spacecraft.  Hence the spacecraft smoothly and tightly curls around the ellipsoidal KOZ to the goal.

The 2-norm $\DeltaV$ costs for this example come to 0.611 m/s (unsmoothed) and 0.422 m/s (smoothed).  
Counter-intuitively, these costs are on the same order of magnitude and slightly cheaper than the planar case; the added freedom given by the out-of-plane dimension appears to outweigh the high costs typically associated with inclination changes and out-of-plane motion.  The 2-norm cost values correspond to total thruster $\DeltaV$ allocation costs of 0.893 m/s and 0.620 m/s, respectively -- increases of 46\% and 47\% above their counterpart cost metric values.  Again, our cost metric appears to be a reasonable proxy for actual \fuel use.


\begin{figure}
	\subcaptionbox{%
	}[0.45\columnwidth]{%
		\includegraphics[width=\columnwidth]{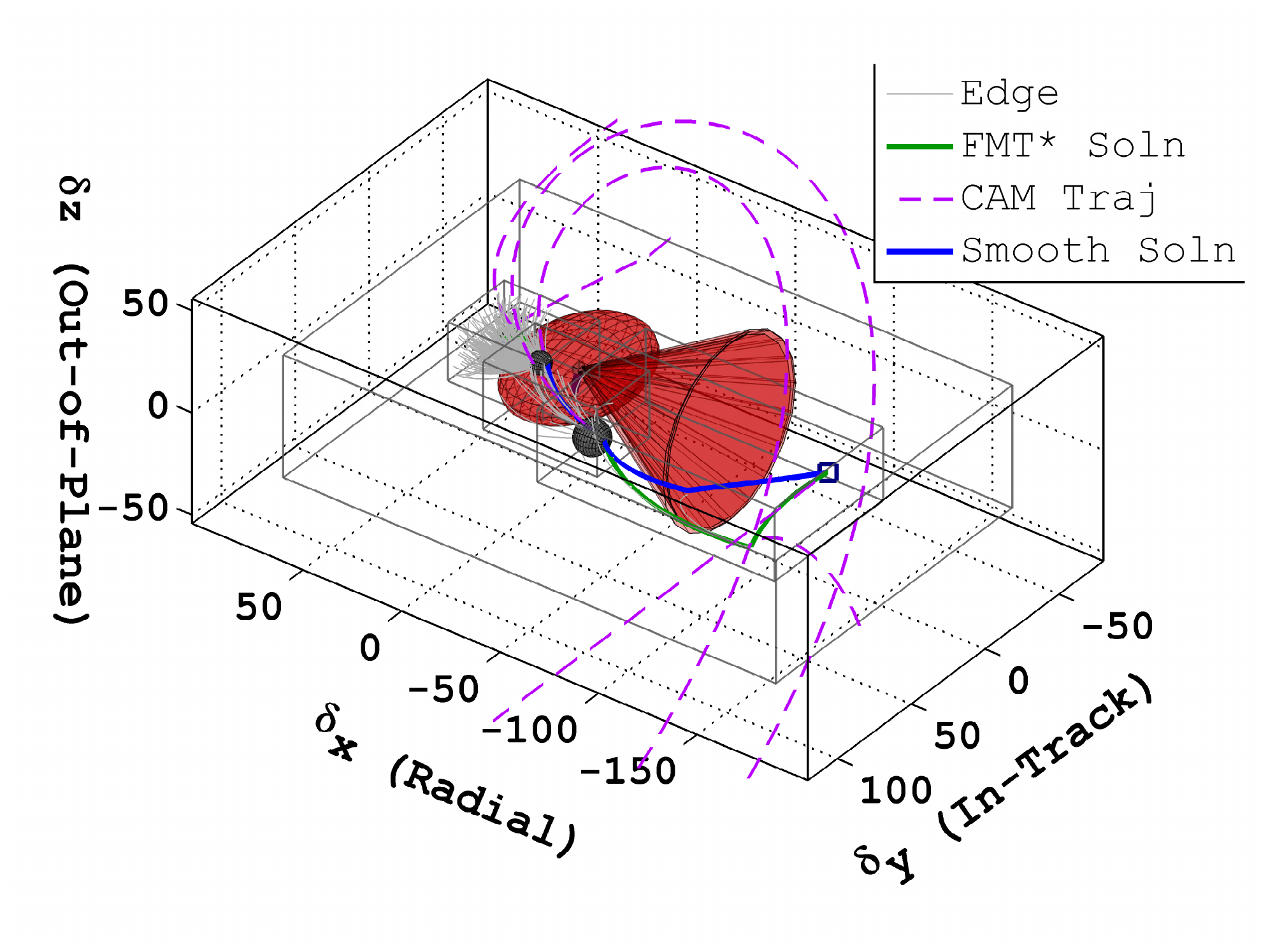}
	}
	\hspace{1em}
	\subcaptionbox{%
	}[0.45\columnwidth]{%
		\includegraphics[width=\columnwidth]{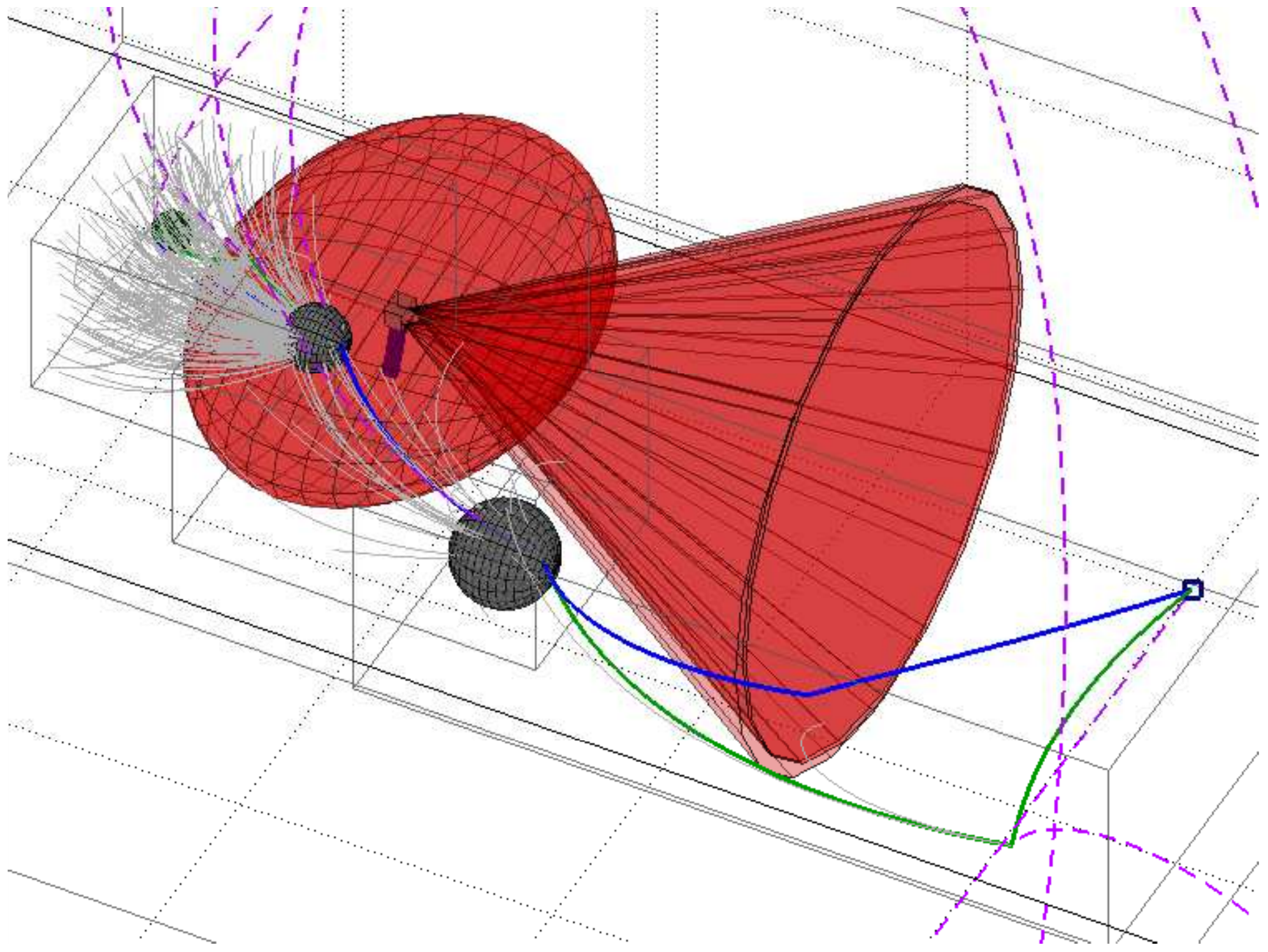}
	}
	\caption{Representative non-planar motion planning solution using the \FMTstar algorithm (\cref{alg: FMTstar}) with $\nSamples = 900$ (300 per subplan), $\CostThreshold = 0.4$ m/s, and relaxed waypoint convergence.  \iftoggle{AIAAjournal}{}{The output from \FMTstar is shown in green, while the trajectory combined with post-processing smoothing is shown in blue.  Explored trajectories found to be safe are shown in grey.  Actively-safe minimum-\fuel abort trajectories are shown as purple dashed lines (one for each burn $\VecDeltaV_i$ along the trajectory).}}
	\label{fig: Non-Planar Solution}
\end{figure} 

\subsection{Performance Evaluation}
To evaluate the performance of our approach, an assessment is necessary of solution quality as a function of planning parameters, most importantly the number of samples $\nSamples$ taken and the reachability set cost threshold $\CostThreshold$.  As proven in \cref{subsec: Theoretical Characterization}, the solution cost will eventually reduce to the optimal value as we increase the sample size $\nSamples$.  Additionally, one can increase the cost threshold $\CostThreshold$ used for nearest-neighbor identification so that more connections are explored.  However, both come at the expense of running time.  To understand the effects of these changes on quality, particularly at finite sample counts where the asymptotic guarantees of \FMTstar do not necessarily imply cost improvements, we measure the cost versus computation time for the planar planning scenario parameterized over several values of each $\nSamples$ and $\CostThreshold$.  

Results are reported in \crefrange{fig: Cost vs Run-time for Varying Jbar}{fig: Cost vs Run-time for Varying Sample Count}.  For a given sequence of sample count/cost threshold pairs, we ran our algorithm in each configuration and recorded the total cost of \emph{successful} runs and their respective run times \footnote{All simulations were implemented in MATLAB 2012b and run on a Windows-operated PC, clocked at 4.00 GHz and equipped with 32.0 GB of RAM.  CVXGEN and CVX \cite{MG-SB:14}, disciplined convex programming solvers, were used to implement $\DeltaV$ allocation and trajectory smoothing, respectively.} 
as measured by wall clock time.  Note that all samples were drawn and their interconnecting steering problems were solved \emph{offline} per our discussion in \cref{subsec: Algorithm}. 
Only the \emph{online} components of each call constitute the run times reported, including running \FMTstar with collision-checking and graph construction, as these are the only elements critical to the real-time implementability of the approach; everything else may be computed offline on ground computers where computation is less restricted, and later uplinked to the spacecraft or stored onboard prior to mission launch.  See \cref{subsec: Algorithm} for details. 
Samples were stored as a $\Dimension \cross \nSamples$ array, while inter-sample steering controls $\VecDeltaV\supopt_i$ and times $\BurnTime_i$ were precomputed as $\nSamples \cross \nSamples$ arrays of $\sidefrac{\Dimension}{2} \cross \NBurns$ and $\NBurns \cross 1$ array elements, respectively.  To reduce memory requirements, steering trajectories $\Vecx\supopt$ and $\Vecq$ where generated online through \cref{eqn: CWH Impulsive Solution} and our nadir-pointing assumption, though in principle they could have easily been stored as well to save additional computation time.

\begin{figure}
	\subcaptionbox{%
		\label{subfig: Cost vs Time for Varying Jbar (Exact)}
		Exact waypoint convergence ($\nSamples = 2000$).
	}[0.45\columnwidth]{%
		\begin{tikzpicture}
			\node[anchor=south west,inner sep=0] (image) {\centering\includegraphics[width=\columnwidth]{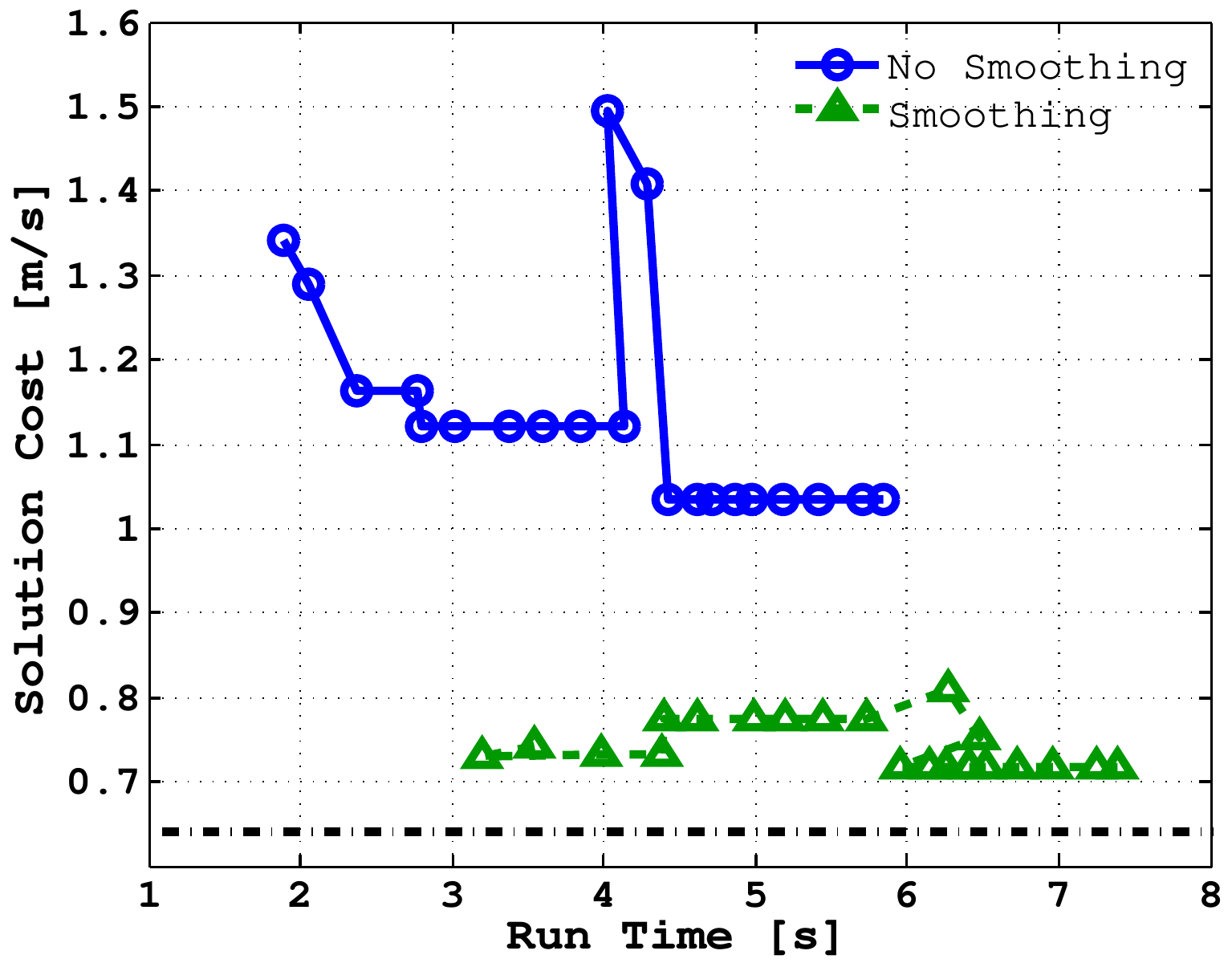}};
			\begin{scope}[x={(image.south east)},y={(image.north west)}]
				\node [anchor=south west] (CostThresholdTrendUnsmoothed) at (0.6,0.545) {\small{$\CostThreshold$ increasing}};
				\node [anchor=south] (CostThresholdTrendSmoothed) at (0.575,0.31) {\small{$\CostThreshold$ increasing}};
				\draw [-stealth,very thick,solid,black] (0.56,0.6) to[out=-40, in=170] (0.71,0.53);
				\draw [-stealth,very thick,solid,black] (0.5,0.31) to (0.65,0.31);
			\end{scope}
		\end{tikzpicture}
	}
	\hspace{1em}
	\subcaptionbox{%
		\label{subfig: Cost vs Time for Varying Jbar (Inexact)}
		Inexact waypoint convergence ($\nSamples = 2000$).
	}[0.45\columnwidth]{%
		\begin{tikzpicture}
			\node[anchor=south west,inner sep=0] (image) {\centering\includegraphics[width=\columnwidth]{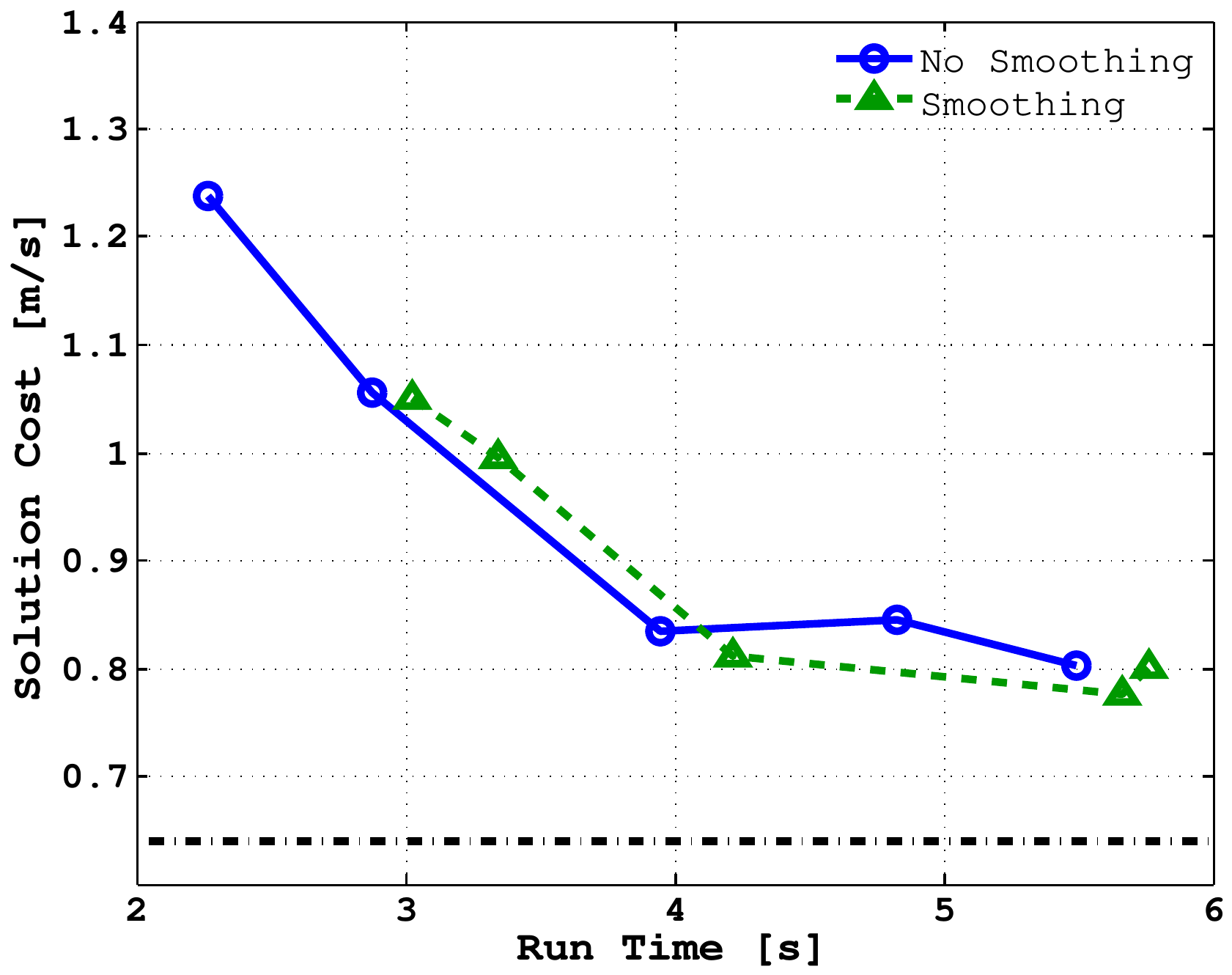}};
			\begin{scope}[x={(image.south east)},y={(image.north west)}]
				\node [anchor=south west] (CostThresholdTrend) at (0.52,0.5) {\small{$\CostThreshold$ increasing}};
				\draw [-stealth,very thick,solid,black] (0.45,0.58) to[out=-40, in=170] (0.65,0.45);
			\end{scope}
		\end{tikzpicture}
	}
	\caption{Algorithm performance for the given LEO proximity operations scenario as a function of varying cost threshold ($\CostThreshold \in \closedinterval{0.2}{0.4}$) with $\nSamples$ held constant\iftoggle{AIAAjournal}{.}{ (lowering $\CostThreshold$ at these $\nSamples$ yields failure).  Results are reported for both (i) trajectories constrained to rendezvous exactly with pre-specified waypoints and (ii) trajectories that can terminate anywhere in $\Xgoal$ (inside a given position/velocity tolerance).}}
	\label{fig: Cost vs Run-time for Varying Jbar}
\end{figure}

\Cref{fig: Cost vs Run-time for Varying Jbar} reports the effects on the solution cost of varying the nearest-neighbor search threshold $\CostThreshold$ while keeping $\nSamples$ fixed.  As described in \cref{subsec: Reachability Sets}, $\CostThreshold$ determines the size of state reachability sets and hence the number of candidate neighbors evaluated during graph construction.  Generally, this means an improvement in cost at the expense of extra processing; though there are exceptions as in \cref{subfig: Cost vs Time for Varying Jbar (Exact)} at $\CostThreshold \approx 0.3$ m/s.  Likely this arises from a neighbor that is found and connected to (at the expense of another, since \FMTstar only adds one edge per nearest-neighborhood) which leads to a particular graph change for which \emph{exact} termination at the goal waypoint is more expensive than usual.  Indeed we see that \emph{for the same sample distribution} this does not occur, as shown in the other case where inexact convergence is permitted.

We can also vary the sample count $\nSamples$ while holding $\CostThreshold$ constant.  From \crefrange{subfig: Cost vs Time for Varying Jbar (Exact)}{subfig: Cost vs Time for Varying Jbar (Inexact)}, we select $\CostThreshold = 0.22$ m/s and $0.3$ m/s, respectively, for each of the two cases (the values which suggest the best solution cost per unit of run time).  Repeating the simulation for varying sample count values, we obtain \cref{fig: Cost vs Run-time for Varying Sample Count}.  Note the general downward trend as run time increases (corresponding to larger sample counts), indicating the classic trade-off in sampling-based planning.  However, there is bumpiness.  Similar to before, this is likely due to new connections previously unavailable at lower sample counts which cause a slightly different graph with an unlucky jump in \fuel cost.  This reinforces the well-known need to tune $\nSamples$ and $\CostThreshold$ before applying sampling-based planners.

\begin{figure}
	\subcaptionbox{%
		\label{subfig: Cost vs Time for Varying Sample Count, Exact}
		Exact waypoint convergence ($\CostThreshold = 0.22$ m/s).
	}[0.45\columnwidth]{%
		\begin{tikzpicture}
			\node[anchor=south west,inner sep=0] (image) {\centering\includegraphics[width=\columnwidth]{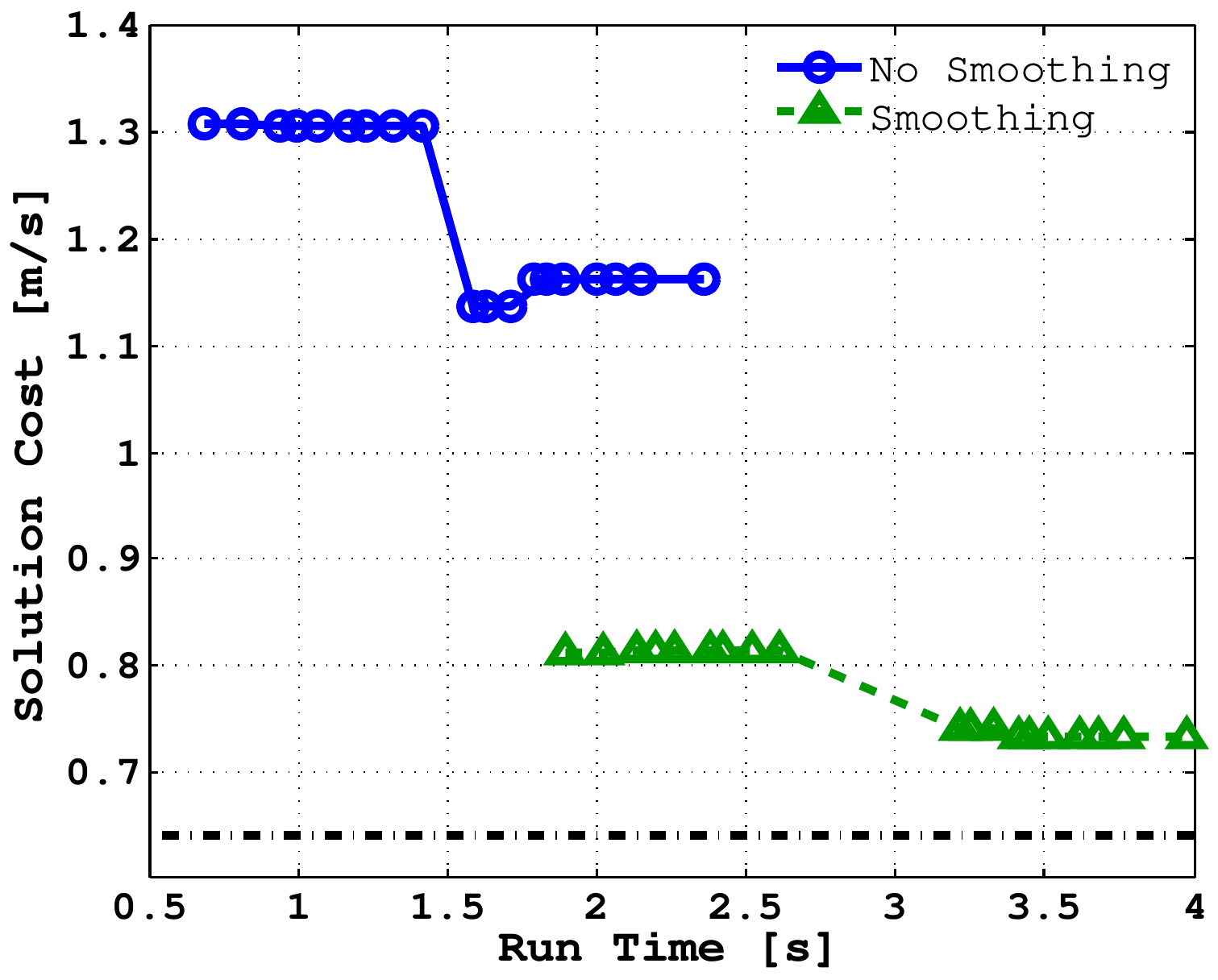}};
			\begin{scope}[x={(image.south east)},y={(image.north west)}]
				\node [anchor=north east] (CostThresholdTrendUnsmoothed) at (0.43,0.59) {\small{$\nSamples$ increasing}};
				\node [anchor=south] (CostThresholdTrendSmoothed) at (0.575,0.39) {\small{$\nSamples$ increasing}};
				\draw [-stealth,very thick,solid,black] (0.28,0.75) to[out=-70, in=175] (0.4,0.6);
				\draw [-stealth,very thick,solid,black] (0.5,0.39) to (0.65,0.39);
			\end{scope}
		\end{tikzpicture}
	}
	\hspace{1em}
	\subcaptionbox{%
		\label{subfig: Cost vs Time for Varying Sample Count, Inexact}
		Inexact waypoint convergence ($\CostThreshold = 0.3$ m/s).
	}[0.45\columnwidth]{%
		\begin{tikzpicture}
			\node[anchor=south west,inner sep=0] (image) {\centering\includegraphics[width=\columnwidth]{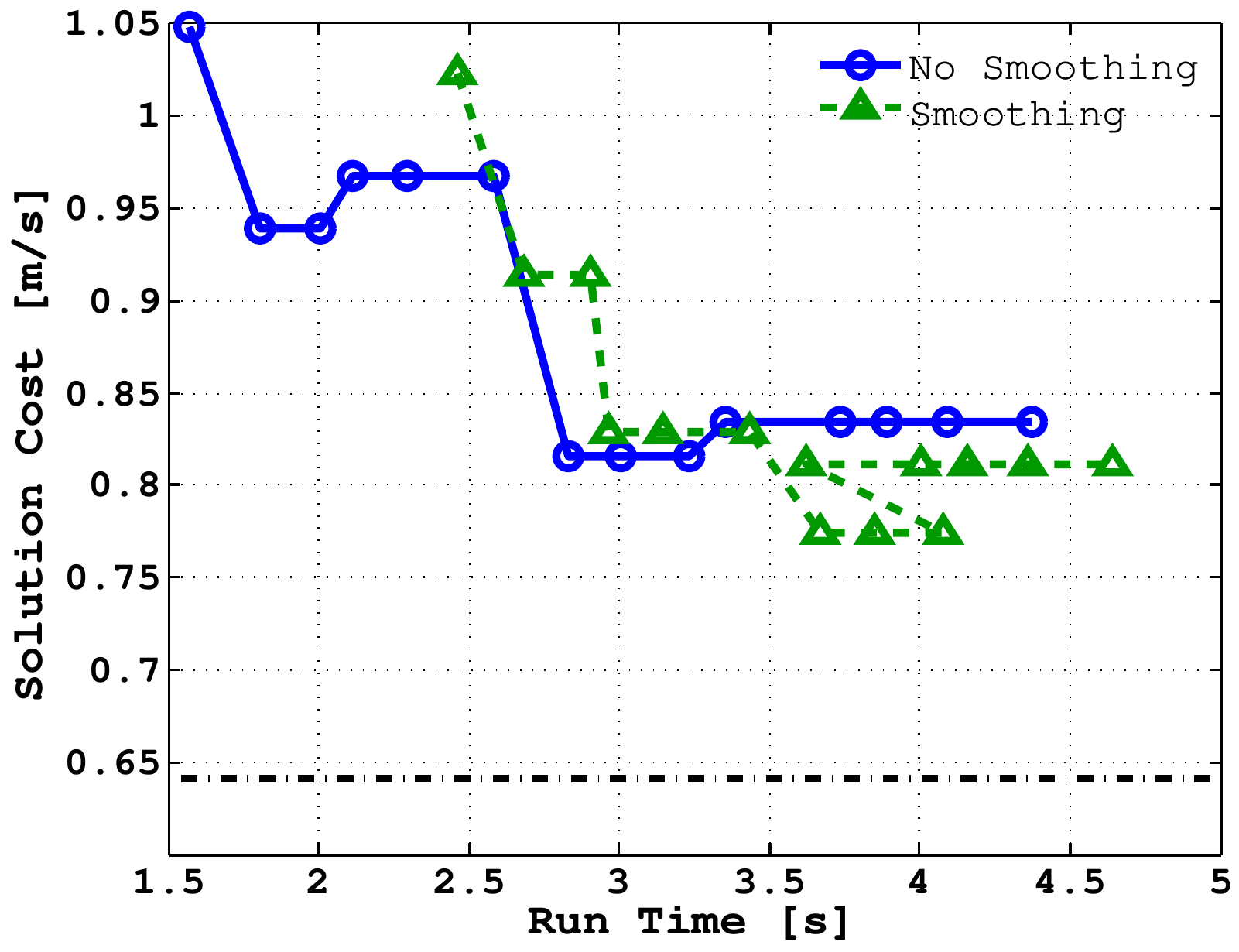}};
			\begin{scope}[x={(image.south east)},y={(image.north west)}]
				\node [anchor=south west] (CostThresholdTrend) at (0.535,0.7) {\small{$\nSamples$ increasing}};
				\draw [-stealth,very thick,solid,black] (0.5,0.8) to[out=-60, in=170] (0.65,0.64);
			\end{scope}
		\end{tikzpicture}
	}
	\caption{Algorithm performance for the given LEO proximity operations scenario as a function of varying sample count ($\nSamples \in \closedinterval{650}{2000}$) with $\CostThreshold$ held constant\iftoggle{AIAAjournal}{.}{ (lowering $\nSamples$ further at these $\CostThreshold$ yields failure).  Results are reported for trajectories both with and without exact waypoint convergence.}}
	\label{fig: Cost vs Run-time for Varying Sample Count}
\end{figure}

As the figures show, the utility of trajectory smoothing is clearly affected by the fidelity of the planning simulation.  In each, trajectory smoothing yields a much larger improvement in cost at modest increases in computation time when we require exact waypoint convergence.  It provides little improvement, on the other hand, when we relax these waypoint tolerances; \FMTstar (with goal region sampling) seems to return trajectories with costs much closer to the optimum in such cases, making the additional overhead of smoothing less favorable.  This conclusion is likely highly problem-dependent; these tools must always be tested and tuned to the particular application.

Note that the overall run times for each simulation are on the order of 1-5 seconds, including smoothing.  This clearly indicates that \FMTstar can return high quality solutions in real-time for spacecraft proximity operations.  Though run on a computer currently unavailable to spacecraft, we hope that our examples serve as a reasonable proof-of-concept; we expect that with a more efficient coding language and implementation, our approach would be competitive on spacecraft hardware.

\section{Conclusions}
\label{sec: Conclusions}

A technique has been presented for efficiently automating minimum-\fuel guidance during near-circular orbit proximity operations, enabling the computation of near-optimal collision-free trajectories in real time (on the order of 1-5 seconds for our numerical examples).  The approach allows our modified version of the \FMTstar sampling-based motion planning algorithm to approximate the solution to the minimal-\fuel trajectory control problem \cref{eqn: CWH Optimal Steering} under impulsive Clohessy-Wiltshire-Hill (CWH) dynamics.  The method begins by discretizing the feasible space of \cref{eqn: CWH Optimal Steering} through state space sampling in the CWH Local-Vertical Local-Horizontal (LVLH) frame.  Next, state samples and their forward \emph{reachability sets}, which we have shown comprise sets bounded by unions of ellipsoids taken over steering maneuver duration, are precomputed offline and stored onboard the spacecraft together with all pairwise steering solutions.  Finally, the \FMTstar algorithm (with built-in trajectory smoothing) is called online to efficiently construct a tree of trajectories through the feasible state space towards a goal region, returning a solution that satisfies a broad range of trajectory constraints (\eg plume impingement, control allocation feasibility, obstacle avoidance, \etc) or else reporting failure.  If desired, additional post-processing using the techniques outlined in \cref{sec: Smoothing/Robustness} can be employed to reduce solution \fuel cost.

The key breakthrough of our solution for autonomous spacecraft guidance is its judicious distribution of computations; in essence, only what \emph{must} be computed onboard, such as collision-checking and graph construction, is computed online -- everything else, including the most intensive computations, are relegated to the ground where computational effort and run time are less critical.  Furthermore, only minimal information (steering problem control trajectories, costs, and nearest-neighbor sets) requires storage on the spacecraft.  Though we have illustrated through simulations the ability to tackle a particular minimum-\fuel LEO homing maneuver problem, it should be noted that the methodology applies equally well to other objectives, such as the minimum-time problem, and can be generalized to other dynamic models and environments.  The approach is flexible enough to handle non-convexity and mixed state-control-time constraints without compromising real-time implementability, so long as 
constraint function evaluation is relatively efficient.  In short, the proposed approach appears to be useful for automating the mission planning process for spacecraft proximity operations, enabling real-time computation of low cost trajectories.

In future work, the authors plan to demonstrate the proposed approach in a number of other proximity operations scenarios, including optimal deep-space guidance, pinpoint asteroid descent, and onboard a set of free-flying, air-bearing robots.  However, the proposed planning framework for impulsively-actuated spacecraft offers several other interesting avenues for future research.  For example, though nothing in the methodology forbids it outside of computational limitations, it would be interesting to revisit the problem with attitude states included in the planning process (instead of abstracted away, as we have done here by assuming an attitude profile).  This would allow direct inclusion of attitude constraints into maneuver planning (\eg enforcing line-of-sight, keeping solar panels oriented towards the Sun to stay power positive, maintaining a communication link between the chaser antenna and ground, \etc).  Also of interest would be other actively-safe policies that relax the need to circularize escape orbits (potentially costly in terms of \fuel use) or which mesh better with trajectory smoothing, without the need to add compensating impulses (see \cref{sec: Smoothing/Robustness}).  Extensions to dynamic obstacles (such as debris or maneuvering spacecraft, which are unfixed in the LVLH frame), elliptical target orbits, higher-order gravitation, curvilinear coordinates, or dynamics under relative orbital elements also represent key research topics vital to extending the method's applicability to more general maneuvers.  Finally, memory and run time performance evaluations of our algorithms on space-like hardware would be necessary in assessing their true benefit to spacecraft planning in practice.


\appendices
\crefalias{section}{appendixsection}

\section{Optimal Circularization Under Impulsive CWH Dynamics}
\label{appendix: Optimal Circularization}
As detailed in \cref{subsec: CWH CAM Policy Intro}, a vital component of our CAM policy is the generation of one-burn minimal-\fuel transfers to circular orbits above or below the target spacecraft orbit.  Assuming a failure occurs at state $\Vecx\left(t\subfail\right) = \Vecx\subfail$, the problem we wish to solve to satisfy \cref{def: Finite-Time Traj Safety} is:
\begin{equation*}
	\begin{optproblem}
		\given{\text{Failure state } \Vecx\subfail, \text{and CAM } \Vecu\subCAM\left(t\subfail \leq \tCAM < \horizonTime\supminus\right) \triangleq \VecZeros, \Vecu\subCAM\left(\horizonTime\right) \triangleq \VecDeltaV\subcirc\left(\Vecx\left(\horizonTime\right)\right)}
		\minimize[\horizonTime]{\DeltaV\subcirc^2(\horizonTime)}
		\subjectto
			\constraint[\text{Initial Condition}]{\Vecx\subCAM\left(t\subfail\right) = \Vecx\subfail}
			\constraint[\text{Invariant Set Termination}]{\Vecx\subCAM(\horizonTime\supplus) \in \Xinvariant}
			\constraint[\text{System Dynamics}]{\dot{\Vecx}\subCAM\left(\tCAM\right) = \StateTransitionFcn\left( \Vecx\subCAM\left(\tCAM\right), \VecZeros, \tCAM \right), \text{for all } t\subfail \leq \tCAM \leq \horizonTime}
			\constraint[\text{KOZ Collision Avoidance}]{\Vecx\subCAM\left(\tCAM\right) \not\in \XsubKOZ, \text{for all } t\subfail \leq \tCAM \leq \horizonTime}
	\end{optproblem}
\end{equation*}
Due to the analytical descriptions of state transitions as given by \cref{eqn: CWH Impulsive Solution}, 
it is a straightforward task to express the decision variable $\horizonTime$, invariant set constraint, and objective function analytically in terms of $\theta(\tCAM) = \omegaRef (\tCAM - \ti)$, the polar angle of the target spacecraft.
The problem is therefore one-dimensional in terms of $\theta$.  We can reduce the invariant set termination constraint to an invariant set positioning constraint if we ensure the spacecraft ends up at a position inside $\Xinvariant$ and circularize the orbit, since $\Vecx(\theta\subcirc\supplus) = \Vecx(\theta\subcirc\supminus) + \left[\begin{smallmatrix} \VecZeros \\ \VecDeltaV\subcirc\left(\theta\subcirc\right) \end{smallmatrix}\right] \in \Xinvariant$.  Denote $\theta\subcirc = \omegaRef (\horizonTime - \ti)$ as the target anomaly at which we enforce circularization.  Now, suppose the failure state $\Vecx(t)$ satisfies the collision avoidance constraint with the KOZ (otherwise the CAM is infeasible and we conclude $\Vecx$ is unsafe).  We can set $\theta\submin = \omegaRef(\tCAM - \ti)$ and integrate the coasting dynamics forward until the chaser touches the boundary of the KOZ ($\theta\submax = \theta_{\mathrm{collision}}\supminus$) or until we have reached one full orbit ($\theta\submax = \theta\submin + 2\pi$) such that, between these two bounds, the CAM trajectory satisfies the dynamics and contains only the coasting segment outside of the KOZ.  Replacing the dynamics and collision avoidance constraints with the bounds on $\theta$ as a box constraint, the problem now reduces to:
\nomenclature[B08b]{$\theta\subcirc$}{True anomaly at which to circularize a Collision-Avoidance-Maneuver orbit}%
\nomenclature[B08b]{$\theta\submin, \theta\submax$}{True anomaly search bounds for Collision-Avoidance-Maneuver circularization}%
\begin{equation*}
	\begin{optproblem}
		\minimize[\theta\subcirc]{\DeltaV\subcirc^2\left(\theta\subcirc\right)}
		\subjectto
		\constraint[\text{Theta Bounds}]{ \theta\submin \leq \theta\subcirc \leq \theta\submax }
		\constraint[\text{Invariant Set Positioning}]{ \posCrossTrack^2\left(\theta\subcirc\supminus\right) \geq \semiaxis_{\posCrossTrack}^2 }
	\end{optproblem}
\end{equation*}

Restricting our search range to $\theta \in \closedinterval{\theta\submin}{\theta\submax}$, this is a function of one variable and one constraint, something we can easily optimize analytically using the method of Lagrange multipliers.  To solve, we seek to minimize the Lagrangian, $\Lagrangian = \DeltaV\subcirc^2 + \lambda \InequalityConstraint\subcirc$, where $\InequalityConstraint\subcirc(\theta) = \semiaxis_{\posCrossTrack}^2 - \posCrossTrack^2\left(\theta\subcirc\supminus\right)$.  There are two cases to consider:
\nomenclature[B11b]{$\lambda$}{Lagrange multiplier}%
\nomenclature[AL ]{$\Lagrangian$}{Lagrangian}%
\paragraph{Case 1: Inactive Invariant Set Positioning Constraint}
\newcommand{\velCrossTrackSinThetaTerm}{\left( 3 \omegaRef \posCrossTrack\subfail + 2 \velInTrack\subfail \right)}
\newcommand{\velCrossTrackCosThetaTerm}{\velCrossTrack\subfail}
\newcommand{\planarSinTwoThetaTerm}{\frac{3}{4} \velCrossTrackSinThetaTerm^2 - \frac{3}{4} \velCrossTrackCosThetaTerm^2}
\newcommand{\planarCosTwoThetaTerm}{\frac{3}{2} \velCrossTrackCosThetaTerm \velCrossTrackSinThetaTerm }
\newcommand{\sinTwoThetaTerm}{\planarSinTwoThetaTerm + \omegaRef^2 \posOutofPlane\subfail^2 - \velOutofPlane\subfail^2}
\newcommand{\cosTwoThetaTerm}{\planarCosTwoThetaTerm - 2 \omegaRef \velOutofPlane\subfail \posOutofPlane\subfail}
We set $\lambda = 0$ such that $\Lagrangian = \DeltaV\subcirc^2$.  Candidate optimizers $\theta\supopt$ must satisfy $\grad[\theta] \Lagrangian\left(\theta\supopt\right) = 0$.  Taking the gradient of $\Lagrangian$,
\begin{align*}
	\grad[\theta] \Lagrangian
		&= \partialderiv[]{\DeltaV\subcirc^2}{\theta}
		= \left[ \sinTwoThetaTerm \right] \sin 2\theta
			\\ & \quad + \left[ \cosTwoThetaTerm \right] \cos 2\theta
\end{align*}
and setting $\grad[\theta] \Lagrangian\left(\theta\supopt\right) = 0$, we find that:
\begin{equation*}
	\tan 2\theta\supopt = \frac{-\left( \cosTwoThetaTerm \right)}{ \sinTwoThetaTerm }
\end{equation*}

Denote the set of candidate solutions that satisfy Case 1 by $\Theta\supopt_1$.

\paragraph{Case 2: Active Invariant Set Positioning Constraint}
Here the chaser attempts to circularize its orbit at the boundary of the zero-thrust RIC shown in \cref{subfig: CWH Planar Zero-Thrust RIC}.  The positioning constraint is active, and therefore $\InequalityConstraint\subcirc\left(\theta\right) = \semiaxis_{\posCrossTrack}^2 - \posCrossTrack^2\left(\theta\subcirc\supminus\right) = 0$.  This is equivalent to finding where the coasting trajectory from $\Vecx(t)$ crosses $\posCrossTrack(\theta) = \pm \semiaxis_{\posCrossTrack}$ for $\theta \in \closedinterval{\theta\submin}{\theta\submax}$.  This can be achieved using standard root-finding algorithms.  Denote the set of candidate solutions that satisfy Case 2 by $\Theta\supopt_2$.

\paragraph{Solution to the Minimal-Cost Circularization Burn}
The global optimizer $\theta\supopt$ either lies on the boundary of the box constraint, at an unconstrained optimum ($\theta \in \Theta\supopt_1$), or at the boundary of the zero-thrust RIC ($\theta \in \Theta\supopt_2$), all of which are economically obtained through either numerical integration or root-finding solver.  Therefore, the minimal-cost circularization burn time $\horizonTime\supopt$ satisfies:
\begin{equation*}
	\theta\supopt = \omegaRef\left(\horizonTime\supopt - \ti\right) = \argmin_{\theta \in \left\{\theta\submin, \theta\submax\right\} \Union \Theta\supopt_1 \Union \Theta\supopt_2 } \DeltaV\subcirc^2\left(\theta\right)
\end{equation*}
\nomenclature[B08b ]{$\theta\supopt$}{Optimal true anomaly for Collision-Avoidance-Maneuver circularization}%
If no solution exists (which can happen if and only if $\Vecx\subfail$ starts inside the KOZ), the circularization CAM is declared unsafe.  Otherwise, the CAM is saved for future trajectory feasibility verification.

\section{Intermediate Results for the \FMTstar Optimality Proof}
\label{appendix: Useful Results}
We report here a number of useful lemmas concerning bounds on the trajectory costs between samples which are used throughout the asymptotic optimality proof for \FMTstar in \cref{sec: Approach}.
%
We begin with the proof of \cref{lem: CWH Cost Bounds}, which relates the \fuel-burn cost function \cref{eqn: CWH Cost Function} between points $\Vecx\subnaught$ and $\Vecx\subf$ to the norm of the stacked $\DeltaV$-vector $\norm{\MatDeltaV} = \norm[\GramianInv]{\Vecx\subf - \MatPhi(\tf,\ti) \Vecx\subnaught}$.  We then provide a lemma bounding the sizes of the minimum and maximum eigenvalues of $\Gramian$, useful for bounding reachable volumes from $\Vecx\subnaught$.  Finally, we prove \cref{lem: perturbed cost bounds} which forms the basis of our asymptotic optimality analysis for \FMTstar.  Here $\MatPhi(\tf,\ti) = e^{\MatA \ManeuverDuration}$ is the state transition matrix, $\ManeuverDuration = \tf - t\subnaught$ is the maneuver duration, and $\Gramian$ is the $\NBurns = 2$ impulse Gramian matrix:
\begin{equation}
	\label{eqn: gramian}
	\Gramian(\ManeuverDuration) = \MatPhi_v \MatPhiInv_v = \left[\begin{array}{cc} e^{\MatA \ManeuverDuration} \MatB & \MatB \end{array}\right] \transpose{\left[\begin{array}{cc} e^{\MatA \ManeuverDuration} \MatB & \MatB \end{array}\right]},
\end{equation}
where $\MatPhi_v(t,\left\{\BurnTime_i\right\}_i)$ is the aggregate $\DeltaV$ transition matrix corresponding to burn times $\left\{\BurnTime_i\right\}_i = \left\{\ti, \tf\right\}$.
%
\LemmaFuelCostBounds*
\begin{proof}
	For the upper bound, note that by the Cauchy-Schwarz inequality we have
	$ \CostFcn = \norm{\VecDeltaV_1} \cdot 1 + \norm{\VecDeltaV_2} \cdot 1 \leq \sqrt{\norm{\VecDeltaV_1}^2 + \norm{\VecDeltaV_2}^2} \cdot \sqrt{1^2 + 1^2}$.  That is, $\CostFcn \leq \sqrt{2}\norm{\MatDeltaV}$.  Similarly, for the lower bound, note that: $\CostFcn = \sqrt{\left(\norm[]{\VecDeltaV_1} + \norm{\VecDeltaV_2}\right)^2} \geq \sqrt{\norm[]{\VecDeltaV_1}^2 + \norm[]{\VecDeltaV_2}^2} = \norm[]{\MatDeltaV}$.
\end{proof}
\begin{lemma}[Bounds on Gramian Eigenvalues]
	Let $\ManeuverDuration\submax$ be less than one orbital period for the system dynamics of \cref{sec:sysDyn}, and let $\Gramian(\ManeuverDuration)$ be defined as in \cref{eqn: gramian}. Then there exist constants $\GramianEigsLBConst,\GramianEigsUBConst > 0$ such that $\eigmin{\Gramian(\ManeuverDuration)} \geq \GramianEigsLBConst \ManeuverDuration^2$ and $\eigmax{\Gramian(\ManeuverDuration)} \leq \GramianEigsUBConst$ for all $\ManeuverDuration \in \rightclosedinterval{0}{\ManeuverDuration\submax}$.
	\label{lem: Gramian Eig Bounds}
\end{lemma}
\nomenclature[B11b]{$\lambda$}{Eigenvalue}%
\nomenclature[AMm]{$\GramianEigsLBConst, \GramianEigsUBConst$}{Positive constants in Gramian matrix eigenvalue bounds}%
\begin{proof}
	We bound the maximum eigenvalue of $\Gramian$ through norm considerations, yielding $\eigmax{\Gramian(\ManeuverDuration)} \leq \left(\norm{e^{\MatA \ManeuverDuration} \MatB} + \norm{\MatB}\right)^2 \leq \left(e^{\norm{\MatA} \ManeuverDuration\submax} + 1\right)^2$,
	and take $\GramianEigsUBConst = \left(e^{\norm{\MatA}\ManeuverDuration\submax} + 1\right)^2$.  As long as $\ManeuverDuration\submax$ is less than one orbital period, $\Gramian(\ManeuverDuration)$ only approaches singularity near $\ManeuverDuration = 0$ \cite{KA-SRV-PG-JH-LB:09}. Explicitly Taylor-expanding $\Gramian(\ManeuverDuration)$ about $\ManeuverDuration = 0$ reveals that $\eigmin{\Gramian(\ManeuverDuration)} = \sidefrac{\ManeuverDuration^2}{2} + \BigO{\ManeuverDuration^3}$ for small $\ManeuverDuration$, and thus $\eigmin{\Gramian(\ManeuverDuration)} = \BigOmega{\ManeuverDuration^2}$ for all $\ManeuverDuration \in \rightclosedinterval{0}{\ManeuverDuration\submax}$.
\end{proof}
\LemmaPerturbedSteering*
\nomenclature[B07b]{$\eta$}{Positive constant bounding initial condition perturbations in \FMTstar AO proof}%
\begin{proof}
	For bounding the perturbed cost, we consider the two cases separately.\\
	\noindent\underline{Case 1: $\ManeuverDuration = 0.$} Then 2-impulse steering degenerates to a single impulse $\VecDeltaV$; that is, $\Vecx\subf = \Vecx\subnaught + \MatB \VecDeltaV$ with $\norm{\VecDeltaV} = \CostFcn$.  To aid in the ensuing analysis, we write the position and velocity components of each state $\Vecx = \transpose{\left[ \VecrTranspose\ \VecvTranspose\right]}$ as $\Vecr = \left[\,\Identity\ \MatZeros\,\right] \Vecx$ and $\Vecv = \left[\,\MatZeros\ \Identity\,\right] \Vecx$. Note that since $\ManeuverDuration = 0$, we have $\Vecr\subf = \Vecr\subnaught$ and $\Vecv\subf = \Vecv\subnaught + \VecDeltaV$.  We pick the perturbed steering duration $\tilde \ManeuverDuration = \CostFcn^2$ (which will provide an upper bound on the optimal steering cost) and expand the steering system (\cref{eqn: CWH Impulsive Solution}) for small time $\tilde \ManeuverDuration$ as
	\begin{align}
		\Vecr\subf + \delta\Vecr\subf &= \Vecr\subnaught + \delta\Vecr\subnaught + \tilde \ManeuverDuration(\Vecv\subnaught + \delta\Vecv\subnaught + \widetilde\VecDeltaV_1) + \BigO{\tilde \ManeuverDuration^2} \label{eq:r}
		\\ \Vecv\subf + \delta\Vecv\subf &= \Vecv\subnaught + \delta\Vecv\subnaught + \widetilde\VecDeltaV_1 + \widetilde\VecDeltaV_2 + \tilde \ManeuverDuration\left(\MatA_{21}(\Vecr\subnaught + \delta\Vecr\subnaught) + \MatA_{22}(\Vecv\subnaught + \delta\Vecv\subnaught + \widetilde\VecDeltaV_1) \right) + \BigO{\tilde \ManeuverDuration^2} \label{eq:v}
	\end{align}
	where $\MatA_{21} = \left[ \begin{smallmatrix} 3\omegaRef^2 & 0 & 0\\ 0 & 0 & 0\\ 0 & 0 & -\omegaRef^2 \end{smallmatrix} \right]$ and $\MatA_{22} = \left[ \begin{smallmatrix} 0 & 2\omegaRef & 0 \\ -2\omegaRef & 0 & 0 \\ 0 & 0 & 0 \end{smallmatrix} \right]$.  Solving \cref{eq:r} for $\widetilde\VecDeltaV_1$ to first order, we find 
	$\widetilde\VecDeltaV_1 = \inverse{\tilde \ManeuverDuration} (\delta\Vecr\subf - \delta\Vecr\subnaught) - \Vecv\subnaught - \delta\Vecv\subnaught + \BigO{\tilde \ManeuverDuration}$.
	By selecting $\delta \Vecx_c = \transpose{\left[\tilde \ManeuverDuration \VecvTranspose\subnaught\ \MatZerosTranspose\right]}$ (note: $\norm{\delta \Vecx_c} = \CostFcn^2 \norm{\Vecv\subnaught} = \BigO{\CostFcn^2}$) and supposing that $\norm[]{\PerturbationXi} \leq \eta \CostFcn^3$ and $\norm[]{\PerturbationXf - \delta \Vecx_c} \leq \eta \CostFcn^3$, we have that:
	\begin{equation*}
		\norm<normal>{\widetilde\VecDeltaV_1} \leq \CostFcn^{-2} (\norm{\PerturbationXi} + \norm{\PerturbationXf - \delta \Vecx_c}) + \norm{\PerturbationXi} + \BigO{\CostFcn^2}
			= 2\eta \CostFcn + \BigO{\CostFcn^2}.
	\end{equation*}
	Now solving \cref{eq:v} for $\widetilde\VecDeltaV_2 = \VecDeltaV + (\delta\Vecv\subf - \delta\Vecv\subnaught) - \widetilde\VecDeltaV_1 + \BigO{\CostFcn^2}$, and taking norm of both sides:
	\begin{equation*}
		\norm<normal>{\widetilde\VecDeltaV_2} \leq \norm{\VecDeltaV} + \left(\norm{\PerturbationXi} + \norm{\PerturbationXf - \delta \Vecx_c}\right) + 2 \eta \CostFcn + \BigO{\CostFcn^2}
		\leq \CostFcn + 2 \eta \CostFcn + \BigO{\CostFcn^2}.
	\end{equation*}
	Therefore the perturbed cost satisfies:
	\begin{equation*}
		\CostFcn(\PerturbedXi, \PerturbedXf) \leq \norm<normal>{\widetilde\VecDeltaV_1} + \norm<normal>{\widetilde\VecDeltaV_2} \leq \CostFcn\left(1 + 4\eta + \BigO{\CostFcn}\right).
	\end{equation*}
	\\
	\noindent\underline{Case 2: $\ManeuverDuration > 0.$} We pick $\tilde \ManeuverDuration = \ManeuverDuration$ to compute an upper bound on the perturbed cost. Applying the explicit form of the steering control (\cref{eqn: CWH Optimal 2PBVP with N=2}) along with the norm bound $\norm{\MatPhiInv_v} = \eigmin{\Gramian}^{-1/2} \leq \GramianEigsLBConst^{-1/2} \inverse{\ManeuverDuration}$ from \cref{lem: Gramian Eig Bounds}, we have:
	\begin{align*}
		\CostFcn(\PerturbedXi, \PerturbedXf) &\leq \norm{\MatPhiInv_v\left(\tf, \{\ti, \tf\}\right) \left(\PerturbedXf - \MatPhi\left(\tf,\ti\right) \PerturbedXi\right)}
		\\ &\leq \norm{\MatPhiInv_v \left(\Vecx\subf - \MatPhi \Vecx\subnaught\right)} + \norm{\MatPhiInv_v \PerturbationXf} + \norm{\MatPhiInv_v \MatPhi \PerturbationXi}
		\\ &\leq \CostFcn + \GramianEigsLBConst^{-1/2} \inverse{\ManeuverDuration} \norm{\PerturbationXf} + \GramianEigsLBConst^{-1/2} \inverse{\ManeuverDuration} e^{\norm{\MatA}\ManeuverDuration\submax} \norm{\PerturbationXi}
		\leq \CostFcn\left(1 + \BigO{\eta \CostFcn^2 \inverse{\ManeuverDuration} }\right).
	\end{align*}
	In both cases, the deviation of the perturbed steering trajectory (call it $\tilde \Vecx(t)$) from its closest point on the original trajectory is bounded (quite conservatively) by the maximum propagation of the difference in initial conditions; that is, the initial disturbance $\PerturbationXi$ plus the difference in intercept burns $\widetilde\VecDeltaV_1 - \VecDeltaV_1$, over the maximum maneuver duration $\ManeuverDuration\submax$. Thus,
	\begin{equation*}
		\norm{\tilde \Vecx(t) - \Vecx(t)} \leq e^{\norm{\MatA} \ManeuverDuration\submax}\left(\norm{\PerturbationXi} + \norm<normal>{\widetilde\VecDeltaV_1} + \norm{\VecDeltaV_1}\right)
			\leq e^{\norm{\MatA} \ManeuverDuration\submax}\left(\eta \CostFcn^3 + 2\CostFcn + \LittleO{\CostFcn}\right)
			= \BigO{\CostFcn}
	\end{equation*}
	where we have used $\norm{\VecDeltaV_1} \leq \CostFcn$ and $\norm<normal>{\widetilde\VecDeltaV_1} \leq \CostFcn\left(\PerturbedXi, \PerturbedXf\right) \leq \CostFcn + \LittleO{\CostFcn}$ from our above arguments.
\end{proof}


\acknowledgments
This work was supported by \MarcoPavoneNASAECFGrant.

\iftoggle{AIAAjournal}{\section*{\refname}}{}
\bibliography{\bibfiles}

\end{document}

%% file: preamble_confpaper.tex
\usepackage{etoolbox}					

\providetoggle{DEFAULTpaper} 	\toggletrue{DEFAULTpaper}
\providetoggle{AASpaper} 		\togglefalse{AASpaper}
\providetoggle{AIAApaper} 		\togglefalse{AIAApaper}
\providetoggle{AIAAconf}		\togglefalse{AIAAconf}
\providetoggle{AIAAjournal}		\togglefalse{AIAAjournal}
\providetoggle{Coursework} 		\togglefalse{Coursework}
\providetoggle{IEEEpaper} 		\togglefalse{IEEEpaper}
\providetoggle{IEEEconf}		\togglefalse{IEEEconf}
\providetoggle{IEEEjournal} 	\togglefalse{IEEEjournal}
\providetoggle{REVTEXjournal} 	\togglefalse{REVTEXjournal}

\makeatletter
\@ifclassloaded{AAS}{%
	\toggletrue{AASpaper}%
	\togglefalse{DEFAULTpaper}%
}{}%
\@ifclassloaded{AIAA}{%
	\toggletrue{AIAAjournal}%
	\toggletrue{AIAApaper}%
	\togglefalse{DEFAULTpaper}%
}{}%
\@ifclassloaded{aiaa-tc}{%
	\toggletrue{AIAAconf}%
	\toggletrue{AIAApaper}%
	\togglefalse{DEFAULTpaper}%
}{}%
\@ifclassloaded{coursework}{%
	\toggletrue{Coursework}%
}{}%
\@ifclassloaded{ieeeconf}{%
	\toggletrue{IEEEconf}%
	\toggletrue{IEEEpaper}%
	\togglefalse{DEFAULTpaper}%
}{}%
\@ifclassloaded{IEEEAerospaceCLS}{%
	\toggletrue{IEEEconf}%
	\toggletrue{IEEEpaper}%
	\togglefalse{DEFAULTpaper}%
}{}%
\@ifclassloaded{IEEEtrans}{%
	\toggletrue{IEEEjournal}%
	\toggletrue{IEEEpaper}%
	\togglefalse{DEFAULTpaper}%
}{}%
\@ifclassloaded{revtex4-1}{
	\toggletrue{REVTEXjournal}%
}{}
\makeatother

\usepackage{tabularx}					
\usepackage{booktabs}					
\usepackage{pbox}						
\usepackage{comment}					
\usepackage[dvipsnames]{xcolor}			
\usepackage{datetime}					
\usepackage{pgffor}						
\usepackage{url}						
\usepackage{setspace}
\usepackage{varwidth}

\usepackage{graphicx}					
\usepackage[space]{grffile}				
\usepackage{afterpage}					
\usepackage{hyperref}					
\usepackage[noend]{algpseudocode}		
\usepackage{algorithm}					

\usepackage{mymath}						
\usepackage{mytext}						

\iftoggle{DEFAULTpaper}{%
	\usepackage{parskip}					
	\usepackage[margin=1in]{geometry}		
	\usepackage[labelfont=bf,font={small}]{caption}	
	\usepackage{subcaption}					

	\bibliographystyle{unsrt}				
	\setcounter{secnumdepth}{4}				
	\setcounter{tocdepth}{4}				
}{}
\iftoggle{REVTEXjournal}{}{%
	\usepackage{cite}						

	\usepackage{fixltx2e}					
	\usepackage{longtable}
	\usepackage{microtype}					
	\usepackage{multicol}					
}

\iftoggle{AASpaper}{%
	\usepackage{footnpag}			      	

	\usepackage[caption=false]{subfig}
	\bibliographystyle{AAS_publication}		
}{}

\iftoggle{AIAApaper}{%
	\usepackage[caption=false]{subfig}
	\bibliographystyle{aiaa}
}{}

\iftoggle{AIAAconf}{%
}{}

\iftoggle{AIAAjournal}{%
}{}

\iftoggle{Coursework}{%
	\usepackage{subcaption}					
}{}

\iftoggle{IEEEpaper}{%
}{}

\iftoggle{IEEEconf}{%
	\usepackage[labelformat=simple]{subcaption}		

	\DeclareCaptionLabelSeparator{periodspace}{.\quad}
	\captionsetup{font=footnotesize,labelsep=periodspace,singlelinecheck=false}
	\captionsetup[sub]{font=footnotesize,singlelinecheck=true}
	
}{}

\iftoggle{IEEEjournal}{%
	\usepackage[caption=false]{subfig}
	\bibliographystyle{IEEEtran}
}{}

\iftoggle{REVTEXjournal}{%
}{}

\iftoggle{REVTEXjournal}{}{
	\algrenewcommand\algorithmicdo{}
	\algrenewcommand\algorithmicthen{}
	\algrenewcommand\algorithmicprocedure{}
	\MakeRobust{\Call}

	\makeatletter
	\newcommand{\algindent}[1][1]{%
		\setlength\@tempdima{\algorithmicindent}%
		\hskip\dimexpr#1\@tempdima\relax%
	}
	\makeatother
}

\makeatletter
\@ifpackageloaded{subfig}{%
	\usepackage{adjustbox}
	\usepackage{xparse}

	\define@choicekey[mypreamble]{subcaptionbox}{horizalign}[\uservalue\usernum]{c,l,r,s}[c]{%
		\ifcase\usernum\relax
			\centering%
		\or
			\raggedright%
		\or
			\raggedleft%
		\or
		\fi
	}

	\NewDocumentCommand{\subcaptionbox}{ O{} m o o m }{
		\IfNoValueTF{#3}{%
			\subfloat[#1][#2]{#5}%
		}{%
			\IfNoValueTF{#4}{%
				\adjustbox{valign=t}{\begin{minipage}{#3}\subfloat[#1][#2]{#5}\end{minipage}}%
			}{%
				\adjustbox{valign=t}{\begin{minipage}{#3}%
					\setkeys[mypreamble]{subcaptionbox}{horizalign=#4}%
					\subfloat[#1][#2]{#5}%
				\end{minipage}}%
			}
		}
	}
}{}
\makeatother

%% file: optimality.tex
\newcommand{\PerturbedTraj}{\tilde{\Vecx}(t)} 
\newcommand{\PerturbationXi}{\delta \Vecx\subnaught}
\newcommand{\PerturbationXf}{\delta \Vecx\subf}
\newcommand{\PerturbedXi}{\tilde{\Vecx}\subnaught} 
\newcommand{\PerturbedXf}{\tilde{\Vecx}\subf} 
\newcommand{\CostToCome}{c}
\renewcommand{\finitesequence}[3][i=0]{\left\{#3\right\}_{#1}^{#2}}
For sampling-based algorithms, \emph{asymptotic optimality} refers to the property that
as the number of samples $\nSamples \to \infty$, 
the cost of the trajectory (a.k.a.\ ``path'') returned by the planner 
approaches that of the optimal cost.  Here a proof is presented showing asymptotic optimality for the planning algorithm and problem setup used in this paper.  We note that while CWH dynamics are the primary focus of this work, the following proof methodology extends to any general linear system controlled by a finite sequence of impulsive actuations, whose fixed-duration 2-impulse steering problem is uniquely determined (\eg a wide array of second-order control systems).

The proof proceeds analogously to \cite{LJ-ES-AC-ea:15} by showing that it is always possible to construct an approximate path from points in $\SampleSet$ that closely follows the optimal path. Similarly to \cite{LJ-ES-AC-ea:15}, we will make use here of a concept called the $\ell_2$-\emph{dispersion} of a set of points, which upper bounds how far away a point in $\SubSpace$ can be from its nearest point in $\SampleSet$ as measured by the $\ell_2$-norm.
\begin{definition}[$\ell_2$-dispersion]
	For a finite, non-empty set $\SampleSet$ of points in a $\Dimension$-dimensional compact Euclidean subspace $\SubSpace$ with positive Lebesgue measure, its $\ell_2$-dispersion $\DispersionFcn\left(\SampleSet\right)$ is defined as:
	\begin{align*}
		\DispersionFcn\left(\SampleSet\right)
			&\triangleq \sup_{\Vecx \in \SubSpace} \min_{\Vecs \in \SampleSet} \norm{\Vecs - \Vecx}
			\\ &= \sup \left\{R > 0 \suchthat \exists \Vecx \in \SubSpace \text{ with } \ball{\Vecx}{R} \intersect \SampleSet = \nullset \right\},
	\end{align*}
	where $\ball{\Vecx}{R}$ is a Euclidean ball with radius $R$ centered at state $\Vecx$.
	\label{def: l2 dispersion}
\end{definition}

We also require a means for quantifying the deviation that small endpoint pertubations can bring about in the 2-impulse steering control. This result is necessary to ensure that the particular placement of the points of $\SampleSet$ is immaterial; only its low-dispersion property matters.

\begin{restatable}[Steering with Perturbed Endpoints]{lemma}{LemmaPerturbedSteering}
	For a given steering trajectory $\Vecx(t)$ with initial time $\ti$ and final time $\tf$, let $\Vecx\subnaught \coloneqq \Vecx(\ti)$, $\Vecx\subf \coloneqq \Vecx(\tf)$, $\ManeuverDuration \coloneqq \tf - \ti$, and $\CostFcn \coloneqq \CostFcn\left(\Vecx\subnaught, \Vecx\subf\right)$.  Consider now the steering trajectory $\PerturbedTraj$ between perturbed start and end points $\PerturbedXi = \Vecx\subnaught + \PerturbationXi$ and $\PerturbedXf = \Vecx\subf + \PerturbationXf$.
	\\ \noindent\underline{Case 1: $\ManeuverDuration = 0$.}  There exists a perturbation center $\delta \Vecx_c$ (consisting of only a position shift) with $\norm{\delta \Vecx_c} = \BigO{\CostFcn^2}$ such that if $\norm{\PerturbationXi} \leq \eta \CostFcn^3$ and $\norm{\PerturbationXf - \delta \Vecx_c} \leq \eta \CostFcn^3$ then $\CostFcn\left(\PerturbedXi, \PerturbedXf\right) \leq \CostFcn\left(1 + 4 \eta + \BigO{\CostFcn}\right)$ and the spatial deviation of the perturbed trajectory from $\Vecx(t)$ is $\BigO{\CostFcn}$.
	\\ \noindent\underline{Case 2: $\ManeuverDuration > 0$.} If $\norm[]{\PerturbationXi} \leq \eta \CostFcn^3$ and $\norm[]{\PerturbationXf} \leq \eta \CostFcn^3$ then $\CostFcn\left(\PerturbedXi, \PerturbedXf\right) \leq \CostFcn\left(1 + \BigO{\eta \CostFcn^2 \inverse{\ManeuverDuration} }\right)$ and the spatial deviation of the perturbed trajectory from $\Vecx(t)$ is $\BigO{\CostFcn}$.
	\label{lem: perturbed cost bounds}
\end{restatable}
\begin{proof}
	For the proof, see \cref{appendix: Useful Results}.
\end{proof}

We are now in a position to prove that the cost of the trajectory returned by \FMTstar approaches that of an optimal trajectory as the number of samples $\nSamples \to \infty$.  The proof proceeds in two steps.  First, we establish that there is a sequence of waypoints in $\SampleSet$ that are placed closely along the optimal path and approximately evenly-spaced in cost.  Then we show that the existence of these waypoints guarantees that \FMTstar finds a path with a cost close to that of the optimal cost.  The theorem and proof combine elements from Theorem 1 in \cite{LJ-ES-AC-ea:15} and Theorem IV.6 from \cite{ES-LJ-MP:15b}.

\begin{definition}[Strong $\delta$-Clearance]
	A trajectory $\Vecx(t)$ is said to have \emph{strong} $\delta$-clearance if, for some $\delta > 0$ and all $t$, the Euclidean distance between $\Vecx(t)$ and any point in $\Xobs$ is greater than $\delta$.
	\label{def: Strong Delta Clearance}
\end{definition}

\begin{theorem}[Existence of Waypoints near an Optimal Path]
	Let $\Vecx\supopt(t)$ be a feasible trajectory for the motion planning problem \cref{eqn: CWH Optimal Steering} with strong $\delta$-clearance, let $\Vecu\supopt(t) = \finiteseries[i=1]{\NBurns} \VecDeltaV\supopt_i \cdot \delta\left(t - \BurnTime\supopt_i\right)$ be its associated control trajectory, and let $\CostFcn\supopt$ be its cost. Furthermore, let $\SampleSet \union \{\Vecx\subinit\}$ be a set of $\nSamples \in \naturals$ points from $\Xfree$ with dispersion $\DispersionFcn\left(\SampleSet\right) \leq \gamma \nSamples^{\sidefrac{-1}{\Dimension}}$. Let $\epsilon > 0$, and choose $\CostThreshold = 4 \left(\sidefrac{\gamma \nSamples^{\sidefrac{-1}{\Dimension}}}{\epsilon}\right)^{\sidefrac{1}{3}}$. Then, provided that $n$ is sufficiently large, there exists a sequence of points $\finitesequence[k=0]{K}{\Vecy_k}$, $\Vecy_k \in \SampleSet$ such that $\CostFcn(\Vecy_k, \Vecy_{k+1}) \leq \CostThreshold$, the cost of the path $\Vecy(t)$ made by joining all of the steering trajectories between $\Vecy_k$ and $\Vecy_{k+1}$ is $\sum_{k=0}^{K-1} \CostFcn(\Vecy_k, \Vecy_{k+1}) \leq (1 + \epsilon)\CostFcn\supopt$, and $\Vecy(t)$ is itself strong $(\delta/2)$-clear.
	\label{th: waypoint existence}
\end{theorem}
\begin{proof}
	We first note that if $\CostFcn\supopt = 0$ then we can pick $\Vecy\subnaught = \Vecx\supopt(t\subnaught)$ and $\Vecy_1 = \Vecx\supopt(t\subf)$ as the only points in $\{\Vecy_k\}$ and the result is trivial. Thus assume that $\CostFcn\supopt > 0$.  Construct a sequence of times $\finitesequence[k=0]{K}{t_k}$ and corresponding points $\Vecx\supopt_k = \Vecx\supopt(t_k)$ spaced along $\Vecx\supopt(t)$ in cost intervals of $\sidefrac{\CostThreshold}{2}$. We admit a slight abuse of notation here in that $\Vecx\supopt(\BurnTime\supopt_i)$ may represent a state with any velocity along the length of the impulse $\VecDeltaV\supopt_i$; to be precise, pick $\Vecx\supopt\subnaught = \Vecx\subinit$, $t\subnaught = 0$, and for $k = 1, 2, \ldots$ define $j_k = \min\left\{j \suchthat \sum_{i=1}^j \norm{\VecDeltaV\supopt_i} > k \frac{\CostThreshold}{2}\right\}$	and select $t_k$ and $\Vecx\supopt_k$ as:
	\begin{align*}
		t_k &= \BurnTime\supopt_{j_k}
		\\ \Vecx\supopt_k &= \lim_{t \rightarrow t_k\supminus} \Vecx\supopt(t) + \left(k \frac{\CostThreshold}{2} - \sum_{i=1}^{j_k-1} \norm{\VecDeltaV\supopt_i} \right)\MatB \frac{\VecDeltaV\supopt_i}{\norm{\VecDeltaV\supopt_i}}.
	\end{align*}
	Let $K = \sidefrac{\ceil{\CostFcn\supopt}}{\left(\sidefrac{\CostThreshold}{2}\right)}$ and set $t_K = \tf$, $\Vecx\supopt_K = \Vecx\supopt(\tf)$. Since the trajectory $\Vecx\supopt(t)$ to be approximated is fixed, for sufficiently small $\CostThreshold$ (equivalently, sufficiently large $n$) we may ensure that the control applied between each $\Vecx\supopt_k$ and $\Vecx\supopt_{k+1}$ occurs only at the endpoints.  In particular this may be accomplished by choosing $n$ large enough so that $\CostThreshold < \min_i{\norm{\VecDeltaV\supopt_i}}$.  In the limit $\CostThreshold \rightarrow 0$, the vast majority of the 2-impulse steering connections between successive $\Vecx\supopt_k$ will be zero-time maneuvers (arranged along the length of each burn $\VecDeltaV\supopt_i$) with only $N$ positive-time maneuvers spanning the regions of $\Vecx\supopt(t)$ between burns. By considering this regime of $n$, we note that applying 2-impulse steering between successive $\Vecx\supopt_k$ (which otherwise may only approximate the performance of a more complex control scheme) requires cost no greater than that of $\Vecx\supopt$ itself along that step, \ie $\sidefrac{\CostThreshold}{2}$.

	We now inductively define a sequence of points $\finitesequence[k=0]{K}{\VecxHat\supopt_k}$ by $\VecxHat\supopt_0 = \Vecx\supopt_0$ and for each $k > 0$: (1) if $t_k = t_{k-1}$, pick $\VecxHat\supopt_k = \Vecx\supopt_k + \delta \Vecx_{c,k} + (\VecxHat\supopt_{k-1} - \Vecx\supopt_{k-1})$, where $\delta \Vecx_{c,k}$ comes from \cref{lem: perturbed cost bounds} for zero-time approximate steering between $\Vecx\supopt_{k-1}$ and $\Vecx\supopt_k$ subject to perturbations of size $\eps \CostFcn^3$; (2) otherwise if $t_k > t_{k-1}$, pick $\VecxHat\supopt_k = \Vecx\supopt_k + (\VecxHat\supopt_{k-1} - \Vecx\supopt_{k-1})$. The reason for defining these $\VecxHat\supopt_k$ is that the process of approximating each $\VecDeltaV\supopt_i$ by a sequence of small burns necessarily incurs some short-term position drift. Since $\delta \Vecx_{c,k} = \BigO{\CostThreshold^2}$ for each $k$, and since $K = \BigO{\CostThreshold^{-1}}$, the maximum accumulated difference satisfies $\max_k \norm{\VecxHat\supopt_k - \Vecx\supopt_k} = \BigO{\CostThreshold}$.

	For each $k$ consider the Euclidean ball centered at $\VecxHat\supopt_k$ with radius $\gamma \nSamples^{-\frac{1}{\Dimension}}$, \ie let $\ClosedBall_k \coloneqq \ClosedBall\left(\VecxHat\supopt_k, \gamma \nSamples^{-\frac{1}{\Dimension}}\right)$.  By \cref{def: l2 dispersion} and our restriction on $\SampleSet$, each $\ClosedBall_k$ contains at least one point from $\SampleSet$.  Hence for every $\ClosedBall_k$ we can pick a waypoint $\Vecy_k$ such that $\Vecy_k \in \ClosedBall_k \intersect \SampleSet$. Then $\norm{\Vecy_k - \hat\Vecx\supopt_k} \leq \gamma \nSamples^{-\frac{1}{\Dimension}} = \sidefrac{\epsilon(\CostThreshold/2)^3}{8}$ for all $k$, and thus by \cref{lem: perturbed cost bounds} (with $\eta = \sidefrac{\epsilon}{8}$) we have that:
	\begin{equation*}
		\CostFcn(\Vecy_k, \Vecy_{k+1}) \leq \frac{\CostThreshold}{2} \left(1 + \frac{\epsilon}{2} + \BigO{\CostThreshold}\right)
			\leq \frac{\CostThreshold}{2} (1 + \epsilon)
	\end{equation*}
	for sufficiently large $n$.  In applying \cref{lem: perturbed cost bounds} to Case 2 for $k$ such that $t_{k+1} > t_k$, we note that the $T^{-1}$ term is mitigated by the fact that there are only a finite number of burn times $\BurnTime\supopt_i$ along $\Vecx\supopt(t)$.  Thus for each such $k$, $t_{k+1} - t_k \geq \min_j (t_{j+1} - t_j) > 0$, so in every case we have $\CostFcn(\Vecy_k, \Vecy_{k+1}) \leq (\sidefrac{\CostThreshold}{2})(1 + \epsilon)$. That is, each steering segment connecting $\Vecy_k$ to $\Vecy_{k+1}$ approximates the cost of the corresponding $\Vecx\supopt_k$ to $\Vecx\supopt_{k+1}$ segment of $\Vecx\supopt(t)$ up to a multiplicative factor of $\epsilon$, and thus:
	\begin{equation*}
		\sum_{k=0}^{K-1} \CostFcn(\Vecy_k, \Vecy_{k+1}) \leq (1 + \epsilon)\CostFcn\supopt.
	\end{equation*}
	Finally, to establish that $\Vecy(t)$, the trajectory formed by steering through the $\Vecy_k$'s in succession, has sufficient obstacle clearance, we note that its distance from $\Vecx\supopt(t)$ is bounded by $\max_k \norm{\VecxHat\supopt_k - \Vecx\supopt_k} = \BigO{\CostThreshold}$ plus the deviation bound from \cref{def: l2 dispersion}, again $\BigO{\CostThreshold}$. For sufficiently large $n$, the total distance, $\BigO{\CostThreshold}$, will be bounded by $\sidefrac{\delta}{2}$, and thus $\Vecy(t)$ will have strong $(\delta/2)$-clearance.
\end{proof}

We now prove that \FMTstar is asymptotically optimal in the number of points $\nSamples$, provided the conditions required in Theorem \ref{th: waypoint existence} hold; note the proof is heavily based on Theorem VI.1 from \cite{ES-LJ-MP:15a}.
\begin{theorem}[Asymptotic Performance of \FMTstar]
	Let $\Vecx\supopt(t)$ be a feasible trajectory satisfying \cref{eqn: CWH Optimal Steering} with strong $\delta$-clearance and cost $\CostFcn\supopt$.  Let $\SampleSet \union \{\Vecx\subnaught\}$ be a set of $\nSamples \in \naturals$ samples from $\Xfree$ with dispersion $\DispersionFcn\left(\SampleSet\right) \leq \gamma \nSamples^{\sidefrac{-1}{\Dimension}}$.	 Finally, let $\CostFcn_\nSamples$ denote the cost of the path returned by \FMTstar with $\nSamples$ points in $\SampleSet$ while using a cost threshold $\CostThreshold(\nSamples) = \LittleOmega{\nSamples^{-1/3\Dimension}}$ and $\CostThreshold = \LittleO{1}$. (That is, $\CostThreshold(\nSamples)$ asymptotically dominates $\nSamples^{-1/3\Dimension}$ and is asymptotically dominated by $1$.)
	Then $\lim_{\nSamples \to \infty} \CostFcn_\nSamples \leq \CostFcn\supopt$.
	\label{th: fmt asymptotic optimality}
\end{theorem}
\begin{proof}
	Let $\epsilon > 0$. Pick $\nSamples$ sufficiently large so that $\sidefrac{\delta}{2} \geq \CostThreshold \geq 4 \left(\sidefrac{\gamma \nSamples^{\sidefrac{-1}{\Dimension}}}{\epsilon}\right)^{\sidefrac{1}{3}}$ such that \cref{th: waypoint existence} holds. That is, there exists a sequence of waypoints $\finitesequence[k=0]{K}{\Vecy_k}$ approximating $\Vecx\supopt(t)$ such that the trajectory $\Vecy(t)$ created by sequentially steering through the $\Vecy_k$ is strong $\sidefrac{\delta}{2}$-clear, whose connection costs satisfy $\CostFcn(\Vecy_k, \Vecy_{k+1}) \leq \CostThreshold$, and whose total cost satisfies $\sum_{k=0}^{K-1}\CostFcn(\Vecy_k,\Vecy_{k+1}) \leq (1+\epsilon)\CostFcn\supopt$. We show that \FMTstar recovers a path with cost at least as good as $\Vecy(t)$; that is, we show that $\lim_{\nSamples \to \infty} \CostFcn_\nSamples \leq \CostFcn\supopt.$
	
	Consider running \FMTstar to completion, and for each $\Vecy_k$, let $\CostToCome(\Vecy_k)$ denote the cost-to-come of $\Vecy_k$ in the generated graph (with value $\infty$ if $\Vecy_k$ is not connected).  
	We show by induction that:
	\begin{equation}
		\min(\CostToCome(\Vecy_m), \CostFcn_\nSamples) \leq \sum_{k=0}^{m-1}\CostFcn(\Vecy_k,\Vecy_{k+1})
		\label{eq: fmt induction}
	\end{equation}
	for all $m \in \finitelist[1]{K}$.  For the base case $m = 1$, we note by the initialization of \FMTstar on \cref{line: Tree Initialization} of \cref{alg: FMTstar} that $\Vecx\subinit$ is in $\Frontier$; therefore, by the design of \FMTstar (per \crefrange{line: Expand Tree Loop Begin}{line: Expand Tree Loop End}), every possible feasible connection is made between the first waypoint $\Vecy\subnaught = \Vecx\subinit$ and its neighbors.  Since $\CostFcn(\Vecy\subnaught, \Vecy_1) \leq \CostThreshold$ and the edge $(\Vecy\subnaught, \Vecy_1)$ is collision free, it is also in the \FMTstar graph. Then $\CostToCome(\Vecy_1) = \CostFcn(\Vecy\subnaught,\Vecy_1)$.  Now assuming that \cref{eq: fmt induction} holds for $m - 1$, then one of the following statements holds:
	\begin{enumerate}
		\item $\CostFcn_\nSamples \leq \sum_{k=0}^{m-2} \CostFcn(\Vecy_k, \Vecy_{k+1})$
		\item $\CostToCome(\Vecy_{m-1}) \leq \sum_{k=0}^{m-2} \CostFcn(\Vecy_k, \Vecy_{k+1})$ and \FMTstar ends before considering $\Vecy_m$.
		\item $\CostToCome(\Vecy_{m-1}) \leq \sum_{k=0}^{m-2} \CostFcn(\Vecy_k, \Vecy_{k+1})$ and $\Vecy_{m-1}\in \Frontier$ when $\Vecy_m$ is first considered
		\item $\CostToCome(\Vecy_{m-1}) \leq \sum_{k=0}^{m-2} \CostFcn(\Vecy_k, \Vecy_{k+1})$ and $\Vecy_{m-1} \notin \Frontier$ when $\Vecy_m$ is first considered.
	\end{enumerate}
	We now show for each case that our inductive hypothesis holds.
	
	\noindent\underline{Case 1}: $\CostFcn_\nSamples \leq \sum_{k=0}^{m-2}\CostFcn(\Vecy_k,\Vecy_{k+1}) \leq \sum_{k=0}^{m-1}\CostFcn(\Vecy_k,\Vecy_{k+1})$.
	
	\noindent\underline{Case 2}: Since at every step \FMTstar considers the node that is the endpoint of the path with the lowest cost, if \FMTstar ends before considering $\Vecy_m$, we have $\CostFcn_\nSamples \leq \CostToCome(\Vecy_m) \leq \CostToCome(\Vecy_{m-1}) + \CostFcn(\Vecy_{m-1},\Vecy_m) \leq \sum_{k=0}^{m-1}\CostFcn(\Vecy_k,\Vecy_{k+1})$.
	
	\noindent\underline{Case 3}: Since the neighborhood of $\Vecy_m$ is collision free by the clearance property of $\Vecy$, and since $\Vecy_{m-1}$ is a possible parent candidate for connection, $\Vecy_m$ will be added to the \FMTstar tree as soon as it is considered with $\CostToCome(\Vecy_m) \leq \CostToCome(\Vecy_{m-1}) + \CostFcn(\Vecy_{m-1},\Vecy_m) \leq \sum_{k=0}^{m-1}\CostFcn(\Vecy_k,\Vecy_{k+1})$.
	
	\noindent\underline{Case 4}: When $\Vecy_m$ is considered, it means there is a node $\Vecz \in \Frontier$ (with minimum cost-to-come through the \FMTstar tree) and $\Vecy_m \in \ReachableSet(\Vecz, \CostThreshold)$. Then $\CostToCome(\Vecy_m) \leq \CostToCome(\Vecz) + \CostFcn(\Vecz,\Vecy_m)$.  Since $\CostToCome(\Vecy_{m-1}) < \infty$, $\Vecy_{m-1}$ must be added to the tree by the time \FMTstar terminates.  Consider the path from $\Vecx\subinit$ to $\Vecy_{m-1}$ in the final \FMTstar tree, and let $\Vecw$ be the last vertex along this path, which is in $\Frontier$ at the time when $\Vecy_m$ is considered.  If $\Vecy_m \in \ReachableSet(\Vecw, \CostThreshold)$, \ie $\Vecw$ is a parent candidate for connection, then:
	\begin{alignat*}{2}
		\CostToCome(\Vecy_m) 	&\leq \CostToCome(\Vecw) + \CostFcn(\Vecw, \Vecy_m) 
								\\ &\leq \CostToCome(\Vecw) + \CostFcn(\Vecw, \Vecy_{m-1}) + \CostFcn(\Vecy_{m-1}, \Vecy_m)  
								\\ &\leq \CostToCome(\Vecy_{m-1}) + \CostFcn(\Vecy_{m-1}, \Vecy_m)
								\\ &\leq \sum_{k=0}^{m-1} \CostFcn(\Vecy_k,\Vecy_{k+1}).  
	\end{alignat*}
	Otherwise if $\Vecy_m \notin \ReachableSet(\Vecw, \CostThreshold)$, then $\CostFcn(\Vecw, \Vecy_m) > \CostThreshold$ and:
	\begin{alignat*}{2}
		\CostToCome(\Vecy_m)	&\leq \CostToCome(\Vecz) + \CostFcn(\Vecz,\Vecy_m)
								\\ &\leq \CostToCome(\Vecw) + \CostThreshold
								\\ &\leq \CostToCome(\Vecw) + \CostFcn(\Vecw, \Vecy_m)
								\\ &\leq \CostToCome(\Vecw) + \CostFcn(\Vecw, \Vecy_{m-1}) + \CostFcn(\Vecy_{m-1}, \Vecy_m)
								\\ &\leq \CostToCome(\Vecy_{m-1}) + \CostFcn(\Vecy_{m-1}, \Vecy_m)
								\\ &\leq \sum_{k=0}^{m-1} \CostFcn(\Vecy_k,\Vecy_{k+1}).
	\end{alignat*}
	where we used the fact that $\Vecw$ is on the path of $\Vecy_{m-1}$ to establish $\CostToCome(\Vecw) + \CostFcn(\Vecw, \Vecy_{m-1}) \leq \CostToCome(\Vecy_{m-1})$.

	
	Thus by induction \cref{eq: fmt induction} holds for all $m$. Taking $m = K$, we finally have that $\CostFcn_\nSamples \leq \CostToCome(\Vecy_K) \leq \sum_{k=0}^{K-1}\CostFcn(\Vecy_{k}, \Vecy_{k+1}) \leq (1+\epsilon)\CostFcn\supopt$, as desired.
\end{proof}
\begin{remark}[Asymptotic Optimality of \FMTstar]
	If the planning problem at hand admits an optimal solution that does not itself have strong $\delta$-clearance, but is arbitrarily approximable both pointwise and in cost by trajectories with strong clearance (see \cite{ES-LJ-MP:15a} for additional discussion on why such an assumption is reasonable), then Theorem~\ref{th: fmt asymptotic optimality} implies the asymptotic optimality of \FMTstar.
\end{remark}